\documentclass[reprint, amsmath, amssymb, superscriptaddress, nofootinbib, onecolumn]{revtex4-2}

\usepackage{graphicx}
\usepackage{natbib}
\usepackage{amsmath,amssymb,amsthm,mathtools}
\usepackage{physics}
\usepackage{caption}
\usepackage{float}
\usepackage{dcolumn}
\usepackage{bm}
\usepackage[colorlinks=true,citecolor=blue,linkcolor=blue,urlcolor=blue]
{hyperref}
\usepackage{subfig}
\usepackage{xcolor}
\usepackage{booktabs}
\usepackage{braket}
\usepackage{tikz}
\usepackage{quantikz}
\usepackage{tabularx}
\newtheorem{proposition}{Proposition}
\newtheorem{lemma}{Lemma}
\newtheorem{theorem}{Theorem}
\newtheorem{corollary}{Corollary}
\theoremstyle{remark}
\newtheorem{remark}{Remark}

\newenvironment{rules}%
{\begin{center}\begin{minipage}{0.94\textwidth}\vspace{2pt}\hrule\vspace{5pt}}%
		{\vspace{5pt}\hrule\vspace{2pt}\end{minipage}\end{center}}

\newcommand{\figabsm}{\begin{quantikz}[column sep=2.20mm, row sep=3.0mm]
\lstick{$n_0$} & \qw & \qw & \qw & \ctrl{5} & \qw & \qw & \qw & \qw & \qw & \qw \\
\lstick{$n_1$} & \qw & \qw & \qw & \qw & \qw & \ctrl{5} & \qw & \qw & \qw & \qw \\
\lstick{$n_2$} & \ctrl{1} & \qw & \qw & \qw & \qw & \qw & \qw & \ctrl{1} & \qw & \qw \\
\lstick{$0\ \text{(scratch)}$} & \targ{} & \gate{X} & \qw & \qw & \ctrl{2} & \qw & \ctrl{3} & \targ{} & \gate{X} & \qw \\
\lstick{$0$} & \qw & \qw & \gate{X} & \qw & \qw & \qw & \qw & \qw & \qw & \qw \\
\lstick{$0$} & \qw & \qw & \qw & \targ{} & \targ{} & \qw & \qw & \qw & \qw & \qw \\
\lstick{$0$} & \qw & \qw & \qw & \qw & \qw & \targ{} & \targ{} & \qw & \qw & \qw \\
\end{quantikz}}

\newcommand{\figadder}{\begin{quantikz}[column sep=2.20mm, row sep=3.0mm]
\lstick{$a_0$} & \ctrl{2} & \ctrl{4} & \targ{} & \qw & \targ{} & \ctrl{3} & \ctrl{3} & \targ{} & \ctrl{3} & \targ{} & \ctrl{4} & \qw & \qw \\
\lstick{$a_1$} & \qw & \qw & \qw & \ctrl{2} & \ctrl{-1} & \targ{} & \targ{} & \ctrl{-1} & \qw & \qw & \qw & \qw & \qw \\
\lstick{$b_0$} & \targ{} & \qw & \ctrl{-2} & \qw & \qw & \qw & \qw & \qw & \qw & \ctrl{-2} & \qw & \targ{} & \qw \\
\lstick{$b_1$} & \qw & \qw & \qw & \targ{} & \qw & \ctrl{-2} & \ctrl{-2} & \qw & \targ{} & \qw & \qw & \qw & \qw \\
\lstick{$0\ \text{(carry)}$} & \qw & \targ{} & \ctrl{-2} & \qw & \qw & \qw & \qw & \qw & \qw & \ctrl{-2} & \targ{} & \ctrl{-2} & \qw \\
\end{quantikz}}

\newcommand{\figinject}{\begin{quantikz}[column sep=2.20mm, row sep=3.0mm]
\lstick{$V_0$} & \ctrl{3} & \ctrl{6} & \targ{} & \qw & \targ{} & \ctrl{4} & \qw & \qw & \qw & \qw & \qw & \qw & \ctrl{4} & \targ{} & \ctrl{4} & \targ{} & \ctrl{6} & \qw & \qw \\
\lstick{$V_1$} & \qw & \qw & \qw & \ctrl{3} & \ctrl{-1} & \targ{} & \qw & \targ{} & \ctrl{4} & \ctrl{4} & \targ{} & \ctrl{4} & \targ{} & \ctrl{-1} & \qw & \qw & \qw & \qw & \qw \\
\lstick{$0$} & \qw & \qw & \qw & \qw & \qw & \qw & \ctrl{3} & \ctrl{-1} & \targ{} & \targ{} & \ctrl{-1} & \qw & \qw & \qw & \qw & \qw & \qw & \qw & \qw \\
\lstick{$\gamma_0$} & \targ{} & \qw & \ctrl{-3} & \qw & \qw & \qw & \qw & \qw & \qw & \qw & \qw & \qw & \qw & \qw & \qw & \ctrl{-3} & \qw & \targ{} & \qw \\
\lstick{$\gamma_1$} & \qw & \qw & \qw & \targ{} & \qw & \ctrl{-3} & \qw & \qw & \qw & \qw & \qw & \qw & \ctrl{-3} & \qw & \targ{} & \qw & \qw & \qw & \qw \\
\lstick{$\gamma_2$} & \qw & \qw & \qw & \qw & \qw & \qw & \targ{} & \qw & \ctrl{-3} & \ctrl{-3} & \qw & \targ{} & \qw & \qw & \qw & \qw & \qw & \qw & \qw \\
 & \qw & \targ{} & \ctrl{-3} & \qw & \qw & \qw & \qw & \qw & \qw & \qw & \qw & \qw & \qw & \qw & \qw & \ctrl{-3} & \targ{} & \ctrl{-3} & \qw \\
\end{quantikz}}

\newcommand{\figweightthree}{\begin{quantikz}[column sep=2.20mm, row sep=3.0mm]
\lstick{$q_{1,1}$} & \qw & \qw & \qw & \qw & \qw & \qw & \ctrl{1} & \qw & \qw & \qw & \qw & \qw & \qw \\
\lstick{$F$} & \gate[6]{R_{ZZ}} & \gate[5]{R_{ZZ}} & \gate[4]{R_{ZZ}} & \gate[3]{R_{ZZ}} & \gate{Z} & \qw & \targ{} & \gate[6]{R_{ZZ}} & \gate[5]{R_{ZZ}} & \gate[4]{R_{ZZ}} & \gate{Z} & \qw & \qw \\
\lstick{$y_0$} & \linethrough & \linethrough & \linethrough & \linethrough & \qw & \gate{Z} & \qw & \linethrough & \linethrough & \linethrough & \qw & \qw & \qw \\
\lstick{$y_1$} & \linethrough & \linethrough & \linethrough &  & \qw & \qw & \qw & \linethrough & \linethrough & \linethrough & \qw & \gate{Z} & \qw \\
\lstick{$y_2$} & \linethrough & \linethrough &  & \qw & \qw & \qw & \qw & \linethrough & \linethrough &  & \qw & \qw & \qw \\
\lstick{$y_3$} & \linethrough &  & \qw & \qw & \qw & \qw & \qw & \linethrough &  & \qw & \qw & \qw & \qw \\
\lstick{$y_4$} &  & \qw & \qw & \qw & \qw & \qw & \qw &  & \qw & \qw & \qw & \qw & \qw \\
\end{quantikz}}

\newcommand{\figframes}{\begin{quantikz}[column sep=2.20mm, row sep=3.0mm]
\lstick{$q_{1,1}$} & \qw & \qw & \qw & \qw & \qw & \qw & \qw & \ctrl{1} & \ctrl{6} & \ctrl{5} & \ctrl{4} & \qw & \qw & \qw & \qw & \qw & \qw \\
\lstick{$F$} & \qw & \qw & \qw & \qw & \gate{Z} & \qw & \ctrl{2} & \targ{} & \qw & \qw & \qw & \qw & \qw & \qw & \gate{Z} & \qw & \qw \\
\lstick{$y_0$} & \qw & \qw & \qw & \qw & \qw & \gate{Z} & \qw & \qw & \qw & \qw & \qw & \qw & \qw & \qw & \qw & \qw & \qw \\
\lstick{$y_1$} & \qw & \qw & \qw & \gate{R_z} & \qw & \qw & \targ{} & \qw & \qw & \qw & \qw & \qw & \qw & \qw & \qw & \gate{Z} & \qw \\
\lstick{$y_2$} & \qw & \qw & \gate{R_z} & \qw & \qw & \qw & \qw & \qw & \qw & \qw & \targ{} & \qw & \qw & \gate{R_z} & \qw & \qw & \qw \\
\lstick{$y_3$} & \qw & \gate{R_z} & \qw & \qw & \qw & \qw & \qw & \qw & \qw & \targ{} & \qw & \qw & \gate{R_z} & \qw & \qw & \qw & \qw \\
\lstick{$y_4$} & \gate{R_z} & \qw & \qw & \qw & \qw & \qw & \qw & \qw & \targ{} & \qw & \qw & \gate{R_z} & \qw & \qw & \qw & \qw & \qw \\
\end{quantikz}}

\newcommand{\figcompetitor}{\begin{quantikz}[column sep=2.20mm, row sep=3.0mm]
\lstick{$n_{1,0}$} & \ctrl{1} & \qw & \qw & \qw & \qw & \qw & \qw & \qw & \qw & \qw & \qw & \qw & \qw & \qw \\
\lstick{$n_{1,1}$} & \targ{} & \qw & \qw & \qw & \qw & \ctrl{5} & \qw & \qw & \qw & \qw & \qw & \ctrl{5} & \qw & \qw \\
\lstick{$n_{3,0}$} & \qw & \ctrl{1} & \qw & \qw & \qw & \qw & \qw & \qw & \qw & \qw & \qw & \qw & \qw & \qw \\
\lstick{$n_{3,1}$} & \qw & \targ{} & \qw & \qw & \qw & \qw & \qw & \qw & \qw & \ctrl{3} & \qw & \qw & \qw & \qw \\
\lstick{$n_{5,0}$} & \qw & \qw & \ctrl{1} & \qw & \qw & \qw & \qw & \qw & \qw & \qw & \qw & \qw & \qw & \qw \\
\lstick{$n_{5,1}$} & \qw & \qw & \targ{} & \ctrl{1} & \qw & \qw & \qw & \ctrl{1} & \qw & \qw & \qw & \qw & \qw & \qw \\
\lstick{$0\ (T)$} & \qw & \qw & \qw & \targ{} & \gate{R_z} & \targ{} & \gate{R_z} & \targ{} & \gate{R_z} & \targ{} & \gate{R_z} & \targ{} & \gate{R_z} & \qw \\
\end{quantikz}}

\newcommand{\figftskeleton}{
\begin{quantikz}[column sep=2.2mm, row sep=3.2mm]
\lstick{$n$ (system)} & \qw & \gate[2]{|m|} & \gate[3]{\times_e} & \gate[3]{\Sigma^{\pm}_e}
  & \qw & \qw & \qw & \gate[3]{\Sigma^{\pm\dagger}_e} & \gate[3]{\times_e^{\dagger}}
  & \gate[2]{|m|^{\dagger}} & \rstick{$n$} \qw \\
\lstick{$0$ ($|m|$)} & \qw & \qw & \qw & \qw & \qw & \qw & \qw & \qw & \qw & \qw
  & \rstick{$0$} \qw \\
\lstick{$0$ ($P_{pq}$, $C_e$, $C_e^2$)} & \qw & \qw & \qw & \qw & \gate{\mathrm{sq}_e}
  & \gate[2]{+} & \gate{\mathrm{sq}_e^{\dagger}} & \qw & \qw & \qw & \rstick{$0$} \qw \\
\lstick{$\ket{\gamma_L}$} & \qw & \qw & \qw & \qw & \qw & \qw & \qw & \qw & \qw & \qw
  & \rstick{$e^{-i\theta V}\ket{\gamma_L}$} \qw
\end{quantikz}
}

\newcommand{\fignisqgadget}{
\begin{quantikz}[column sep=2.4mm, row sep=3.4mm]
\lstick{$q$ (system, $\Delta$ basis)} & \qw & \qw & \gate[3]{\mathcal{A}_{+}} & \qw & \qw
  & \qw & \gate[3]{\mathcal{A}_{-}} & \qw & \rstick{$q$} \qw \\
\lstick{$0$ (fold $F$)} & \qw & \qw & \qw & \qw & \qw & \qw & \qw & \qw & \rstick{$0$} \qw \\
\lstick{$\ket{0}^{\otimes w}$} & \qw & \gate{H^{\otimes w}} & \qw
  & \gate{\mathrm{QFT}_{\rm ns}^\dagger} & \gate{e^{-i\theta 4^{r} v^2}} & \gate{\mathrm{QFT}_{\rm ns}}
  & \qw & \gate{H^{\otimes w}} & \rstick{$\ket{0}^{\otimes w}$} \qw
\end{quantikz}
}

\newcommand{\figbaselinewfour}{
\begin{quantikz}[column sep=2.4mm, row sep=3.0mm]
\lstick{$n_{a,j}$} & \ctrl{1} & \qw & \qw & \qw & \qw & \qw & \ctrl{1} & \qw \\
\lstick{$n_{b,k}$} & \targ{} & \ctrl{1} & \qw & \qw & \qw & \ctrl{1} & \targ{} & \qw \\
\lstick{$n_{c,l}$} & \qw & \targ{} & \ctrl{1} & \qw & \ctrl{1} & \targ{} & \qw & \qw \\
\lstick{$n_{d,m}$} & \qw & \qw & \targ{} & \gate{R_z(\theta_{jklm})} & \targ{}
  & \qw & \qw & \qw
\end{quantikz}
}

\newcommand{\figcentredf}{
\begin{quantikz}[column sep=3.2mm, row sep=3.2mm]
\lstick{$n_0$} & \gate{P(\phi_0)} & \gate[3]{\mathrm{QFT}^\dagger}
  & \gate{P(\phi_0)} & \qw \\
\lstick{$n_1$} & \gate{P(\phi_1)} & \qw & \gate{P(\phi_1)} & \qw \\
\lstick{$n_2$} & \gate{P(\phi_2)} & \qw & \gate{P(\phi_2)} & \qw
\end{quantikz}
}

\begin{document}
	
	\title{Sum-of-squares compilation of matrix-model dynamics\\
		for NISQ and fault-tolerant hardware}
	
	
	\author{Ayanendu Dutta}
	\email[Corresponding author: ]{ayanendudutta@gmail.com}
	\affiliation{Department of Physics, Jadavpur University, Kolkata-700032, India}
	\affiliation{Department of Physics, Presidency University, Kolkata-700073, India}
	
	\author{Nabanita Sarkar}
	\email{nabanitas.official@gmail.com}
	\affiliation{Department of Mathematics, Jadavpur University, Kolkata-700032, India}
	
	
	
	\date{\today}
	
	\begin{abstract}
		The $\mathrm{SU}(N)$ matrix-model potential
		$V=-\tfrac{g^2}{4}\sum_{I,J}\mathrm{Tr}[X_I,X_J]^2$, the interaction of the
		BFSS/BMN family, whose thermal large-$N$ states are dual to black holes, is
		diagonal in the coordinate basis and is compiled in the literature as a phase
		polynomial with $O(d^2N^4Q^4)$ rotations per Trotter step. That cost is close to
		optimal for its gate class: we show that for a generic Trotter step the Walsh
		support $m_{\min}$ lower-bounds both the rotation count and the two-qubit count
		of every circuit over $\{\mathrm{CNOT},\mathrm{CZ},\mathrm{CPhase},R_{ZZ},
		\mathrm{SWAP},X,R_z\}$, ancillas included, and we build a parity network
		attaining $1.01$--$1.10\,m_{\min}$. Improvement therefore requires a gate outside
		that class, and the exact factorisation of $V$ into a sum of squared commutator
		entries supplies one in each era. For NISQ, Hadamard: each squared entry is
		added into a Fourier register and squared there, $\Theta(d^2N^3Q^3)$ two-qubit
		gates. For fault tolerance, $T$: reversible arithmetic with one phase-gradient
		injection per squared entry, $\Theta(d^2N^3Q^2)$ T gates and no arbitrary-angle
		rotation. Both are built and verified. With a native $R_{ZZ}$ the NISQ circuit
		uses fewer two-qubit gates than \emph{every} parity network from $Q=6$ at
		$\mathrm{SU}(2)$ and $Q=2$ at $\mathrm{SU}(5)$; the oracle beats a
		precision-matched Hamming-weight-phasing incumbent at every point built except
		the smallest, by up to $11.9\times$ in T count and, at equal width, $25.5\times$
		in surface-code spacetime. An ablation with the same
		arithmetic but no factorisation costs up to $21\times$ more and loses to that
		incumbent at $15$ of $17$ points, so the factorisation, not the arithmetic,
		produces the win. Both constructions survive the supersymmetric deformation to
		mini-BMN with fermions, and complete Trotter evolutions on a statevector
		simulator reproduce the exact product formula to $2\times10^{-12}$.
	\end{abstract}
	\maketitle

\tableofcontents

\section{Introduction}
\label{sec:intro}

The $\mathrm{SU}(N)$ matrix models obtained by reducing maximally supersymmetric
Yang--Mills theory to one dimension, i.e., the BFSS model~\cite{bfss} and its mass
deformation, the BMN model~\cite{bmn}, describe $d$ Hermitian $N\times N$
matrices coupled through the quartic commutator square $-\mathrm{Tr}\,[X_I,X_J]^2$.
Their interest is that they are ordinary quantum mechanics and, at the same time, a
non-perturbative definition of a gravitational system, because at strong coupling
and large $N$ the thermal states of the D0-brane system are dual to a black
zero-brane geometry of type IIA supergravity~\cite{imsy}. Questions that a
gravitational description cannot settle on its own --- how infalling matter is
scrambled, how a horizon thermalises, what survives at finite $N$ where the geometry
is only approximate --- become questions about the real-time evolution of a
Hamiltonian with finitely many bosonic degrees of freedom. The same commutator
square, in the same coordinate variables, also carries the orbifold-lattice
formulation of Yang--Mills theory~\cite{buser2021,bergner2024}, so the target is not
a single model but a family that reaches as far as QCD.

What has been established about that Hamiltonian is almost entirely thermal.
Euclidean lattice Monte Carlo reproduces the supergravity energy of the D0-brane
system, $\alpha'$ corrections included, at high
precision~\cite{hanada2014,berkowitz2016,mcsmc2023}, which is direct evidence that
the duality holds where it can be tested. Real time is where the picture stops. The
classical limit is chaotic, with a measured Lyapunov spectrum~\cite{gurari2016}, and
quantum corrections have been estimated inside a Gaussian truncation of the density
matrix~\cite{buividovich2019}, but neither is a controlled calculation of the quantum
dynamics, and Euclidean Monte Carlo cannot supply one because recovering real-time
correlators from imaginary-time data is an exponentially ill-conditioned inversion.
The observables that would actually probe the black-hole interpretation, among them
scrambling times, out-of-time-order correlators and the approach to equilibrium, are
therefore open, and a quantum computer is the natural instrument for
them~\cite{bauer2023}. Writing the truncated Hamiltonian on qubits is by now a solved
problem. What is not solved is the cost of the circuit that implements one step of
its time evolution, and that cost is highest in exactly the corner, large $N$,
fine grids, long evolutions, where the holographic physics lives. Gate count is
therefore the quantity that has to come down, with the two-qubit count setting the
accumulated error on NISQ hardware and the $T$ count the magic-state overhead under a
surface code.

Digital proposals exist. A framework for BFSS and BMN on a quantum computer was set
out in Ref.~\cite{gharibyan2021}, the orbifold lattice was introduced as a
formulation that truncates systematically and extends to
QCD~\cite{buser2021,bergner2024}, and small matrix models including mini-BMN have
been studied variationally and with deep
learning~\cite{hanhartnoll2020,rinaldi2022,asplund2016}. Resource counting, however,
is now dominated by the coordinate-basis (``universal framework'')
truncation~\cite{halimeh2026,hanada2025}, in which each boson carries $Q$ qubits
holding a point of a uniform grid and the potential becomes a sum of Pauli-$Z$
strings of weight at most four, in the same form for every $N$, $d$ and volume. Fermions have been added in the same language~\cite{halimeh2025ferm},
gauge invariance has been analysed~\cite{hanada2025gauge}, and Hamiltonians with
fewer terms have been proposed~\cite{bergner2026}. Those works also fix the cost,
since the potential term of one Suzuki--Trotter step grows as
$d(d-1)N^4Q^4$~\cite{halimeh2026}, later estimates retain that growth, and circuit
optimisation is named as the outstanding task~\cite{hanada2025gauge}. Hardware
agrees. The first digital simulation of a single-matrix $\mathrm{SU}(2)$ model on a
trapped-ion device found circuits of $\gtrsim10^3$ two-qubit gates noise-dominated
and restricted to the smallest truncation, identifying qubit count and circuit depth
rather than physics as what keeps holographic matrix models out of
reach~\cite{hartnett2026}.

This paper asks how far that cost can be pushed down, separately for the two
hardware eras, without changing the Hamiltonian simulated. The potential is diagonal
in the coordinate basis, so every compilation of it is a phase polynomial, and the
question is how cheaply one with $\Theta(d^2N^4Q^4)$ Walsh characters can be applied.
We measure ourselves against the two incumbents that own the field, namely
parity-network compilation on the NISQ side, taken at full strength rather than as
published, and the fault-tolerant estimate of Ref.~\cite{halimeh2026} in the
strongest built form we could construct.

Compiling circuits of this kind is a mature subject. For diagonal phases the
governing object is the circuit--polynomial correspondence and the parity networks it
licenses, developed through the Walsh description of diagonal
unitaries~\cite{welch2014}, {\sc GraySynth} and the CNOT-complexity of CNOT--phase
circuits~\cite{amy2018}, rule-based optimisation~\cite{nam2018} and
architecture-aware phase-polynomial synthesis~\cite{vandaele2022}. More broadly,
terms are grouped into commuting clusters and diagonalised
simultaneously~\cite{vandenberg2020,murairi2023}, resynthesised as phase or Pauli
gadgets in the ZX-calculus~\cite{cowtan2019,kissinger2020}, scheduled block-wise with
routing folded into synthesis~\cite{paulihedral2022,tket2021,sabre2019}, or searched
over by learned agents~\cite{alphatensor2025}. Orthogonal savings come from the
product formula itself~\cite{childs2021,tran2020,layden2022}, from working below the
circuit model~\cite{clinton2021} and from constructions of time-independent
depth~\cite{kokcu2022}, while fault-tolerant cost is lowered by cheaper primitives
such as the measurement-uncomputed logical AND~\cite{gidney2018}, Hamming-weight
phasing~\cite{kivlichan2020,campbell2021}, ancilla-free rotation
synthesis~\cite{rossselinger2016}, linear-$T$ table lookup~\cite{babbush2018} and
fermionic swap networks~\cite{kivlichan2018}. Lattice field theory has driven much of
this, from instruction-level estimates for Kogut--Susskind QCD~\cite{kannam2021} and
general algorithms for non-Abelian theories~\cite{kan2021} to near-optimal
techniques~\cite{rhodes2024}, explicit $\mathrm{SU}(3)$ circuits~\cite{balaji2025},
an arithmetic compilation of a quartic lattice potential~\cite{madhav2026} and the
question of how many gates a gauge theory needs~\cite{murairi2022}. All of these optimise \emph{within} a fixed gate class
over a fixed list of terms, which is also their limit, and it is the limit we first
quantify and then leave.

Throughout, the model is the bosonic potential~\eqref{eq:H} on the centred
grid~\eqref{eq:grid} with $Q$ qubits per boson, extended in Sec.~\ref{sec:bmn} to
mini-BMN, and counts refer to the potential term of one Trotter step unless stated
otherwise. The kinetic layer, and for mini-BMN the fermion layer, are the same
circuit in every compilation compared here, so they dilute ratios but cannot change
rankings. Two-qubit counts are logical, NISQ counts are given with a native $R_{ZZ}$
and with CNOT only, and fault-tolerant counts use the Gidney AND
model~\cite{gidney2018}.

Three results follow. The first is that the standard compilation is already optimal
within its own gate class. For a generic Trotter step the Walsh support $m_{\min}$
lower-bounds the rotation count of every $\{\mathrm{CNOT},X,R_z\}$ circuit and the
two-qubit count of every
$\{\mathrm{CNOT},\mathrm{CZ},\mathrm{CPhase},R_{ZZ},\mathrm{SWAP},X,R_z\}$ circuit,
ancillas included (Lemmas~\ref{lem:anc} and~\ref{lem:twoq}). We compute $m_{\min}$
exactly, give it in closed form for every $N$, and build a structure-aware parity
network that attains $1.009$--$1.10\,m_{\min}$, so the bound is a yardstick rather
than an abstraction: improvement requires leaving the gate class
(Sec.~\ref{sec:opt}).

The second is a pair of constructions that do leave it, one per era, built on the
exact factorisation of the potential into squared commutator entries. On NISQ
hardware the circuit adds each entry into a Fourier register and squares it there as
a quadratic form, which costs $12Q^3+O(Q^2)$ two-qubit gates at $\mathrm{SU}(2)$ in
both gate sets against $m_{\min}=\tfrac92Q^4+O(Q^3)$ (Sec.~\ref{sec:nisq}).
Fault-tolerantly the oracle evaluates each entry by reversible arithmetic and injects
$\mu_eC_e^2$ into a catalytic phase-gradient register, so that no arbitrary-angle
rotation appears anywhere and there is no synthesis error (Sec.~\ref{sec:ft}). Both
are built and verified against closed-form resource formulas, and the floor grows as
$\Theta(d^2N^4Q^4)$ \cite{halimeh2026} against $\Theta(d^2N^3Q^3)$ and $\Theta(d^2N^3Q^2)$ for ours
(Sec.~\ref{sec:scaling}).

The third is that the factorisation rather than the arithmetic is what buys this.
Applying the same arithmetic and pricing to the quartic monomials costs
$1.9$--$20.7\times$ more, and conversely the kinetic term, already optimal as a phase
polynomial, becomes $5.4$--$9.5\times$ more expensive compiled arithmetically
(Sec.~\ref{sec:ablation}). The mini-BMN potential is again a single sum of squares,
both constructions absorb the Myers and mass terms inside the same squares, and the
crossovers move by at most one step in $Q$ (Sec.~\ref{sec:bmn}).

\begin{table}[htbp]\centering\small
\setlength{\tabcolsep}{5pt}
\caption{Where each method wins, and by how much. $Q^\ast$ is the smallest number of qubits per boson, within the range built, at which our circuit is no more expensive than the incumbent; below $Q^\ast$ the incumbent is cheaper, and the entry $>Q$ means the crossover lies beyond the largest $Q$ built. NISQ columns compare against $m_{\min}$, the two-qubit floor of \emph{every} parity network (Lemma~\ref{lem:twoq}), in the native-$R_{ZZ}$ and CNOT-only gate sets; the all-to-all logical counts are exact for trapped-ion hardware. Heavy-hex (distance~7, 115 qubits) and square ($10\times12$) columns compare the native-$R_{ZZ}$ emission against the built structure-aware network after routing both with the same router (Sec.~\ref{sec:routing}); ``--'' marks points not routed. FT compares T gates against precision-matched Hamming-weight phasing (unbatched), the primary fault-tolerant incumbent of Sec.~\ref{sec:incumbents}. The last two columns give the ratio incumbent/ours \emph{for the potential term} at the largest $Q$ built, which is also the largest ratio measured; full-step ratios, which add the layers common to every compilation, are a few per cent to $16\%$ smaller and are in Tables~\ref{tab:ftsu2}--\ref{tab:ftmatrix}, \ref{tab:fullstep} and~\ref{tab:bmnstep}. $\mathrm{SU}(2)$ appears on both coordinate routes (Sec.~\ref{sec:matrixform}); mini-BMN is $\mathrm{SU}(2)$, $d=3$, $\kappa=1$ (Sec.~\ref{sec:bmn}), where the NISQ columns take the better of the affine and split circuits.}
\label{tab:crossover}
\begin{tabular}{llcccccrr}
\hline
model & route & \multicolumn{2}{c}{NISQ $Q^\ast$} & \multicolumn{2}{c}{routed $Q^\ast$ (RZZ)} & FT & NISQ ratio & FT ratio\\ & & RZZ & CNOT & heavy-hex & square & $Q^\ast$ & \multicolumn{2}{c}{potential term}\\ & & & & & & & (RZZ) & (HWP)\\
\hline
$\mathrm{SU}(2)$ & adjoint & 6 & 7 & 7 & 7 & 3 & 2.40 ($Q{=}10$) & 3.94 ($Q{=}6$)\\
$\mathrm{SU}(2)$ & matrix & $>4$ & $>4$ & -- & -- & 2 & 0.86 ($Q{=}4$) & 2.82 ($Q{=}4$)\\
$\mathrm{SU}(3)$ & matrix & 3 & 4 & 4 & 4 & 2 & 1.82 ($Q{=}4$) & 5.20 ($Q{=}4$)\\
$\mathrm{SU}(4)$ & matrix & 3 & 3 & 3 & 3 & 2 & 2.71 ($Q{=}4$) & 8.16 ($Q{=}4$)\\
$\mathrm{SU}(5)$ & matrix & 2 & 3 & 2 & 2 & 2 & 3.82 ($Q{=}4$) & 11.91 ($Q{=}4$)\\
mini-BMN & adjoint & 6 & 7 & -- & -- & 4 & 1.37 ($Q{=}7$) & 3.30 ($Q{=}6$)\\
\hline
\end{tabular}\end{table}

Table~\ref{tab:crossover} states where each method wins, and every crossover in it is
a lower bound on $N$ or on $Q$. That is the right way round, because holography is a
large-$N$ statement, large $N$ forces a large grid per
boson~\cite{hanada2025,hartnett2026}, and our advantage grows along both axes
(Sec.~\ref{sec:whichway}). 

Several ingredients are used and not claimed. The circuit--polynomial correspondence
and parity-network synthesis are taken from Ref.~\cite{amy2018} and the Walsh
description of diagonal unitaries from Ref.~\cite{welch2014}. Driving a phase by
kickback from reversible arithmetic is standard~\cite{kitaev2002}, including for
gauge theory~\cite{kan2021} and a quartic lattice potential~\cite{madhav2026}. We take Gidney's logical AND, the catalytic phase-gradient
register and Hamming-weight phasing~\cite{gidney2018,kivlichan2020,campbell2021}
together with the Cuccaro adder~\cite{cuccaro2004}. Steps 1--3 of our NISQ gadget are Fourier-basis
addition~\cite{draper2000,ruiz2017} as used to write a polynomial of binary variables
into a register by Hadamards, controlled phases and an inverse
QFT~\cite{gilliam2021,bucher2025}, while the simplifications S1--S3 rest on Draper
triangularity and parity-network bookkeeping. Sum-of-squares factorisation has been
exploited before in quantum chemistry~\cite{sa2025}. The formulation we
compile and the estimate we compare against are those of
Refs.~\cite{halimeh2026,hanada2025,halimeh2025ferm}, the Trotter error we quote is
the interference-aware leading order of Refs.~\cite{layden2022,tran2020}, and the
mini-BMN model and the fact that its bosonic potential is a total square are due to
Refs.~\cite{asplund2016} and~\cite{hanhartnoll2020}. The fermion layer uses
Jordan--Wigner and weighted Hamming-weight phasing, and the surface-code numbers the
estimator of Ref.~\cite{beverland2022}. The bookkeeping F1--F3 and the kinetic rewrites K1--K4 are
standard practice, and Sec.~\ref{sec:aer} validates the circuits.

What is new here, to the best of our knowledge, is the following. The two-qubit floor
of Lemma~\ref{lem:twoq} is elementary within the framework of Ref.~\cite{amy2018},
but the exactly computed supports and the closed forms~\eqref{eq:mminsu2}
and~\eqref{eq:mminN} are not available elsewhere. The exact commutator sum of squares
has not been used as the compilation unit for matrix models before, nor realised
either by Fourier-register squaring or by integer arithmetic with one phase-gradient
injection per squared entry, and we show that the first beats every parity network
and the second a precision-matched Hamming-weight-phasing incumbent in the regimes
stated. The closed-form resource formulas are verified gate for gate against the
built circuits, and the ablation isolates the factorisation rather than the
arithmetic as the source of the fault-tolerant win. Carrying the mini-BMN total
square as the compilation unit in both eras is likewise new, as is the first
gate-level, verified Trotter step of mini-BMN with fermions in the coordinate basis,
priced for NISQ, fault-tolerant and surface-code hardware. The ancilla-invariance of
Lemma~\ref{lem:anc} follows from the correspondence and is not claimed separately.

Sections~\ref{sec:model}--\ref{sec:opt} set up the factorisation, the cost model and
the bound, Secs.~\ref{sec:nisq}--\ref{sec:ft} give the two constructions, and
Secs.~\ref{sec:scaling}--\ref{sec:ablation} isolate why they work. The rest is
evidence that the numbers are real: the full Trotter step and a surface-code
estimate (Sec.~\ref{sec:kinetic}), complete statevector time evolutions
(Sec.~\ref{sec:aer}) and the supersymmetric extension (Sec.~\ref{sec:bmn}). The appendices hold secondary
baselines, routed counts, transcribed circuits, the verification protocol and the
negative results, and the code package~\cite{codepackage} reproduces every number and
figure.

\section{Model, encoding and the sum-of-squares identity}
\label{sec:model}

\subsection{Hamiltonian and coordinates}
For $d$ bosonic $N\times N$ Hermitian matrices $X_I$,
\begin{equation}
  H=\tfrac12\sum_I\mathrm{Tr}\,P_I^2+V,\qquad
  V=-\frac{g^2}{4}\sum_{I,J}\mathrm{Tr}\,[X_I,X_J]^2 .
  \label{eq:H}
\end{equation}
\emph{Adjoint coordinates}: $X_I=\sum_aX_I^a\tau^a$ with
$\mathrm{Tr}(\tau^a\tau^b)=\delta^{ab}$, $[\tau^a,\tau^b]=if^{abc}\tau^c$,
giving $d(N^2-1)$ bosons. \emph{Matrix-entry coordinates}: $X_I=S_I+iA_I$ with
$S_I$ real symmetric and $A_I$ real antisymmetric, giving $dN^2$ bosons, the
``without traceless condition'' form of Ref.~\cite{halimeh2026}. The trace part
is free, decoupled, and absent from every commutator.

Each boson carries $Q$ qubits holding $n\in[0,\Lambda)$, $\Lambda=2^Q$, with
\begin{equation}
  \hat x=\delta\Big(n-\tfrac{\Lambda-1}{2}\Big)=\tfrac{\delta}{2}m,
  \qquad m\equiv2n-(\Lambda-1),
  \label{eq:grid}
\end{equation}
so $m$ is an \emph{odd} integer in $[-(\Lambda-1),\Lambda-1]$. Writing
$n=\sum_j2^{\,j}b_j$ gives $m=\sum_j2^{\,j}(2b_j-1)$, that is, as an operator,
\begin{equation}
  \boxed{\;m=-\sum_{j=0}^{Q-1}2^{\,j}Z_j\;}
  \label{eq:mZ}
\end{equation}
exactly, with no constant term: every coordinate is a pure weight-one Walsh
operator, and every bilinear $m_pm_q$ of distinct bosons is a pure weight-two
operator with $Q^2$ terms of coefficient $2^{\,j+k}$. Everything below rests on
this identity, the lower bound of Sec.~\ref{sec:opt} because it fixes the
Walsh support exactly, and both circuits because it makes the potential an exact
integer.

\subsection{Absolute value and sign are free}
\label{sec:absm}
Since $m=2n-(2^Q-1)$, $m>0$ iff $\mathrm{MSB}(n)=1$: the sign is one qubit of
the register. For the magnitude set $n'=n$ if $\mathrm{MSB}(n)=1$ and
$n'=\lnot n=\Lambda-1-n$ otherwise; in both cases $|m|=2n'-\Lambda+1$, and with
$n'=2^{Q-1}+n''$,
\begin{equation}
  |m|=2n''+1,
\end{equation}
whose binary form is $[1,n'_0,\dots,n'_{Q-2}]$. The circuit
(Fig.~\ref{fig:absm}) uses \textbf{zero AND gates}: $2Q$ CNOTs and $Q$ $X$ gates
per boson, all Clifford. Every multiplication in the fault-tolerant oracle of
Sec.~\ref{sec:ft} is therefore \emph{unsigned}; the sign travels as one control
bit and is applied once, to the accumulation, by a conditional negation. This is
not a micro-optimisation: a signed multiplier would carry the sign through every
partial product, and the potential has $\Theta(N^3)$ products per matrix pair.

\subsection{The potential is an exact sum of squares}
\label{sec:identities}
For Hermitian $X_I$, $-\mathrm{Tr}[X_I,X_J]^2=\|[X_I,X_J]\|_F^2\ge0$, hence
\begin{equation}
  V=\frac{g^2}{2}\sum_{I<J}\big\|[X_I,X_J]\big\|_F^2 .
  \label{eq:sos}
\end{equation}
This is the compilation unit of both methods: instead of the $\Theta(N^4)$
quartic monomials of the expanded potential, the object to evaluate is the
$\Theta(N^2)$ commutator entries, each a bilinear with $\Theta(N)$ terms, and
each squared exactly once.

\paragraph{$\mathrm{SU}(2)$.} With $\tau^a=\sigma^a/\sqrt2$,
$f^{abc}=\sqrt2\,\varepsilon^{abc}$, so that $V=g^2\sum_{I<J}|\vec v_I\times\vec
v_J|^2$. By the Lagrange identity
\begin{equation}
  V=g^2\sum_{I<J}\big(R_IR_J-S_{IJ}^2\big),\qquad
  R_I=\sum_a(X_I^a)^2,\quad S_{IJ}=\sum_aX_I^aX_J^a ,
  \label{eq:gram}
\end{equation}
and componentwise
\begin{equation}
  V=g^2\sum_{I<J}\sum_{c=1}^{3}\big(K^c_{IJ}\big)^2,\qquad
  K^c_{IJ}=\varepsilon^{abc}X_I^aX_J^b .
  \label{eq:cross}
\end{equation}
Each $K^c$ is a weight-two operator with $2Q^2$ terms. Both methods compile the
cross-product form~\eqref{eq:cross}, which is the cheaper of the two in both
eras; the Gram form~\eqref{eq:gram} is used as the independent reference
Hamiltonian in Sec.~\ref{sec:aer}.

\paragraph{Matrix-entry form.}
\label{sec:matrixform}
With $X_I=S_I+iA_I$,
$[X_I,X_J]=([S_I,S_J]-[A_I,A_J])+i([S_I,A_J]+[A_I,S_J])$, and because
$(S_IS_J)^\top=S_JS_I$, $(A_IA_J)^\top=A_JA_I$, $(S_IA_J)^\top=-A_JS_I$, each
commutator needs one real matrix product, four per pair; every coefficient is
$\pm1$.

\begin{corollary}[exact integer potential]
\label{cor:integer}
In matrix-entry coordinates on the grid~\eqref{eq:grid},
$V_{\rm int}=\sum_{I<J}\sum_{ij}|[X_I,X_J]_{ij}|^2$ in $m$-units is a
non-negative integer, and $V=\tfrac{g^2}{2}(\tfrac{\delta}{2})^4V_{\rm int}$.
\end{corollary}

\begin{rules}
\noindent\textbf{Which coordinates are used, and why.} The analogue of
Corollary~\ref{cor:integer} is \emph{false} in adjoint coordinates for $N\ge3$:
the $\mathrm{SU}(3)$ structure constants include $f^{458}=f^{678}=\sqrt3/2$, so
the quartic coefficients of $V$ include $\sqrt3/2$ and $\sqrt3$, and for
$\mathrm{SU}(4)$, $15$ of $21$ distinct coefficients are non-dyadic.
Integrality is what puts the fault-tolerant phase on an exact dyadic ladder
(Sec.~\ref{sec:precision}), so \textbf{every $N\ge3$ circuit in this paper uses
matrix-entry coordinates}, while $\mathrm{SU}(2)$, where both routes are
integral, is reported on the adjoint (cross-product) route, which is the cheaper
of the two, with the matrix route given alongside wherever both were built.
$\mathrm{SU}(2)$ therefore carries two sets of numbers and they differ: at $Q=2$
the fault-tolerant oracle loses to the precision-matched incumbent on the
adjoint route ($0.77\times$) and wins on the matrix route ($1.13\times$). The
cheaper route giving the worse ratio is not a contradiction: the adjoint route is
cheaper in absolute cost for \emph{both} sides, and it is the incumbent that
gains more. Passing from matrix entries to adjoint coordinates at
$\mathrm{SU}(2)$, $Q=2$ takes the Walsh support, which sets the incumbent's cost,
from $186$ to $60$ characters --- a factor $3.1$ --- while our product count,
which sets ours, falls only from $12$ to $6$; our oracle drops from $2952$ to
$1500$ T and the incumbent from $3336$ to $1156$. A ratio is not a cost, and the
route that minimises our absolute cost is the one we report. Every table caption
below names the route it uses.
\end{rules}

Finally, since $\mathrm{Re}[X_I,X_J]$ is antisymmetric and
$\mathrm{Im}[X_I,X_J]$ symmetric,
\begin{equation}
  \|[X_I,X_J]\|_F^2=2\sum_{i<j}\mathrm{Re}_{ij}^2+2\sum_{i<j}\mathrm{Im}_{ij}^2
  +\sum_i\mathrm{Im}_{ii}^2 ,
\end{equation}
that is $N^2$ weighted squares per pair, with
\begin{equation}
   \#\mathrm{Re}_{i<j}=\#\mathrm{Im}_{i<j}=4(N-1),\qquad \#\mathrm{Im}_{ii}=2(N-1)
  \label{eq:nterms}
\end{equation}
candidate products per entry after $A_{ii}=0$ and after merging products that share a boson pair (the diagonal entries merge in pairs, which is why they carry half as many). Products sharing the same boson
pair are merged; both counts are checked against the circuits' own term lists.
All identities of this section are verified to machine precision, and the
integer oracle reproduces $-\tfrac{g^2}{4}\sum\mathrm{Tr}[X_I,X_J]^2$ from
Hermitian matrices to $5.6\times10^{-16}$ (Appendix~\ref{sec:verif}).

\section{Cost model, incumbents and precision}
\label{sec:cost}

Before any circuit is built, the rules of the comparison have to be fixed,
because a resource claim is only as strong as the incumbent it is measured
against. This section names \emph{one primary incumbent per era}; every headline
claim in the paper is made against that one, and the weaker baselines are
reported for completeness in Appendix~\ref{app:secondary}.

\begin{rules}
\noindent\textbf{Comparison rules.} (i)~Every incumbent compiles the \emph{same}
exact Walsh expansion of the same integer potential $V_{\rm int}$, on the same
grid and at the same truncation; no compilation here changes a term of the
Hamiltonian. (ii)~Every comparison is made \emph{within one gate set}: NISQ
counts against NISQ counts in the same gate set, T counts against T counts at
the same precision. (iii)~\textbf{Primary NISQ incumbent}: the proved two-qubit
floor $m_{\min}$ of every parity network (Lemma~\ref{lem:twoq}), together with a
built structure-aware network that nearly attains it (Sec.~\ref{sec:competitor}).
\textbf{Primary fault-tolerant incumbent}: Hamming-weight phasing of the
equal-angle classes with every rotation priced by its own dyadic angle
(Sec.~\ref{sec:incumbents}), built and verified. (iv)~Layers common to all
compilations, kinetic, and for mini-BMN fermionic, are the same circuit
everywhere, and are reported separately as well as inside full steps.
\end{rules}

\subsection{NISQ metrics}
\label{sec:nisqmetrics}
We count two-qubit gates and two-qubit depth of the emitted gate sequence (ASAP
schedule, commuting gates not reordered), for two gate sets: (i) native
arbitrary-angle $R_{ZZ}$ and CPhase, each one two-qubit gate (available on
hardware with a native $e^{-i\theta ZZ}$ interaction, e.g.\ the trapped-ion
devices noted in Ref.~\cite{halimeh2026}); (ii) CNOT only, with arbitrary
single-qubit rotations. Counts are logical unless stated otherwise. On an
all-to-all device, such as the trapped-ion hardware with native $ZZ$ that is also the
natural native-$R_{ZZ}$ platform, no routing is needed, so the logical counts
\emph{are} the routed counts there; heavy-hex and square-lattice counts are in
Sec.~\ref{sec:routing}.

The like-for-like rule matters here in a way that is easy to get wrong.
Reference~\cite{halimeh2026} compiles every Pauli-$Z$ string as a CNOT ladder,
one $R_z$ and the reverse ladder, and counts CNOTs as the NISQ metric, i.e.\
gate set (ii); it notes, but does not use, native arbitrary-angle $ZZ$ gates.
Comparing a native-$R_{ZZ}$ count against a CNOT count would flatter our circuit
by up to a factor of two. Here every comparison is made within one gate set, and
in each gate set our circuit is emitted in the better of two forms
(Sec.~\ref{sec:nisq}). The floor of Lemma~\ref{lem:twoq} covers both gate sets,
so it is the right yardstick for either.

\subsection{Fault-tolerant metrics and primitive costs}
\label{sec:costmodel}
We count T gates in the Gidney model~\cite{gidney2018}: a logical AND costs $4$
T to compute and $0$ to uncompute, and an $n$-bit ripple addition uses $n-1$
logical ANDs, so that
\begin{equation}
    T_{\rm add}(n)=T_{\rm incr}(n)=T_{\rm condneg}(n)=4(n-1) .
  \label{eq:Tprim}
\end{equation}
The circuits themselves use Cuccaro adders~\cite{cuccaro2004}
(Fig.~\ref{fig:adder}), which makes them exhaustively verifiable as classical
permutations; the Gidney count is charged on the cost-event log. The windowed
multiplier and the squarer, and the derivation of their costs
$T_{\rm mul}(w_a,w_b,w_o)$ and $T_{\rm sq}(w,w_o)$, are given in
Appendix~\ref{app:circuits}, Eqs.~\eqref{eq:Tmul} and~\eqref{eq:Tsq}. One
consequence deserves emphasis, because it is built into every formula below and
is the reason a compute--phase--uncompute oracle is not automatically cheap:
\emph{uncomputation is not free}. Only the ANDs inside a primitive uncompute for
free; running an adder backwards computes and uncomputes the same $n-1$ ANDs, so
its inverse costs the same $T_{\rm add}(n)$, and a compute--phase--uncompute
oracle pays a factor two that a phase polynomial does not. (Where a
\emph{standalone} two's-complement subtraction is needed, only in the
Gram-route oracle of Appendix~\ref{app:secondary}, it is built as a negation
followed by an addition and charged $8(n-1)$.)

\subsection{Precision and the dyadic ladder}
\label{sec:precision}
Our fault-tolerant oracle adds $A\cdot V_{\rm int}$ into an $L$-qubit catalytic
phase-gradient register, imprinting $e^{-2\pi iAV_{\rm int}/2^L}$ exactly. With
$A=2^k$ (each injection shifted by $k$, cost $\le T_{\rm add}(L)$) the Trotter
step is chosen on the dyadic ladder $\theta=2\pi\,2^{k-L}$, and the potential
phase carries \emph{no} synthesis error. A general $b$-bit $A$ costs one
addition per set bit and, because the injection is made per squared entry (F1,
Sec.~\ref{sec:ft}), per entry: at $\mathrm{SU}(2)$, $Q=2$, $A=5$, $11$ and $183$
raise the T-count from $1500$ to $1680$, $1860$ and $2292$ ($+53\%$ for six set
bits), with the payload $A\,V_{\rm int}\bmod 2^L$ verified on $60$ random basis
states each and the built counts equal to the closed form. This is why the step
is chosen on the ladder, and all tables use $A=1$.

At the same $\theta$ every incumbent rotation angle $\theta\hat\phi(y)$ is
dyadic too, and the same catalytic register applies it exactly by one
addition~\cite{gidney2018}. We therefore charge such a rotation by its own
angle, not by the register width: a phase $2\pi(\text{odd})/2^b$ costs
\begin{equation}
  T_d(b)=\begin{cases}0 & b\le2\quad (Z\text{ or }S,\ \text{Clifford}),\\
  1 & b=3\quad (T),\\ 4(b-1) & b\ge4\quad(\text{$b$-bit addition into }|\gamma_L\rangle),\end{cases}
  \label{eq:Td}
\end{equation}
and the number of significant bits of the incumbent's rotation on the character
$\chi_y$ is $b=L-1-\nu_2(\hat\phi(y))$: the factor $2^{\nu_2}$ of the
coefficient is Draper triangularity applied to the incumbent, and the extra $-1$
is phasing the parity \emph{bit} rather than the $\pm1$ character
($\chi=1-2b_{\rm par}$ doubles the angle). In Hamming-weight phasing the
rotation attached to weight bit $k$ gains a further $2^k$, i.e.\
$b=L-1-\nu_2-k$. Charging every incumbent rotation the flat
$T_{\rm add}(L)=4(L-1)$ instead would overcharge it by $30$--$40\%$ at
$\mathrm{SU}(2)$; that weaker price is reported in Table~\ref{tab:ftrobust}, and
every claim of a win is made against the valuation-aware one. Our own oracle is
unaffected: it has no rotation, and its injection of $\mu_eC_e^2$ costs
$T_{\rm add}(L-\nu_2(\mu_e))$ per entry.

For orientation, Ross--Selinger synthesis of a non-dyadic angle costs
$3\log_2(1/\varepsilon)+O(\log\log(1/\varepsilon))$ T~\cite{rossselinger2016}, so
the value $T_{\rm typ}=10$ quoted in Ref.~\cite{halimeh2026} corresponds to
$\varepsilon\approx10^{-1}$, and $T_{\rm rot}=50$, $100$ to
$\varepsilon\approx10^{-5}$, $10^{-10}$. This is one reason the published
estimate is treated as a secondary baseline: at a precision anyone would
actually use, it is several times more expensive than the incumbent our headline
numbers are measured against.

Because $A=1$ and $L>\log_2\max V_{\rm int}$, every $|\hat\phi(y)|\le\max
V_{\rm int}<2^{L}$, so no character's angle is a multiple of $2\pi$ on the
ladder: the only ladder effect is the valuation saving just described. This also
settles a tension a careless reading might find between Secs.~\ref{sec:opt}
and~\ref{sec:ft}. The lower bound of Sec.~\ref{sec:opt} is a statement about one
gate class at generic angle; the fault-tolerant method escapes it by using a
gate outside that class (Toffoli, hence $T$) and by living on the ladder, where
the phase is exact by construction rather than synthesised. No bound is
contradicted, and none of the tables compares a bound in one setting with a
circuit in another.

\subsection{The primary fault-tolerant incumbent}
\label{sec:incumbents}
All incumbents compile the \emph{same} exact Walsh expansion of $V_{\rm int}$
(integer coefficients). The primary one is \textbf{Hamming-weight phasing with
valuation-aware pricing}, $B_{\rm HWP}^{\rm match}$. Characters with equal
$|\hat\phi(y)|$ share one angle up to sign (a sign is a Clifford $X$ on the
parity copy). A class of $m$ characters costs $m-w(m)$ Toffolis, $w$ the binary
weight, and $\lfloor\log_2m\rfloor+1$
rotations~\cite{gidney2018,kivlichan2020,campbell2021}, with $m$ parity copies
and $m-w(m)$ clean ancillas; each rotation is priced by Eq.~\eqref{eq:Td}. It is
evaluated in two forms, and the two have distinct roles, which we fix here once
so that the denominator never shifts silently. \emph{Unbatched} Hamming-weight
phasing is the cheaper of the two in T gates, so it is the reference for every
\textbf{T-count} claim; it is also the more demanding in width, needing far more
ancillas than our oracle. \emph{Width-matched} Hamming-weight phasing, batches
restricted so that its ancilla demand does not exceed ours, costs more in T
but is the only fair comparison once width is priced, so it is the reference for
every \textbf{width and surface-code spacetime} claim. Both are reported in every
fault-tolerant table, and each claim below names which it uses. It is \emph{built}: parity copies by CNOT
fan-in, an $X$ for negative coefficients, a Hamming-weight tree of full adders
(carry $((a\oplus c)(b\oplus c))\oplus c$, one AND) and half adders, one phase
gate per output bit, and the exact inverse. It is verified by exact simulation
(system restored, ancillas clean, and the integer identity
$\sum_k2^{k+1}c\,h_k=\sum_c c\,m_c-(V_{\rm int}-c_0)$ exactly) at ten
configurations up to $\mathrm{SU}(4)$, and its T-count and width equal the
formula at all $17$ points, unbatched and width-matched. This is the strongest
incumbent we know how to construct, and it is the one every fault-tolerant claim
below is measured against.

Three secondary fault-tolerant baselines are carried for orientation: the
published estimate $B_{\rm pub}=T_{\rm typ}\,d(d-1)N^4Q^4$ of
Ref.~\cite{halimeh2026} at $T_{\rm typ}=10$; one rotation per Walsh character,
$B_{\rm W}$, priced by Eq.~\eqref{eq:Td}; and Hamming-weight phasing with the
flat price $4(L-1)$ per rotation. On the NISQ side the secondary baselines are
the published grouping estimate $(2+4/Q)\,m_{\min}$ CNOTs (and
$(1+4/Q)\,m_{\min}$ with a native $R_{ZZ}$, our evaluation) and GraySynth as
implemented in Qiskit~\cite{amy2018}. All of these are weaker than the primary
incumbents; they are collected in Appendix~\ref{app:secondary} and, where they
appear in the main text, they are labelled as secondary.

\section{The phase-polynomial lower bound}
\label{sec:opt}

\begin{proposition}[circuit--polynomial correspondence~\cite{amy2018,welch2014}]
\label{prop:cp}
A circuit over $\{\mathrm{CNOT},X,R_z\}$ on $n$ qubits implements
$|x\rangle\mapsto e^{if(x)}|Ax+c\rangle$ with
$f(x)=\sum_i\theta_i(-1)^{y_i\cdot x}$, one Walsh character per $R_z$. The
Walsh expansion of a real function on $\mathbb{F}_2^n$ is unique.
\end{proposition}

Phases are defined modulo $2\pi$, so uniqueness constrains the circuit only up
to integer-valued shifts $2\pi k(x)$, whose Walsh coefficients lie in
$2\pi\,2^{-n}\mathbb{Z}$. We therefore state the bounds for a \emph{generic}
Trotter step: $\phi=\theta P$ with $P$ an integer polynomial and $\theta$ outside
a countable set (e.g.\ $\theta/\pi\notin\mathbb{Q}$). Then
$\theta\hat P(y)\notin2\pi\,2^{-n}\mathbb{Z}$ for every $y$ in the support, and
no shift can remove a character.

\begin{lemma}[rotations; ancillas do not help]
\label{lem:anc}
Let $U=\mathrm{diag}(e^{i\theta P(x)})$ on $n$ qubits with $\theta$ generic, and
$C$ any $\{\mathrm{CNOT},X,R_z\}$ circuit on $n+k$ qubits implementing $U$ up to
global phase, ancillas in and out in $|0\rangle$. Then $C$ has at least
$m_{\min}=|\{y\ne0:\hat P(y)\ne0\}|$ rotations. The bound is attained.
\end{lemma}

\begin{lemma}[two-qubit floor]
\label{lem:twoq}
Let
$\mathcal G=\{\mathrm{CNOT},\mathrm{CZ},\mathrm{CPhase},R_{ZZ},\mathrm{SWAP},X,R_z\}$.
Under the same hypotheses, every $\mathcal G$ circuit implementing $U$, with
ancillas, contains at least $m_{\min}-m_1$ two-qubit gates, $m_1$ the number of
weight-one characters in the support.
\end{lemma}

\begin{proposition}[no weight-one characters]
\label{prop:even}
Let $P$ be a polynomial in the coordinates~\eqref{eq:grid} all of whose monomials
have even total degree.
Then every character in the Walsh support of $P$ has even weight; in particular
$m_1=0$ and the floor of Lemma~\ref{lem:twoq} is exactly $m_{\min}$.
\end{proposition}

\begin{remark}
	\label{rem:odd}
	Odd-degree monomials do produce odd-weight characters, the mini-BMN Myers term
	of Sec.~\ref{sec:bmn} produces weight-three ones, but they do not restore the
	weight-one escape when every such monomial involves at least two distinct bosons,
	since by Theorem~\ref{thm:lb} each character then names at least two of them. This
	is the case for mini-BMN, where $m_1=0$ still holds (Appendix~\ref{app:bmn}).
\end{remark}

\begin{theorem}[$\Omega(Q^\kappa)$ for polynomial potentials on the grid]
\label{thm:lb}
Let $P$ have degree $\kappa$ in the coordinates~\eqref{eq:grid}. A monomial in
$\kappa$ \emph{distinct} bosons with collected coefficient $c\ne0$ contributes
exactly $Q^\kappa$ non-zero Walsh coefficients, and these sets are pairwise
disjoint for different boson multisets. Hence
$m_{\min}\ge Q^\kappa\cdot\#\{\text{distinct-boson }\kappa\text{-tuples with
non-zero coefficient}\}$.
\end{theorem}

Proofs are in Appendix~\ref{app:proofs}. The mechanism in each case is the same
and is worth stating once: over $\{\mathrm{CNOT},X\}$ every wire carries an
affine form of the input bits, so the set of linear forms available to be phased
grows by at most one per two-qubit gate, and by~\eqref{eq:mZ} an even-degree
polynomial in the grid coordinates produces only even-weight characters, which
kills the weight-one escape. Proposition~\ref{prop:even} applies to the bosonic
potential~\eqref{eq:H} for any $N$, $d$
and either coordinate system, and is confirmed numerically at every point: the
supports contain weight-two and weight-four characters only.

For $\mathrm{SU}(2)$, $d=2$ the support is exactly the $3Q^4$ characters of the
three distinct quadruples, the $6\binom{Q}{2}^2$ weight-four characters of the
squares $m_{1a}^2m_{2b}^2$ ($a\ne b$), and the $6\binom Q2$ weight-two
characters $Z_{a,j}Z_{a,j'}$:
\begin{equation}
  m_{\min}^{\mathrm{SU}(2),\,d=2}=3Q^4+6\tbinom Q2^2+6\tbinom Q2
  =\tfrac92Q^4-3Q^3+\tfrac92Q^2-3Q,
  \label{eq:mminsu2}
\end{equation}
which reproduces the exact expansion at $Q=2,\dots,10$
(Table~\ref{tab:competitor}); the generalisation to all $N$ is
Eq.~\eqref{eq:mminN}. Table~\ref{tab:support} in Appendix~\ref{app:proofs} gives
$m_{\min}$ against Theorem~\ref{thm:lb} at twelve points, where the bound holds
with ratio $1.07$--$1.37$. The supports are computed, not estimated: the fast
Walsh--Hadamard transform of the exact phase table agrees with symbolic
Pauli-string enumeration ($60/60$, $315/315$, $171/171$), an independent
symbolic expansion from the squared bilinears reproduces every $m_{\min}$ of
Tables~\ref{tab:support}, \ref{tab:competitor} and~\ref{tab:ftmatrix}, and an
explicit $m_{\min}$-rotation phase polynomial reproduces the target unitary to
$4\times10^{-15}$.

\begin{corollary}
\label{cor:wall}
Within the gate class of Lemma~\ref{lem:twoq}, phase-polynomial compilation of
$e^{-i\tau V}$, as used in Refs.~\cite{halimeh2026,hanada2025}, cannot be
improved below $m_{\min}$ rotations or $m_{\min}$ two-qubit gates by any parity
network, ordering, ZX rewrite or ancilla allocation. Improvement requires a gate
outside that class: $T$ (Toffoli) or $H$.
\end{corollary}

The statement is about a gate class, not about the problem: Secs.~\ref{sec:nisq}
and~\ref{sec:ft} do reduce the cost, by leaving the class in the two ways
Corollary~\ref{cor:wall} permits. Before building them, though, the bound has to
be shown to be tight, or it is only a bound.

\subsection{The competitor at full strength}
\label{sec:competitor}
A floor is not a circuit, and a comparison against a floor is only fair if some
circuit nearly attains it. It does. Our structure-aware parity network
(i)~rewrites every boson register into the prefix-difference basis $q_0=b_0$,
$q_k=b_{k-1}\oplus b_k$ ($Q-1$ CNOTs per boson), so that
$b_j=q_0\oplus\cdots\oplus q_j$ and moving a bit index by one changes a parity
by one qubit; (ii)~visits all characters on one ancilla target $T$ in a greedy
Hamming tour, one CNOT per unit of Hamming distance and one $R_z$ per character;
(iii)~clears $T$ and undoes the basis change (Fig.~\ref{fig:competitor}). It is
verified by exact simulation of the permutation-times-phase circuit: system
restored, target clean, and the integer identity
$\sum_y\hat P(y)\chi_y(x)=V_{\rm int}(x)-\mathrm{const}$ holds exactly
(exhaustive at $Q=2$, $2000$ random basis states for $Q=3,\dots,10$), plus a
Qiskit Aer statevector check at $Q=2$ ($3.3\times10^{-16}$).

It needs $1.009$--$1.096\,m_{\min}$ CNOTs for $Q\ge4$ ($1.50$ and $1.20$ at
$Q=2,3$), whereas GraySynth as implemented in Qiskit needs $1.90$, $2.02$ and
$2.27$ per character at $Q=2,3,4$ (Table~\ref{tab:competitor}). In matrix-entry
coordinates the built network needs $1.04$--$1.29\,m_{\min}$
(Table~\ref{tab:nisqsun}; not built at the three largest points, where only the
floor is used, which favours the incumbent). \textbf{All NISQ comparisons below
are therefore against $m_{\min}$ or against the built network, never against a
formula lying above the floor.}

\begin{table}[htbp]\centering\small
\caption{The NISQ competitor at full strength, $\mathrm{SU}(2)$, $d=2$, potential term per Trotter step. $m_{\min}$: exact Walsh support, which is the two-qubit floor of Lemma~\ref{lem:twoq} for \emph{every} parity network. ``Network'': the structure-aware parity network of Sec.~\ref{sec:competitor}, built and verified (exhaustively at $Q=2$, $2000$ random basis states otherwise; exact integer phase identity). GraySynth: Qiskit's \texttt{synth\_cnot\_phase\_aam}. ``Published'': the CNOT estimate of Ref.~\cite{halimeh2026}, $(2+4/Q)\,m_{\min}$. Depth: two-qubit depth of the emitted network.}
\label{tab:competitor}
\begin{tabular}{rrrrrrr}
\hline
$Q$ & $m_{\min}$ & network & network$/m_{\min}$ & GraySynth & published & depth\\
\hline
2 & 60 & 90 & 1.500 & 114 & 240 & 80\\
3 & 315 & 378 & 1.200 & 635 & 1050 & 358\\
4 & 1020 & 1118 & 1.096 & 2320 & 3060 & 1086\\
5 & 2535 & 2662 & 1.050 & -- & 7098 & 2619\\
6 & 5328 & 5496 & 1.032 & -- & 14208 & 5442\\
7 & 9975 & 10186 & 1.021 & -- & 25650 & 10121\\
8 & 17160 & 17408 & 1.014 & -- & 42900 & 17332\\
9 & 27675 & 27992 & 1.011 & -- & 67650 & 27905\\
10 & 42420 & 42802 & 1.009 & -- & 101808 & 42704\\
\hline
\end{tabular}\end{table}

\section{The NISQ construction: Fourier-register squaring}
\label{sec:nisq}

\subsection{Which primitive to leave the gate class with}
\label{sec:exchange}
Corollary~\ref{cor:wall} says that improvement needs $T$ or $H$. Which of the
two to use is decided by hardware, and the two eras decide it oppositely. A
logical AND costs $4$ T fault-tolerantly but about $6$ CNOTs on NISQ hardware
(three with a relative-phase Toffoli, three more to uncompute). A dyadic
rotation costs up to $T_{\rm add}(L)=68$--$140$ T fault-tolerantly, less for
angles of high $2$-adic valuation, Sec.~\ref{sec:precision}, but one native
$R_{ZZ}$ on NISQ hardware, or two CNOTs, about one with the register frames of
S5 below. Arithmetic is therefore cheap exactly where rotations are expensive,
and expensive where they are cheap. This is not a heuristic: building the
Toffoli-arithmetic route and charging it with NISQ costs shows it needs
$15$--$6.4\times$ the two-qubit gates of the Fourier route
(Appendix~\ref{app:negative}). The NISQ exit is therefore Hadamard, and it is
available because Fourier-basis arithmetic computes into a register using phase
rotations and Hadamards alone.

\subsection{Construction}
\label{sec:nisqconstr}
Let $v=\sum_{(p,q)}c_{pq}m_pm_q$ be one squared quantity: $K^c_{IJ}$ for
$\mathrm{SU}(2)$~\eqref{eq:cross}, one commutator entry for $N\ge3$. By
Eq.~\eqref{eq:mZ}, $v=\sum_ta_tP_t$ with $P_t=Z_{p,j}Z_{q,k}$ and
$a_t=c_{pq}2^{\,j+k}$. The gadget (Fig.~\ref{fig:nisqgadget}) applies
$e^{-i\theta v^2}$:
\begin{enumerate}
\item $H^{\otimes w}|0\rangle=|\phi(0)\rangle$, the Fourier state of $0$ (the
first QFT is free).
\item \textbf{Fourier addition} $\mathcal A_+$ of $v/2^r$ ($r$ below). Adding
$aP$ to the Fourier register multiplies the level of significance $2^{\,l}$ by
$e^{i\alpha_lP(1-Z_l)/2}=e^{i\alpha_lP/2}\,e^{-i\alpha_lPZ_l/2}$, with
$\alpha_l=2\pi a2^{\,l}/2^w$. The parity $P_t$ is carried by one fold wire $F$
and the second factor is $R_{ZZ}(\alpha_l)$ on $(F,y_l)$.
\item Inverse QFT without swaps: the register holds $v/2^r$ in two's complement.
\item \textbf{Square in place:} $e^{-i\theta4^rv_{\rm reg}^2}$ is $w$
single-qubit phases and $w(w-1)/2$ $R_{ZZ}$ on the register --- no arithmetic.
\item QFT, $\mathcal A_-$ (subtract), $H^{\otimes w}$: the register and $F$
return to $|0\rangle$.
\end{enumerate}
Steps 1--3 are the known Fourier-basis polynomial
evaluation~\cite{gilliam2021,bucher2025,draper2000}; the new element is step 4,
the squaring of a weight-two bilinear inside the register, which replaces the
$\sim Q^4$ characters of $v^2$ by $\sim Q^3$ register rotations. The whole
advantage of the method is that the squaring is free of arithmetic: a quadratic
form on $w=2Q+O(1)$ register qubits, whose cost is $O(Q^2)$ per entry rather
than $O(Q^4)$.

\begin{figure}[t]
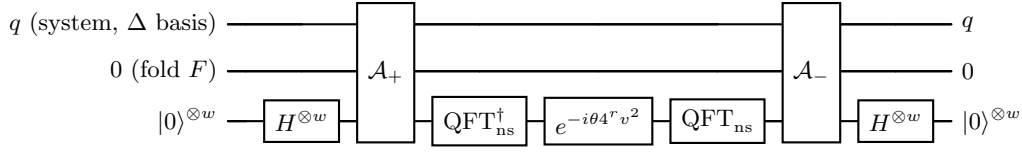
\centering
\fignisqgadget
\caption{The NISQ gadget for one squared quantity $v$. $\mathcal A_\pm$: the
Fourier addition of $\pm v/2^r$, applied through the fold wire $F$ (steps 2
and 5, Fig.~\ref{fig:weight3}); the system is in the prefix-difference basis
($\Delta$) for the whole potential step. $\mathrm{QFT}_{\rm ns}$: QFT without
swaps (the bit reversal is absorbed into the level assignment,
Remark~\ref{rem:trap}). The square is a quadratic form in the register bits. One
fold wire and one $w$-qubit register serve all gadgets sequentially.}
\label{fig:nisqgadget}
\end{figure}

\subsection{The two simplifications that move the crossover}
\label{sec:simpl}
Applied naively, the gadget would apply all $w$ level rotations to every term
and, in the CNOT-only gate set, pay two CNOTs per rotation. Five exact rewrites
remove most of that; none changes a term of the Hamiltonian and none is claimed
as new. Two of them move the crossover and are stated here; the remaining
three: S2, cancellation of the state-independent factors
$e^{i\alpha_lP_t/2}$ between the two passes; S3, the boustrophedon fold tour in
the prefix-difference basis; and S4, halving the register when $\sum c_{pq}$ is
even, are bookkeeping and are given in Appendix~\ref{app:circuits}.
\begin{description}
\item[S1, Draper triangularity.] If $\alpha_l/\pi$ is an even integer the level
factor is the identity; if it is an odd integer $x$ it equals
$-i(-1)^{(x-1)/2}Z_FZ_l$, two single-qubit $Z$ gates. An $R_{ZZ}$ is needed only
when $\alpha_l/\pi\notin\mathbb Z$, i.e.\ on
$\min\!\big(w,\,w+r-1-\nu_2(a_t)\big)$ levels, $\nu_2$ the 2-adic valuation,
instead of all $w$ (Fig.~\ref{fig:weight3}). Together with S2--S4 this removes
$34$--$45\%$ of the unsimplified gadget's two-qubit gates at every $Q$
(Table~\ref{tab:nisqrobust}).
\item[S5, register frames (CNOT only).] Without a native $R_{ZZ}$, each level
rotation $e^{-i\alpha P_tZ_l/2}$ would cost two CNOTs. Instead, a register qubit
whose level is active for the current term is kept in the frame
$y_l\oplus P_t$, so that its rotation is a single-qubit $R_z$. Moving to the next
term, $P_{t+1}=P_t\oplus q$ for a one-qubit move, costs one CNOT $q\to y_l$ per
level that stays active, one CNOT $F\to y_l$ per level that enters or leaves, and
the fold move, about one CNOT per rotation. For a jump of Hamming distance
$d>2$ the staying levels are unframed and reframed through $F$ instead (two
CNOTs each). Frames are cleared before the QFT, so the register is plain between
the passes; QFT phases and the square use CX--$R_z$--CX. Figure~\ref{fig:frames}
shows the same two terms as Fig.~\ref{fig:weight3} in this emission. Register
frames cut the CNOT-only count by $18$--$40\%$ relative to the naive translation
of every $R_{ZZ}$ into two CNOTs, and with the naive translation the CNOT-only
crossover would move from $Q=7$ to $Q=9$ at $\mathrm{SU}(2)$.
\end{description}
In each gate set we use the cheaper emission: native $R_{ZZ}$ where it exists,
register frames otherwise.

\begin{remark}[a trap]
\label{rem:trap}
The inverse QFT's bit-reversal swaps cannot simply be dropped: without them the
transform acts on the bit-reversed state, which is not a Fourier state. We
absorb the reversal into the addition by placing significance $2^{\,l}$ on
register qubit $w-1-l$; the register then emerges in normal bit order.
\end{remark}

\subsection{Correctness}
\label{sec:nisqcorrect}
Every gate that touches a system qubit is a CNOT into $F$ or diagonal, so for a
system basis state the system and $F$ stay classical and the register can be
simulated exactly with $2^w$ amplitudes. Chaining the gadgets of a whole
potential step, we check that each gadget returns the register to $|0\rangle$
and $F$ to $0$, that the system is restored, and that the accumulated phase
equals $-\theta V_{\rm int}(x)$ up to a state-independent constant. Register
frames are tracked exactly as a classical bit per basis state and must be zero
before every non-diagonal register gate and at the end. Both emissions are
checked. Leakage is $<10^{-14}$ and the phase error $\le2.1\times10^{-12}$ at
every checked point: $\mathrm{SU}(2)$ exhaustively at $Q=2$ and on random basis
states for $Q=3,\dots,8$; matrix-entry route at
$(N,Q)=(2,2),(3,2),(3,3),(3,4),(4,2),(5,2)$. Independently, the complete
$\mathrm{SU}(2)$, $Q=2$ potential circuit (18 qubits) run on the Qiskit Aer
statevector simulator from a random system state returns
$e^{-i\theta V_{\rm int}}$ times the state, ancillas in $|0\rangle$, to
$2.1\times10^{-15}$ (native emission) and $4.0\times10^{-15}$ (register frames).
The method is exact: it changes no term of the Hamiltonian. Complete Trotter
evolutions, with the kinetic layer, are run on the same simulator in
Sec.~\ref{sec:aer}.

\subsection{Closed-form resources}
\label{sec:nisqclosed}
For the $\mathrm{SU}(2)$ cross-product route, S4 gives $w=2Q+1$ ($Q\ge2$), and
S1 leaves $2Q+1-j-k$ $R_{ZZ}$ levels for term $(j,k)$, so
$\sum_{j,k}(2Q+1-j-k)=Q^2(Q+2)$ per product and pass. With two products, two
passes, three components, the tour
$\tau(Q)=2[2+2(Q^2-1)+(Q+k_{\rm e}+3)]$, $k_{\rm e}=Q-1$ (odd $Q$) or $0$ (even
$Q$), the two QFTs, the square and the basis change on six bosons,
\begin{equation}
  \boxed{\;N_{2q}^{\rm NISQ}=
  \begin{cases}
  12Q^3+54Q^2+27Q+6 & Q\ \text{even},\\
  12Q^3+54Q^2+33Q & Q\ \text{odd},
  \end{cases}
  \qquad \text{qubits}=8Q+2\;}
  \label{eq:nisqcount}
\end{equation}
with native $R_{ZZ}$/CPhase. The naive CNOT translation costs
$24Q^3+96Q^2+36Q+6$ (even) or $24Q^3+96Q^2+42Q$ (odd). With register frames
(S5) every move inside a block changes the active-level count by one, so it
costs $2+\min(n_t,n_{t+1})$ CNOTs with $n=2Q+1-j-k$; each gadget costs
$2[\,2+(2Q+1)+2W+(4Q+6)+(Q+2-k_{\rm e})\,]+3w(w-1)$ with
$W=2(Q^2-1)+\sum_{\rm moves}\min(n_t,n_{t+1})$, and the parity-dependent terms
cancel:
\begin{equation}
  \boxed{\;N_{\rm CNOT}^{\rm NISQ}=12Q^3+78Q^2+54Q+18\qquad(\text{all }Q\ge2).\;}
  \label{eq:nisqcnot}
\end{equation}
on the same $8Q+2$ qubits. For arbitrary bilinears (matrix-entry route) the same
rules give explicit sums over blocks. All these formulas reproduce the built
circuits exactly at the $21$ points of Tables~\ref{tab:nisqsu2}
and~\ref{tab:nisqsun}. Against Eq.~\eqref{eq:mminsu2}, $m_{\min}/N\to3Q/8$ in
both gate sets as $Q\to\infty$: the advantage grows linearly in $Q$, with no
free parameter to tune.

\subsection{Results}
\label{sec:nisqresults}
\emph{$\mathrm{SU}(2)$, adjoint route.} Table~\ref{tab:nisqsu2} and
Fig.~\ref{fig:nisq}. With native $R_{ZZ}$ the circuit uses fewer two-qubit gates
than any parity network from $Q=6$ ($1.13\times$ the floor, $1.17\times$ the
built network) to $Q=10$ ($2.40\times$); with CNOT only and register frames,
from $Q=7$ ($1.20\times$) to $Q=10$ ($2.08\times$). Below these points it loses:
at $Q=4$ the built network needs $1118$ two-qubit gates against our $1746$
(native) or $2250$ (CNOT). The two-qubit depth of the emitted sequences crosses
earlier, at $Q=5$ in both emissions, and reaches $2.68\times$ (native) and
$4.21\times$ (CNOT) at $Q=10$, relevant because depth, not count, is what a
coherence budget buys. The circuit needs $2Q+1$ more qubits than the network
(fold wire and register, reused by all gadgets).

\begin{table}[htbp]\centering\footnotesize
\setlength{\tabcolsep}{3.2pt}
\caption{NISQ cost per Trotter step, potential term, $\mathrm{SU}(2)$, $d=2$, adjoint coordinates (the cross-product entries $K^c_{12}$ of Eq.~\eqref{eq:cross}). Qubits: network/ours (ours adds a fold wire and a $w=2Q+1$ register, reused by every gadget). $m_{\min}$ is the two-qubit floor of \emph{every} parity network (Lemma~\ref{lem:twoq}); ``network'' is the built structure-aware network of Sec.~\ref{sec:competitor}. Ours: the gadget of Sec.~\ref{sec:nisq}, built and verified, in the better emission for each gate set --- native $R_{ZZ}$/CPhase (RZZ, Eq.~\eqref{eq:nisqcount}) and CNOT only with register frames (CNOT, Eq.~\eqref{eq:nisqcnot}). Ratios $>1$: ours cheaper; $m_{\min}/$ours$>1$ means ours beats every parity network in that gate set. Depth: two-qubit depth of the emitted sequence (ASAP, commuting gates not reordered). Weaker baselines are in Table~\ref{tab:nisqrobust}.}
\label{tab:nisqsu2}
\begin{tabular}{rrrrrrrrrrrr}
\hline
$Q$ & qubits & $m_{\min}$ & network & ours & ours & net/ & $m_{\min}/$ & $m_{\min}/$ & depth & depth & depth\\ & net/ours & (floor) & & RZZ & CNOT & ours & ours & ours & net & ours & ours\\ & & & & & & RZZ & RZZ & CNOT & & RZZ & CNOT\\
\hline
2 & 13/18 & 60 & 90 & 372 & 534 & 0.24 & 0.16 & 0.11 & 80 & 320 & 359\\
3 & 19/26 & 315 & 378 & 909 & 1206 & 0.42 & 0.35 & 0.26 & 358 & 763 & 676\\
4 & 25/34 & 1020 & 1118 & 1746 & 2250 & 0.64 & 0.58 & 0.45 & 1086 & 1482 & 1221\\
5 & 31/42 & 2535 & 2662 & 3015 & 3738 & 0.88 & 0.84 & 0.68 & 2619 & 2579 & 1916\\
6 & 37/50 & 5328 & 5496 & 4704 & 5742 & 1.17 & 1.13 & 0.93 & 5442 & 4084 & 2947\\
7 & 43/58 & 9975 & 10186 & 6993 & 8334 & 1.46 & 1.43 & 1.20 & 10121 & 6123 & 4164\\
8 & 49/66 & 17160 & 17408 & 9822 & 11586 & 1.77 & 1.75 & 1.48 & 17332 & 8702 & 5825\\
9 & 55/74 & 27675 & 27992 & 13419 & 15570 & 2.09 & 2.06 & 1.78 & 27905 & 11971 & 7708\\
10 & 61/82 & 42420 & 42802 & 17676 & 20358 & 2.42 & 2.40 & 2.08 & 42704 & 15912 & 10143\\
\hline
\end{tabular}\end{table}

\begin{figure}[tbp]\centering
\includegraphics[width=\textwidth]{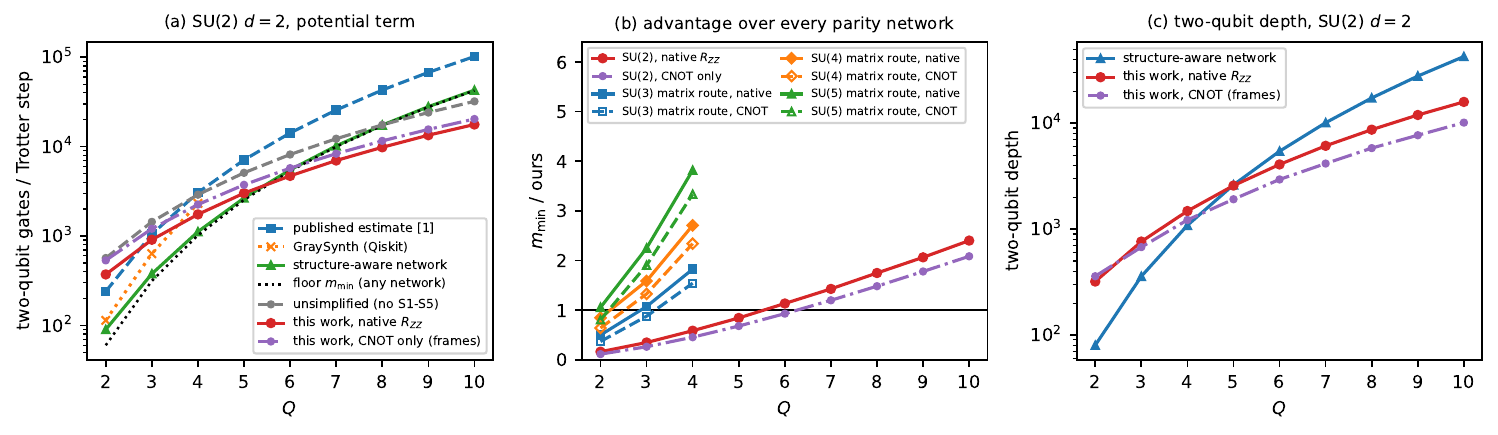}
\caption{NISQ, potential term. (a) Two-qubit gates per Trotter step,
$\mathrm{SU}(2)$, $d=2$, adjoint route: the floor $m_{\min}$ of every parity
network, the built structure-aware network, ours with native $R_{ZZ}$ and with
CNOT only (register frames), and the secondary baselines of
Appendix~\ref{app:secondary}. (b) $m_{\min}$ divided by our count in each gate
set ($>1$: fewer two-qubit gates than any parity network in that gate set),
$\mathrm{SU}(2)$ adjoint and the matrix-entry route at $N=3,4,5$. (c) Two-qubit
depth.}
\label{fig:nisq}
\end{figure}

\emph{$\mathrm{SU}(N)$: the matrix-entry route.}
\label{sec:nisqsun}
For $N\ge3$ one gadget is applied per squared commutator entry
(Sec.~\ref{sec:identities}), with angle $\theta$ times its Frobenius weight.
Table~\ref{tab:nisqsun}: with native $R_{ZZ}$ the circuit beats every parity
network from $Q=3$ at $\mathrm{SU}(3)$ ($1.07\times$, $1.82\times$ at $Q=4$) and
$\mathrm{SU}(4)$ ($1.59\times$, $2.71\times$) and from $Q=2$ at
$\mathrm{SU}(5)$ ($1.05\times$, $2.25\times$, $3.82\times$). With CNOT only and
register frames it wins from $Q=4$ at $\mathrm{SU}(3)$ ($1.54\times$) and from
$Q=3$ at $\mathrm{SU}(4)$ ($1.33\times$, $2.34\times$) and $\mathrm{SU}(5)$
($1.90\times$, $3.33\times$). The advantage grows with $N$ at fixed $Q$ because
the floor grows with the number of distinct quartic vertices while our count
grows with the number of squared entries times their bilinear terms,
quantified in Sec.~\ref{sec:scaling}. On the $\mathrm{SU}(2)$ matrix route the
crossover is not reached within $Q\le4$ ($0.86\times$ at $Q=4$), which is why
$\mathrm{SU}(2)$ is compiled on the adjoint route.

\begin{table}[htbp]\centering\small
\setlength{\tabcolsep}{4pt}
\caption{NISQ, matrix-entry route (Sec.~\ref{sec:nisqsun}), $d=2$, potential term per Trotter step: one gadget per squared commutator entry, in the better emission for each gate set. $m_{\min}$ is the two-qubit floor of every parity network; ``network'' is the built structure-aware network (verified on $300$ random basis states, exhaustively at $N=Q=2$). Our counts are read off the built circuits and equal the closed-form evaluation exactly; the exact-simulation checks are listed in Appendix~\ref{sec:verif}; the $\mathrm{SU}(2)$ adjoint route is Table~\ref{tab:nisqsu2}.}
\label{tab:nisqsun}
\begin{tabular}{rrrrrrrrrr}
\hline
$N$ & $Q$ & qubits & $m_{\min}$ & network & ours & ours & $m_{\min}$/ours & $m_{\min}$/ours & depth\\ & & & & & RZZ & CNOT & RZZ & CNOT & ours\\
\hline
2 & 2 & 16+7 & 186 & 240 & 796 & 1136 & 0.23 & 0.16 & 680\\
2 & 3 & 24+9 & 969 & 1068 & 1912 & 2480 & 0.51 & 0.39 & 1646\\
2 & 4 & 32+11 & 3128 & 3312 & 3620 & 4536 & 0.86 & 0.69 & 3172\\
3 & 2 & 36+8 & 2112 & 2566 & 4245 & 5748 & 0.50 & 0.37 & 3857\\
3 & 3 & 54+10 & 10881 & 11782 & 10152 & 12420 & 1.07 & 0.88 & 9370\\
3 & 4 & 72+12 & 34788 & 36160 & 19071 & 22572 & 1.82 & 1.54 & 17823\\
4 & 2 & 64+8 & 9644 & 11738 & 11440 & 15088 & 0.84 & 0.64 & 10750\\
4 & 3 & 96+11 & 49398 & 53094 & 31136 & 37280 & 1.59 & 1.33 & 29376\\
4 & 4 & 128+13 & 157296 & -- & 58112 & 67280 & 2.71 & 2.34 & 55426\\
5 & 2 & 100+9 & 28950 & 35362 & 27490 & 35990 & 1.05 & 0.80 & 25987\\
5 & 3 & 150+11 & 147840 & -- & 65755 & 77690 & 2.25 & 1.90 & 63004\\
5 & 4 & 200+13 & 469820 & -- & 123080 & 140910 & 3.82 & 3.33 & 118881\\
\hline
\end{tabular}\end{table}

\subsection{Routed counts}
\label{sec:routing}
\emph{All-to-all.} On a device with all-to-all connectivity, such as the trapped-ion
hardware with native $ZZ$ that is the natural platform for the native-$R_{ZZ}$
emission, no SWAP is ever needed, so the logical counts of
Tables~\ref{tab:nisqsu2} and~\ref{tab:nisqsun} are the routed counts (checked:
transpiling every circuit without a coupling map adds no two-qubit gate). The
native-$R_{ZZ}$ crossovers therefore stand as they are for that architecture:
the ``logical counts'' caveat is exact for ions and concerns only
nearest-neighbour devices.

\emph{Heavy-hex and square lattices.} Every NISQ circuit --- our gadget in both
emissions and the built structure-aware network --- is mapped with the same
router (Qiskit SabreLayout and SabreSwap, three seeds, best kept, optimisation
level 1; a SWAP costs three two-qubit gates) onto a distance-7 heavy-hex lattice
(115 qubits) and a $10\times12$ square lattice (120 qubits), at every point that
fits (Table~\ref{tab:routing}, Fig.~\ref{fig:routing}). Routing multiplies every
count, by $3.7$--$7.6\times$ on heavy-hex and $2.2$--$4.7\times$ on the square
lattice, and ours somewhat more than the network, because both our fold wire and
our register interact with many system qubits. Against the routed built network
the crossovers move by one step in $Q$, the same on both lattices and in both
gate sets: $\mathrm{SU}(2)$ from $Q=7$ (native $1.24\times$ heavy-hex,
$1.18\times$ square; CNOT $1.07\times$, $1.01\times$; $2.13$/$2.09\times$ and
$1.92$/$1.84\times$ at $Q=10$), $\mathrm{SU}(3)$ from $Q=4$ (native
$1.42$/$1.45\times$, CNOT $1.20$/$1.21\times$) and $\mathrm{SU}(4)$ from $Q=3$
(native $1.24$/$1.32\times$, CNOT $1.06$/$1.12\times$). At $\mathrm{SU}(5)$,
$Q=2$ the native emission wins after routing ($1.19$/$1.25\times$) and the CNOT
emission ties ($1.00$/$1.02\times$). The routed comparison is against the built
network, not the floor $m_{\min}$, which is not a circuit; routing is heuristic,
but the same router, seeds and settings are applied to every circuit.

\begin{figure}[tbp]\centering
\includegraphics[width=\textwidth]{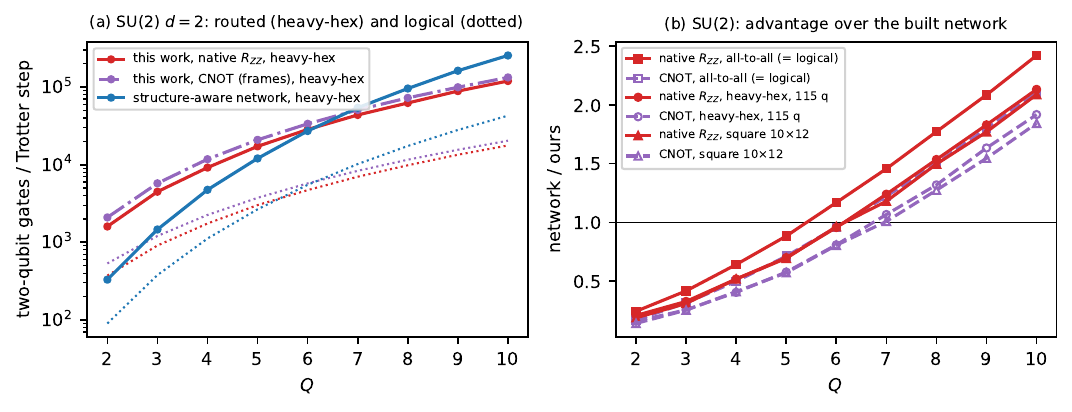}
\caption{Routing, $\mathrm{SU}(2)$, $d=2$, potential term. (a) Two-qubit gates
after routing onto a distance-7 heavy-hex lattice (solid) and before routing
(dotted), for our two emissions and the built network. (b) Built network divided
by ours, on all-to-all hardware (the logical counts), heavy-hex and the
$10\times12$ square lattice, for both gate sets. Counts are in
Table~\ref{tab:routing}.}
\label{fig:routing}
\end{figure}

\section{The fault-tolerant construction: sum-of-squares arithmetic}
\label{sec:ft}

\subsection{Construction}
\label{sec:ftconstr}
Both eras compile the same data: the squared entries
$v_e=\sum_{(p,q)}c_{pq}m_pm_q$ with Frobenius weights $\mu_e$, so that
$V_{\rm int}=\sum_e\mu_ev_e^2$, the $3P$ cross-product entries $K^c_{IJ}$ for
$\mathrm{SU}(2)$ in adjoint coordinates, the $N^2$ commutator entries per pair
for the matrix route. The oracle (Fig.~\ref{fig:ftskeleton}) proceeds as
follows:
\begin{enumerate}
\item build $|m|$ for each boson (Fig.~\ref{fig:absm}); zero ANDs;
\item for each entry $e$: (a) compute every product $|m_p||m_q|$ with the
windowed multiplier and keep it live; (b) add it into the entry register $C_e$
with its sign $\mathrm{MSB}(n_p)\oplus\mathrm{MSB}(n_q)\oplus[c_{pq}<0]$;
(c) $|C_e|$ by one conditional negation, squared into $C_e^2$; (d) add
$\mu_eC_e^2$ directly into the catalytic phase-gradient register
$|\gamma_L\rangle\propto\sum_ke^{2\pi ik/2^L}|k\rangle$
(Fig.~\ref{fig:inject}); (e) invert (c)--(a);
\item uncompute the magnitudes.
\end{enumerate}
It contains no arbitrary-angle rotation anywhere, and the phase is exact on the
dyadic ladder of Sec.~\ref{sec:precision} ($K^c$ and, by
Corollary~\ref{cor:integer}, the matrix entries are integers in $m$-units). This
is the structural difference from the incumbent: the incumbent's cost is one
rotation per Walsh character and therefore scales with
$m_{\min}=\Theta(N^4Q^4)$, while ours is one multiplication per product term and
one squaring and one injection per entry, which is $\Theta(N^3Q^2)$.

One point of the arithmetic bookkeeping is worth a factor and is therefore
stated here; the other two: F2, signed accumulation by carry-in, and F3,
merged products with per-entry register widths, are in
Appendix~\ref{app:circuits}. None is claimed as new.
\begin{description}
\item[F1, per-entry injection with live products.] The natural organisation,
compute each product, add it into $C$, uncompute it inside the compute stage,
sum all squares into a register $V$, inject $V$ once, then mirror the whole
compute stage, multiplies every product four times. Injecting $\mu_eC_e^2$
per entry and inverting only the entry block, with the products kept live until
then, multiplies each product twice and removes the $V$ register; the price is
one $T_{\rm add}(L)$ per entry instead of one per step. The ablation of
Sec.~\ref{sec:ablation} quantifies the same effect independently: its
single-injection variant costs about twice its per-monomial one.
\end{description}

\begin{figure}[t]
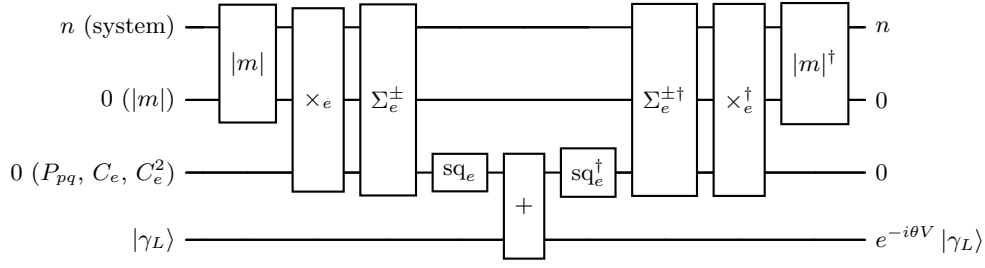
\centering
\figftskeleton
\caption{Skeleton of the fault-tolerant oracle. Magnitudes are extracted once
(Fig.~\ref{fig:absm}); then, entry by entry, products are formed and kept live
($\times_e$), accumulated with their signs by carry-in ($\Sigma^\pm_e$), the
magnitude of $C_e$ is squared ($\mathrm{sq}_e$), $\mu_eC_e^2$ is added into
$|\gamma_L\rangle$ (Fig.~\ref{fig:inject}), and the entry block is inverted
before the next entry. The factor two of Eq.~\eqref{eq:Tsos} is the block
inversion.}
\label{fig:ftskeleton}
\end{figure}

\subsection{Closed-form resources}
\label{sec:ftclosed}
With $T_{\rm condneg}(n)=T_{\rm add}(n)$ and the primitives of
Sec.~\ref{sec:costmodel},
\begin{equation}
\begin{aligned}
  T=\sum_e\Big\{&2\Big[\sum_{(p,q)\in e}\big(T_{\rm mul}(Q,Q,2Q)
  +T_{\rm add}(w_e-\nu_2(c_{pq}))\big)+T_{\rm condneg}(w_e)
  +T_{\rm sq}(w_e{-}1,2w_e{-}2)\Big]\\
  &+T_{\rm add}(L-\nu_2(\mu_e))\Big\},
\end{aligned}
  \label{eq:Tsos}
\end{equation}
on $n_bQ+L+{}$scratch qubits, where $w_e$ is the entry register width (F3). The
formula equals the T-count read off the built circuit at all $17$ points. Its
structure is the whole point of Sec.~\ref{sec:ablation}: the cost is linear in
the number of \emph{products} $\sum_en_e$, not in the number of quartic
monomials.

\subsection{Results against the primary incumbent}
\label{sec:ftresults}
Tables~\ref{tab:ftsu2}, \ref{tab:ftmatrix} and Fig.~\ref{fig:ft}. Against
precision-matched Hamming-weight phasing the oracle loses only at
$\mathrm{SU}(2)$, $Q=2$ on the adjoint route ($0.77\times$; on the matrix route
the same point is $1.13\times$) and wins at every other point: $1.42\times$ at
$\mathrm{SU}(2)$ adjoint, $Q=3$, and $3.94\times$ at $Q=6$; on the matrix route
$1.13$, $1.17$, $1.63$, $2.00\times$ at $Q=2$ for $N=2,\dots,5$, and $2.82$,
$5.20$, $8.16$, $11.91\times$ at $Q=4$.

\emph{Width.} Unbatched Hamming-weight phasing needs far more ancillas than the
oracle ($146\,289$ against $475$ at $\mathrm{SU}(5)$, $Q=4$). Restricted to our
ancilla budget it costs more, and the oracle wins at every point
($1.30$--$20.1\times$) except $\mathrm{SU}(2)$ adjoint, $Q=2$, where the
incumbent's whole demand ($34$ ancillas) is below ours ($52$) and the two
comparisons coincide. Our own ancillas, $W-\text{sys}-L$, range from $52$ to
$475$ against system registers of $12$ to $200$ qubits; this is a real cost of
the method, and it is what the width-matched comparison and the surface-code
estimate of Sec.~\ref{sec:surface} are for. The Toffoli depth of the as-built
oracle is in Table~\ref{tab:ftrobust}.

\emph{Secondary baselines.} Against the published estimate at $T_{\rm typ}=10$
the oracle wins at every point, from $1.73\times$ to $47.2\times$; against one
precision-matched rotation per Walsh character, $2.06$--$301\times$. Both are
weaker than the primary incumbent, the first because its rotation price
corresponds to $\varepsilon\approx10^{-1}$ (Sec.~\ref{sec:precision}), the second
because it does not share angles, and neither is used for a headline claim.

\emph{Full Trotter step.} The optimised kinetic layer (Sec.~\ref{sec:kinopt}) is
identical in every compilation and dilutes every ratio: $0.79\times$
($\mathrm{SU}(2)$ adjoint, $Q=2$), $1.37\times$ ($Q=3$), $3.33\times$ ($Q=6$),
$4.93\times$ ($\mathrm{SU}(3)$, $Q=4$), $7.85\times$ ($\mathrm{SU}(4)$, $Q=4$)
and $11.55\times$ ($\mathrm{SU}(5)$, $Q=4$) (Table~\ref{tab:fullstep}). No
crossover moves.

\begin{table}[htbp]\centering\small
\setlength{\tabcolsep}{4pt}
\caption{Fault-tolerant T-count per Trotter step, potential term, $\mathrm{SU}(2)$, $d=2$, adjoint coordinates (the three cross-product entries $K^c_{12}$). ``Ours'' (T-count) and the width $W$ (all logical qubits) are read off the built and verified sum-of-squares oracle of Sec.~\ref{sec:ft} (F1--F3). $L$ is the phase-gradient width. Ratios are incumbent/ours ($>1$: ours cheaper). $B_{\rm HWP}$, the primary incumbent, is Hamming-weight phasing of the equal-angle classes with each rotation charged what its own angle costs, $T_d(L-1-\nu_2(\hat\phi)-k)$ (Sec.~\ref{sec:precision}); it is built and verified, with T-count and width equal to the formula at every point. It is given unbatched --- the cheaper of the two in T, hence the reference for T-count claims --- and width-matched, ``$W$'', with batches restricted so that its ancilla demand does not exceed ours, which is the reference for width and spacetime claims (Sec.~\ref{sec:incumbents}). Our own ancillas are $W-\text{sys}-L$. ``Full step'': potential $+$ the identical optimised kinetic layer (Sec.~\ref{sec:kinopt}), against unbatched $B_{\rm HWP}$. $B_{\rm pub}$, placed last, is the cost formula of Ref.~\cite{halimeh2026}, a secondary baseline: its rotation price $T_{\rm typ}=10$ corresponds to $\varepsilon\approx10^{-1}$ (Sec.~\ref{sec:precision}), so it is the weaker comparison and carries no headline claim. Further baselines, widths and depths are in Table~\ref{tab:ftrobust}.}
\label{tab:ftsu2}
\begin{tabular}{rrrrrrrrrr}
\hline
$Q$ & sys & $W$ & $L$ & $m_{\min}$ & ours & $B_{\rm HWP}^{\rm match}$ & $B_{\rm HWP}^{{\rm match},W}$ & full step & $B_{\rm pub}^{10}$\\
\hline
2 & 12 & 82 & 18 & 60 & 1500 & 0.77 & 0.77 & 0.79 & 3.41\\
3 & 18 & 108 & 22 & 315 & 2748 & 1.42 & 1.72 & 1.37 & 9.43\\
4 & 24 & 135 & 26 & 1020 & 4380 & 2.14 & 3.17 & 1.97 & 18.70\\
5 & 30 & 163 & 30 & 2535 & 6396 & 2.96 & 4.99 & 2.59 & 31.27\\
6 & 36 & 191 & 34 & 5328 & 8796 & 3.94 & 7.07 & 3.33 & 47.15\\
\hline
\end{tabular}\end{table}

\begin{table}[htbp]\centering\small
\setlength{\tabcolsep}{4pt}
\caption{Fault-tolerant T-count per Trotter step, potential term, matrix route, $d=2$. ``Ours'' (T-count) and the width $W$ (all logical qubits) are read off the built and verified sum-of-squares oracle of Sec.~\ref{sec:ft} (F1--F3). $L$ is the phase-gradient width. Ratios are incumbent/ours ($>1$: ours cheaper). $B_{\rm HWP}$, the primary incumbent, is Hamming-weight phasing of the equal-angle classes with each rotation charged what its own angle costs, $T_d(L-1-\nu_2(\hat\phi)-k)$ (Sec.~\ref{sec:precision}); it is built and verified, with T-count and width equal to the formula at every point. It is given unbatched --- the cheaper of the two in T, hence the reference for T-count claims --- and width-matched, ``$W$'', with batches restricted so that its ancilla demand does not exceed ours, which is the reference for width and spacetime claims (Sec.~\ref{sec:incumbents}). Our own ancillas are $W-\text{sys}-L$. ``Full step'': potential $+$ the identical optimised kinetic layer (Sec.~\ref{sec:kinopt}), against unbatched $B_{\rm HWP}$. $B_{\rm pub}$, placed last, is the cost formula of Ref.~\cite{halimeh2026}, a secondary baseline: its rotation price $T_{\rm typ}=10$ corresponds to $\varepsilon\approx10^{-1}$ (Sec.~\ref{sec:precision}), so it is the weaker comparison and carries no headline claim. Further baselines, widths and depths are in Table~\ref{tab:ftrobust}. The $N=4,5$ oracles are verified on $5$--$40$ sampled basis states; $N=2,3$ as in Appendix~\ref{sec:verif}.}
\label{tab:ftmatrix}
\begin{tabular}{rrrrrrrrrrr}
\hline
$N$ & $Q$ & sys & $W$ & $L$ & $m_{\min}$ & ours & $B_{\rm HWP}^{\rm match}$ & $B_{\rm HWP}^{{\rm match},W}$ & full step & $B_{\rm pub}^{10}$\\
\hline
2 & 2 & 16 & 120 & 24 & 186 & 2952 & 1.13 & 1.30 & 1.12 & 1.73\\
2 & 3 & 24 & 150 & 26 & 969 & 5096 & 1.83 & 2.54 & 1.76 & 5.09\\
2 & 4 & 32 & 185 & 30 & 3128 & 7912 & 2.82 & 4.61 & 2.61 & 10.35\\
3 & 2 & 36 & 201 & 24 & 2112 & 11796 & 1.17 & 1.89 & 1.17 & 2.20\\
3 & 3 & 54 & 265 & 29 & 10881 & 20184 & 2.79 & 5.19 & 2.72 & 6.50\\
3 & 4 & 72 & 328 & 33 & 34788 & 31032 & 5.20 & 10.13 & 4.93 & 13.36\\
4 & 2 & 64 & 297 & 26 & 9644 & 28432 & 1.63 & 2.91 & 1.62 & 2.88\\
4 & 3 & 96 & 409 & 30 & 49398 & 53072 & 4.05 & 7.53 & 3.96 & 7.81\\
4 & 4 & 128 & 509 & 34 & 157296 & 80720 & 8.16 & 14.88 & 7.85 & 16.24\\
5 & 2 & 100 & 422 & 26 & 28950 & 62180 & 2.00 & 3.48 & 1.99 & 3.22\\
5 & 3 & 150 & 565 & 30 & 147840 & 104740 & 5.82 & 9.86 & 5.71 & 9.67\\
5 & 4 & 200 & 711 & 36 & 469820 & 160620 & 11.91 & 20.14 & 11.55 & 19.92\\
\hline
\end{tabular}\end{table}

\begin{figure}[tbp]\centering
\includegraphics[width=\textwidth]{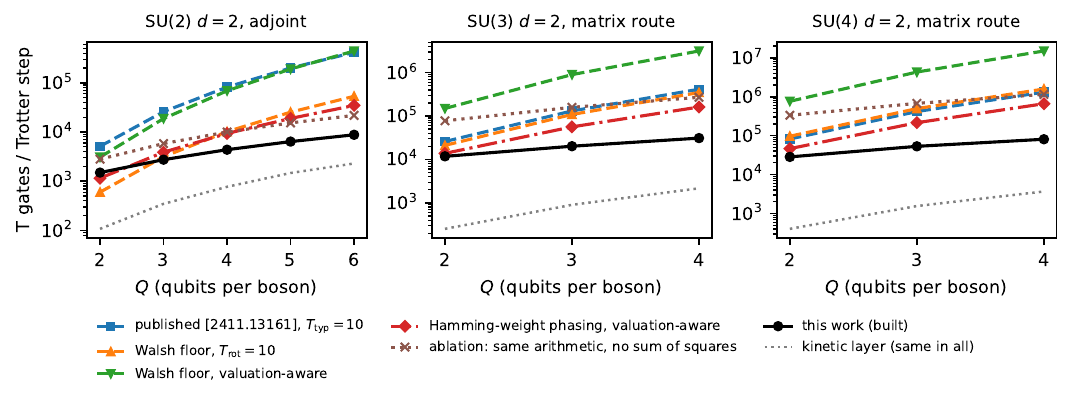}
\caption{Fault-tolerant T-count per Trotter step, potential term: the primary
incumbent (Hamming-weight phasing, valuation-aware), the built oracle, the
ablation of Sec.~\ref{sec:ablation} (same arithmetic, no sum of squares), the
optimised kinetic layer common to all of them (Sec.~\ref{sec:kinopt}), and the
secondary baselines of Sec.~\ref{sec:incumbents}.}
\label{fig:ft}
\end{figure}

\section{Scaling with \texorpdfstring{$N$}{N}}
\label{sec:scaling}
The measured advantages grow with $N$; this section derives why, and checks the
derivation against the built circuits. Per matrix pair, in matrix-entry
coordinates with merged products, the number of product terms summed into the
$N^2$ squared entries is
\begin{equation}
  \sum_en_e=2N(N-1)(2N-1),
  \label{eq:sumne}
\end{equation}
since each of the $N(N-1)/2$ entries $\mathrm{Re}_{i<j}$ and $\mathrm{Im}_{i<j}$
has $4(N-1)$ products (Eq.~\eqref{eq:nterms}) and each of the $N$ entries
$\mathrm{Im}_{ii}$ has $2(N-1)$ merged ones. The distinct quartic monomials,
by multiplicity type, are
\begin{equation}
  n_{1111}=\tfrac12N(N-1)(12N^2-28N+17),\qquad n_{22}=N(N-1)(4N-3),\qquad
  n_{112}=2N(N-1);
  \label{eq:nmono}
\end{equation}
these three polynomials were interpolated at $N=2,\dots,6$ and checked exactly
against the built monomial lists at $N=7,\dots,10$
(Table~\ref{tab:scalepoly}); Eq.~\eqref{eq:sumne} is derived. Extending
Theorem~\ref{thm:lb} to repeated bosons (a squared boson contributes the
$\binom Q2$ characters $Z_{a,j}Z_{a,j'}$),
\begin{equation}
  m_{\min}^{d=2}=n_{1111}Q^4+n_{22}\tbinom Q2^2+n_{112}\tbinom Q2Q^2
  +2N^2\tbinom Q2+N(N-1)Q^2,
  \label{eq:mminN}
\end{equation}
which equals the exact Walsh expansion at all twelve points $N=2,\dots,5$,
$Q=2,3,4$; for $d>2$ pairs sharing a matrix share only weight-two characters, so
the leading term is $P\,n_{1111}Q^4$ with $P=d(d-1)/2$. To leading order in $Q$,
with every factor of $N$ exact,
\begin{equation}
\begin{aligned}
  \text{floor:}\quad&m_{\min}\simeq P\,n_{1111}Q^4=\Theta(d^2N^4Q^4),\\
  \text{NISQ, ours:}\quad&2P\textstyle\sum_en_e\,Q^2(w_e-Q)=\Theta(d^2N^3Q^3),\\
  \text{FT, HWP:}\quad&\simeq4\,m_{\min}=\Theta(d^2N^4Q^4),\\
  \text{FT, ours:}\quad&2P\textstyle\sum_en_e\,[8Q^2+T_{\rm add}(w_e)]=\Theta(d^2N^3Q^2),\\
  \text{ablation:}\quad&2P\,n_{1111}\cdot48Q^2=\Theta(d^2N^4Q^2),
\end{aligned}
\label{eq:leading}
\end{equation}
where $Q^2(w_e-Q)$ counts the $R_{ZZ}$ levels left by S1 per product and pass,
$T_{\rm mul}(Q,Q,2Q)=8Q^2$, and $w_e=2Q+O(\log N)$. Every advantage therefore
grows at fixed $Q$ like
\begin{equation}
  \frac{n_{1111}}{\sum_en_e}=\frac{12N^2-28N+17}{4(2N-1)}\sim\frac{3N}{2},
  \label{eq:growth}
\end{equation}
i.e.\ linearly in $N$, and in $Q$ like $Q$ (NISQ), $Q^2$ (fault tolerance
against Hamming-weight phasing) or not at all (the ablation). The last entry is
the sharpest form of the statement made in the next section: the arithmetic
alone buys a constant, the factorisation buys a power of $N$.

Figure~\ref{fig:scaling}(a) compares the measured ratios with these
leading-order predictions, and Table~\ref{tab:scaleratio} gives them at every
point. The NISQ prediction is within $0.73$--$1.00$ of the measurement
everywhere; the fault-tolerant one within $0.81$--$0.85$ at $Q=4$ for $N\ge3$ and
less accurate at small $N$ and $Q$ (a factor $2.6$ at $N=Q=2$), where
subleading terms dominate; the measured ablation ratio is $0.52$--$0.83$ of its
prediction, rising with $Q$, because our additions, subleading in $Q$, are
not in the leading term. Both the growth in $N$ and the growth in $Q$ are
reproduced.

\subsection{The crossovers and the physics lie in the same direction}
\label{sec:whichway}
Every advantage established here is an increasing function of $N$ and of $Q$,
and every crossover in Table~\ref{tab:crossover} is a lower bound on one of
them. Read pessimistically, that says the methods are useless at small $Q$;
read correctly, it says they become useful exactly where the physics does,
because $N$ and $Q$ are the two axes along which the target regime lies and they
are not independent of one another.

The holographic content of these models is a large-$N$ statement: the thermal
large-$N$ states of BFSS at strong coupling are dual to a black
zero-brane~\cite{hanada2025}, and the supergravity description requires
$N\to\infty$, which is why Ref.~\cite{hartnett2026} identifies the $O(N^2)$
growth of the qubit register as what places the holographic regime beyond
near-term reach. Our advantage grows \emph{linearly in $N$}
(Eq.~\eqref{eq:growth}): the direction that makes the problem hard is the
direction in which the compilation gains.

Large $N$ then forces large $Q$. The authors of the coordinate-basis formulation
identify the large-$N$ limit of matrix models as one of the two cases they know
of that demand a large grid $\Lambda=2^Q$, on the ground that a probe in a
black-hole background has position components of typical size $\sqrt N$ that
must still be resolved at variance of order one, so that the number of points
per boson has itself to grow with $N$~\cite{hanada2025}. The two axes are
therefore traversed together, and our advantage grows along both: linearly in
$N$, and in $Q$ like $Q$ on the NISQ side and $Q^2$ fault-tolerantly. At
$\mathrm{SU}(5)$, $Q=4$, the largest point built, and still far from the
holographic regime, this already compounds to $3.8\times$ fewer two-qubit
gates than any parity network and $11.9\times$ fewer T gates than the primary
incumbent.

Two qualifications keep this honest. We do not quantify the $Q$ required at a
given $N$, coupling and target accuracy; that is a question about truncation
error, not about compilation, and it is not settled here. And a compilation that
is several times cheaper does not by itself bring the black-hole regime within
reach, the register still grows as $O(N^2)$, and that is the binding
constraint. What the constructions change is the constant in front of it, in the
direction where the constant matters most.

\begin{figure}[tbp]\centering
\includegraphics[width=\textwidth]{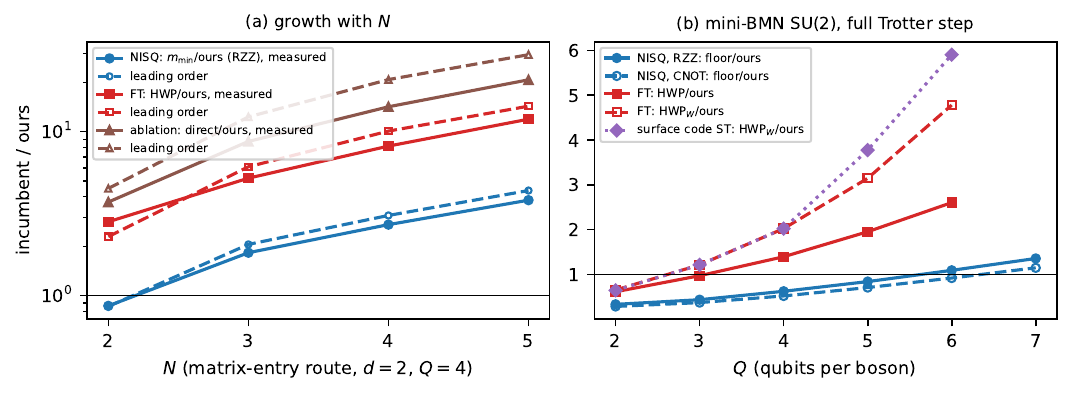}
\caption{(a) Incumbent divided by ours at $Q=4$ against $N$ (matrix-entry route,
$d=2$): NISQ floor $m_{\min}$/ours (native $R_{ZZ}$), fault-tolerant
Hamming-weight phasing/ours, and ablation/ours, measured (solid) and at the
leading order of Eq.~\eqref{eq:leading} (open, dashed). (b) Supersymmetric
mini-BMN, $\mathrm{SU}(2)$, full Trotter step (potential, kinetic and fermion
layers): NISQ floor/ours in both gate sets, fault-tolerant T-count ratios against
unbatched and width-matched Hamming-weight phasing, and the surface-code
spacetime ratio against the width-matched incumbent (Sec.~\ref{sec:bmn}).}
\label{fig:scaling}
\end{figure}

\section{The factorisation, not the arithmetic: two ablations}
\label{sec:ablation}
Reversible arithmetic with a phase-gradient injection is
known~\cite{kitaev2002,gidney2018}, also for a quartic lattice
potential~\cite{madhav2026}. The obvious deflation of Sec.~\ref{sec:ft} is
therefore that the win is just ``known arithmetic applied to a new
Hamiltonian''. It is not, and the point can be settled by experiment rather than
by argument, in two directions.

\paragraph{Same arithmetic, no factorisation.}
We built a control that evaluates the \emph{same} exact integer $V_{\rm int}$,
with the \emph{same} primitives, pricing and gradient register, from its quartic
monomials instead:
\begin{equation}
  V_{\rm int}=\sum_Mc_M\,m_{a_1}m_{a_2}m_{a_3}m_{a_4},
  \label{eq:mono}
\end{equation}
obtained by expanding $\sum_e\mu_ev_e^2$ symbolically ($9$, $23$, $189$, $762$
and $2150$ monomials at $\mathrm{SU}(2)$ adjoint and $N=2,\dots,5$ matrix
entries). Each monomial is a chain of unsigned products of $|m|$ registers (a
repeated boson is squared, not multiplied: types $(1,1,1,1)$, $(1,1,2)$,
$(2,2)$, each with its cheapest order); its sign is the XOR of the MSBs of the
bosons of odd multiplicity and of $\mathrm{sgn}\,c_M$; $|c_M|$ is applied by one
signed shifted addition per set bit. Two injection strategies are built: one
injection per monomial (``direct'', the cheaper), and all monomials accumulated
into one register with a single injection and the compute stage mirrored
(``single''). Both are verified (payload exact, system restored, ancillas clean;
$5$ points $\times$ $2$ modes, $0$ failures) and equal their closed forms.

The result (Table~\ref{tab:ablation}, Fig.~\ref{fig:ft}) is unambiguous. The
direct evaluation costs $1.87$--$2.48\times$ our oracle at $\mathrm{SU}(2)$
adjoint, $2.88$--$3.72\times$ on the $\mathrm{SU}(2)$ matrix entries, and
$6.6$--$8.7$, $11.7$--$14.2$ and $15.5$--$20.7\times$ at $N=3$, $4$, $5$; the
single-injection variant is always about twice the direct one. More importantly,
the direct evaluation \emph{loses to the Hamming-weight-phasing incumbent at
$15$ of the $17$ points} (incumbent/direct $=0.13$--$0.93$), winning only at
$\mathrm{SU}(2)$, $Q=5,6$. The arithmetic by itself does not beat the incumbent;
the factorisation does.

The mechanism is visible in Eq.~\eqref{eq:Tsos}. A monomial of type $(1,1,1,1)$
costs a three-multiplication chain, about $48Q^2$ T forward, while the sum of
squares pays one $Q\times Q$ product, $8Q^2$ T, per product term of an entry and
one square per entry --- and, decisively, the number of product terms grows one
power of $N$ more slowly than the number of monomials, Eq.~\eqref{eq:growth}.
This is also why the ablation's disadvantage grows with $N$ while its ratio to
the incumbent is nearly $Q$-independent: the ablation and the incumbent share
the same $\Theta(N^4)$ object count and differ only in the price per object.

\paragraph{Same factorisation, no advantage to gain.}
The converse control is the kinetic term, and it is at least as sharp. In the
momentum register the kinetic term is quadratic, so its Walsh support is only
$n_bQ(Q-1)/2$ characters and by Lemma~\ref{lem:anc} its phase polynomial is
\emph{already} optimal, there is nothing for a sum of squares to factor. If
arithmetic were what pays, it would still help there. It does not: the built
arithmetic kinetic oracle costs $5.4$--$9.5\times$ the phase-polynomial layer at
$Q=2,3,4$ even when the latter is charged the weak flat price of $10$ T per
rotation, and more against the valuation-aware price of Eq.~\eqref{eq:Td} that
the layer of Sec.~\ref{sec:kinopt} actually pays. Arithmetic is a tool for
turning a large factored object into a small number of evaluations; applied
where there is no factorisation to exploit, it loses. The two controls bracket
the claim from both sides.

\begin{table}[htbp]\centering\footnotesize
\setlength{\tabcolsep}{3.5pt}
\caption{The arithmetic ablation (Sec.~\ref{sec:ablation}): the same exact integer $V_{\rm int}$, primitives, pricing and gradient register as our oracle, but evaluated from its quartic monomials (no sum of squares). ``Direct'': one injection per monomial (the cheaper variant); ``single'': all monomials accumulated into one register, one injection, the compute stage mirrored. Both are built and verified and equal their closed forms. HWP: the precision-matched Hamming-weight-phasing incumbent. All counts are T gates per Trotter step, potential term, $d=2$.}
\label{tab:ablation}
\begin{tabular}{lrrrrrrrrr}
\hline
model & $Q$ & mono- & direct & direct & ours & direct/ & HWP & HWP/ & HWP/\\ & & mials & & single & & ours & & direct & ours\\
\hline
SU(2) adj. & 2 & 9 & 2808 & 5684 & 1500 & 1.87 & 1156 & 0.41 & 0.77\\
SU(2) adj. & 3 & 9 & 5832 & 11748 & 2748 & 2.12 & 3896 & 0.67 & 1.42\\
SU(2) adj. & 4 & 9 & 10008 & 20116 & 4380 & 2.28 & 9356 & 0.93 & 2.14\\
SU(2) adj. & 5 & 9 & 15336 & 30788 & 6396 & 2.40 & 18904 & 1.23 & 2.96\\
SU(2) adj. & 6 & 9 & 21816 & 43764 & 8796 & 2.48 & 34684 & 1.59 & 3.94\\
SU(2) matrix & 2 & 23 & 8508 & 17108 & 2952 & 2.88 & 3336 & 0.39 & 1.13\\
SU(2) matrix & 3 & 23 & 17172 & 34444 & 5096 & 3.37 & 9308 & 0.54 & 1.83\\
SU(2) matrix & 4 & 23 & 29412 & 58940 & 7912 & 3.72 & 22288 & 0.76 & 2.82\\
SU(3) matrix & 2 & 189 & 77652 & 155396 & 11796 & 6.58 & 13820 & 0.18 & 1.17\\
SU(3) matrix & 3 & 189 & 158952 & 318016 & 20184 & 7.88 & 56344 & 0.35 & 2.79\\
SU(3) matrix & 4 & 189 & 270360 & 540848 & 31032 & 8.71 & 161384 & 0.60 & 5.20\\
SU(4) matrix & 2 & 762 & 333240 & 666580 & 28432 & 11.72 & 46296 & 0.14 & 1.63\\
SU(4) matrix & 3 & 762 & 673560 & 1347236 & 53072 & 12.69 & 214744 & 0.32 & 4.05\\
SU(4) matrix & 4 & 762 & 1144824 & 2289780 & 80720 & 14.18 & 658820 & 0.58 & 8.16\\
SU(5) matrix & 2 & 2150 & 960520 & 1921140 & 62180 & 15.45 & 124400 & 0.13 & 2.00\\
SU(5) matrix & 3 & 2150 & 1945640 & 3891396 & 104740 & 18.58 & 609496 & 0.31 & 5.82\\
SU(5) matrix & 4 & 2150 & 3328440 & 6657020 & 160620 & 20.72 & 1913468 & 0.57 & 11.91\\
\hline
\end{tabular}\end{table}

\section{The full Trotter step and surface-code spacetime}
\label{sec:kinetic}

The potential term is where the compilations differ; the kinetic layer is the
same circuit in all of them. That makes it irrelevant to rankings and very
relevant to margins, because once the potential is cheap the common layer can
become the larger part of the step.

\subsection{The kinetic layer}
\label{sec:kinopt}
With $x_n=\delta(n-c)$, $p_k=\tfrac{2\pi}{\Lambda\delta}(k-c)$ and
$c=(\Lambda-1)/2$,
\begin{equation}
  F_{\rm centred}=e^{-2\pi ic^2/\Lambda}\,\Phi\cdot\mathrm{QFT}^\dagger\cdot\Phi,
  \qquad \Phi=\mathrm{diag}\big(e^{2\pi icn/\Lambda}\big)
  =\prod_j\exp\!\big[\tfrac{2\pi ic2^{\,j}}{\Lambda}b_j\big],
  \label{eq:centred}
\end{equation}
so the centring ramps $\Phi$ are products of single-qubit phases
(Fig.~\ref{fig:centredf}): $2Q$ rotations per boson and no two-qubit gates. Our
circuit reproduces $F_{\rm centred}$ to $\le2\times10^{-15}$ for
$Q=2,\dots,5$. The ramps are not optional: omitting $\Phi$ leaves $\hat p^2$
with off-diagonal elements of magnitude $1.7$--$2.0$ in the position basis
against $\|\hat p^2\|\sim10$, i.e.\ an $O(1)$ corruption of the kinetic operator
produced by dropping single-qubit gates.

Per Trotter step and per boson the direct layer is
$F^\dagger\,e^{-2\pi iA_{\rm kin}\tilde m^2/2^L}\,F$: two centred transforms and
one phase polynomial. It is built and verified, the transform reproduces the
analytic kernel to $2\times10^{-15}$ and the whole layer reproduces
$e^{-2\pi iA_{\rm kin}\hat p^2_{\rm int}/2^L}$ to $1.5\times10^{-15}$ for
$Q=2,\dots,5$ and $L=18,26,34$. Two transforms contribute $Q(Q-1)/2$ CPhase and
$\lfloor Q/2\rfloor$ SWAPs each, the phase polynomial $Q(Q-1)/2$ $R_{ZZ}$:
\begin{equation}
  N_{2q}^{\rm kin}=n_b\Big[\tfrac32Q(Q-1)+6\lfloor Q/2\rfloor\Big]
  \quad\text{(native)},\qquad
  n_b\Big[3Q(Q-1)+6\lfloor Q/2\rfloor\Big]\quad\text{(CNOT only)} .
  \label{eq:kin2q}
\end{equation}
Fault-tolerantly every angle in the layer is dyadic, so each rotation is priced
by $T_d$ of Eq.~\eqref{eq:Td}: a $\Phi$ phase $2\pi c2^{\,j}/\Lambda$ has
$b=Q+1-j$; a CPhase at separation $m$ is three phase gates of angle
$2\pi/2^{\,m+2}$ (from $ab=(a+b-a\!\oplus\!b)/2$), $b=m+2$, and there are $Q-m$
of them; the character $Z_jZ_k$ of $\tilde m^2$ carries the integer coefficient
$2A_{\rm kin}2^{\,j+k}$, so $b=L-2-j-k$. Hence
\begin{equation}
  T_{\rm kin}=n_b\Big[4\!\sum_{j=0}^{Q-1}\!T_d(Q{+}1{-}j)
  +6\!\sum_{m=1}^{Q-1}\!(Q{-}m)\,T_d(m{+}2)
  +\!\!\sum_{0\le j<k\le Q-1}\!\!T_d(L{-}2{-}j{-}k)\Big],
  \label{eq:kinT}
\end{equation}
which reproduces the built gate list exactly at all twelve $(Q,L)$ points
tested. At $\mathrm{SU}(2)$, $d=2$, $Q=4$, $L=26$ this is $5124$ T, more than
the $4380$ T of our potential oracle: with the potential this cheap, the layer
common to every compilation has become the larger part of the fault-tolerant
step.

Four exact rewrites fix that. They are standard phasing and
bookkeeping~\cite{gidney2018,kivlichan2020}, none is claimed as new, and they are
stated in full in Appendix~\ref{app:circuits}: K1, the ramps commute through the
momentum-basis phase; K2, the outer ramps telescope across the Trotter product
and cost nothing per step; K3, the bit-reversal swaps are absorbed into a wire
relabelling; and K4, every rotation of the layer carries the same angle on every
boson, so the whole layer becomes one weighted Hamming-weight tree plus a single
subtraction into the gradient register, with the bosons batched to fit our
oracle's ancillas, which are free while the layer runs. K1--K3 are verified on
dense unitaries for $Q=2,\dots,6$ and K2 over three Trotter steps with random
diagonal potentials (deviation $\le1.7\times10^{-15}$); K4 is verified by
reversible simulation, with the gradient payload equal to $A\sum_a\tilde m_a^2$
up to a constant for batches of $6$--$25$ bosons at $Q=2,\dots,6$, every CPhase
class verified for both signs, every qubit restored, and the built T-count equal
to the formula at every point.

The result (Table~\ref{tab:kinopt}) is a layer $3.7$--$9.8\times$ cheaper in T
gates ($5124\to772$ at $\mathrm{SU}(2)$, $Q=4$; $54\,700\to5592$ at
$\mathrm{SU}(5)$, $Q=4$). The kinetic phase polynomial, which dominated the
direct layer with its $L$-bit rotations, now costs one subtraction and a small
tree per batch. The NISQ layer drops to $n_b\,\tfrac32Q(Q-1)$ native gates, or
$n_b\,3Q(Q-1)$ with CNOT only, because the swaps are gone ($180\to108$ at
$\mathrm{SU}(2)$, $Q=4$). This layer is used in every compilation, including the
incumbents, in all full-step tables and in the surface-code estimate.

Table~\ref{tab:fullstep} gives the whole step. Wherever our oracle wins, the
fault-tolerant margins of the full step now lie at most $16\%$ below the
potential-term ones, for instance $5.20\times\to4.93\times$ at
$\mathrm{SU}(3)$, $Q=4$, against $\to3.64\times$ with the direct layer. Wherever
our NISQ circuit wins, its full-step ratios move by less than $3\%$. At the
smallest $Q$, where it loses, the smaller common layer moves them towards the
potential-term values ($0.27\to0.20$ at $\mathrm{SU}(2)$, $Q=2$). No crossover
moves in either era.

\begin{table}[htbp]\centering\footnotesize
\setlength{\tabcolsep}{3.5pt}
\caption{The kinetic layer per Trotter step, all $n_b$ bosons: the direct compilation of Eqs.~\eqref{eq:kin2q} and~\eqref{eq:kinT} against the optimised layer of Sec.~\ref{sec:kinopt} (both built and verified). T gates on the gradient ladder $L$ of each point; ``batches'': number of subtractions into the gradient register for the kinetic phase polynomial (bosons are batched to fit the ancilla budget of our oracle, which is free while the layer runs); ancillas of the optimised layer. Two-qubit gates (direct/optimised) with native $R_{ZZ}$ and with CNOT only. The same layer is used in every compilation compared.}
\label{tab:kinopt}
\begin{tabular}{lrrrrrrrrrr}
\hline
model & $Q$ & $n_b$ & $L$ & T dir. & T opt. & ratio & batches & anc. & 2q RZZ & 2q CNOT\\
\hline
SU(2) adj. & 2 & 6 & 18 & 396 & 108 & 3.67 & 1 & 23 & 54/18 & 72/36\\
SU(2) adj. & 3 & 6 & 22 & 2040 & 348 & 5.86 & 1 & 47 & 90/54 & 144/108\\
SU(2) adj. & 4 & 6 & 26 & 5124 & 772 & 6.64 & 1 & 84 & 180/108 & 288/216\\
SU(2) adj. & 5 & 6 & 30 & 10008 & 1472 & 6.80 & 2 & 93 & 252/180 & 432/360\\
SU(2) adj. & 6 & 6 & 34 & 17052 & 2300 & 7.41 & 2 & 105 & 378/270 & 648/540\\
SU(2) matrix & 2 & 8 & 24 & 720 & 156 & 4.62 & 1 & 33 & 72/24 & 96/48\\
SU(2) matrix & 3 & 8 & 26 & 3104 & 460 & 6.75 & 1 & 63 & 120/72 & 192/144\\
SU(2) matrix & 4 & 8 & 30 & 7600 & 1020 & 7.45 & 1 & 112 & 240/144 & 384/288\\
SU(3) matrix & 2 & 18 & 24 & 1620 & 252 & 6.43 & 1 & 51 & 162/54 & 216/108\\
SU(3) matrix & 3 & 18 & 29 & 7632 & 900 & 8.48 & 1 & 122 & 270/162 & 432/324\\
SU(3) matrix & 4 & 18 & 33 & 18396 & 2164 & 8.50 & 2 & 220 & 540/324 & 864/648\\
SU(4) matrix & 2 & 32 & 26 & 3136 & 404 & 7.76 & 1 & 81 & 288/96 & 384/192\\
SU(4) matrix & 3 & 32 & 30 & 13952 & 1556 & 8.97 & 1 & 209 & 480/288 & 768/576\\
SU(4) matrix & 4 & 32 & 34 & 33472 & 3680 & 9.10 & 2 & 340 & 960/576 & 1536/1152\\
SU(5) matrix & 2 & 50 & 26 & 4900 & 576 & 8.51 & 1 & 115 & 450/150 & 600/300\\
SU(5) matrix & 3 & 50 & 30 & 21800 & 2344 & 9.30 & 1 & 313 & 750/450 & 1200/900\\
SU(5) matrix & 4 & 50 & 36 & 54700 & 5592 & 9.78 & 2 & 474 & 1500/900 & 2400/1800\\
\hline
\end{tabular}\end{table}

\begin{table}[htbp]\centering\small
\setlength{\tabcolsep}{4pt}
\caption{The full Trotter step: potential $+$ kinetic. The optimised kinetic layer (Sec.~\ref{sec:kinopt}) is identical in every compilation compared here, so it dilutes every ratio. NISQ columns: two-qubit gates with native $R_{ZZ}$ (the floor $m_{\min}$ of the potential term divided by our count, before and after adding the kinetic layer to both). FT columns: T gates, against the precision-matched Hamming-weight-phasing incumbent. For $\mathrm{SU}(2)$ both eras use the same cross-product entries $K^c_{12}$ on the same $6Q$ system qubits. The mini-BMN full step (potential, kinetic and fermion layers) is in Table~\ref{tab:bmnft}.}
\label{tab:fullstep}
\begin{tabular}{lrrrrrrrrr}
\hline
model & $Q$ & kin. 2q & pot. 2q & pot. & full & kin. T & pot. T & pot. & full\\ & & & ours & ratio & ratio & & ours & ratio & ratio\\
\hline
SU(2) adj. & 2 & 18 & 372 & 0.16 & 0.20 & 108 & 1500 & 0.77 & 0.79\\
SU(2) adj. & 3 & 54 & 909 & 0.35 & 0.38 & 348 & 2748 & 1.42 & 1.37\\
SU(2) adj. & 4 & 108 & 1746 & 0.58 & 0.61 & 772 & 4380 & 2.14 & 1.97\\
SU(2) adj. & 5 & 180 & 3015 & 0.84 & 0.85 & 1472 & 6396 & 2.96 & 2.59\\
SU(2) adj. & 6 & 270 & 4704 & 1.13 & 1.13 & 2300 & 8796 & 3.94 & 3.33\\
SU(2) matrix & 2 & 24 & 796 & 0.23 & 0.26 & 156 & 2952 & 1.13 & 1.12\\
SU(2) matrix & 3 & 72 & 1912 & 0.51 & 0.52 & 460 & 5096 & 1.83 & 1.76\\
SU(2) matrix & 4 & 144 & 3620 & 0.86 & 0.87 & 1020 & 7912 & 2.82 & 2.61\\
SU(3) matrix & 2 & 54 & 4245 & 0.50 & 0.50 & 252 & 11796 & 1.17 & 1.17\\
SU(3) matrix & 3 & 162 & 10152 & 1.07 & 1.07 & 900 & 20184 & 2.79 & 2.72\\
SU(3) matrix & 4 & 324 & 19071 & 1.82 & 1.81 & 2164 & 31032 & 5.20 & 4.93\\
SU(4) matrix & 2 & 96 & 11440 & 0.84 & 0.84 & 404 & 28432 & 1.63 & 1.62\\
SU(4) matrix & 3 & 288 & 31136 & 1.59 & 1.58 & 1556 & 53072 & 4.05 & 3.96\\
SU(4) matrix & 4 & 576 & 58112 & 2.71 & 2.69 & 3680 & 80720 & 8.16 & 7.85\\
SU(5) matrix & 2 & 150 & 27490 & 1.05 & 1.05 & 576 & 62180 & 2.00 & 1.99\\
SU(5) matrix & 3 & 450 & 65755 & 2.25 & 2.24 & 2344 & 104740 & 5.82 & 5.71\\
SU(5) matrix & 4 & 900 & 123080 & 3.82 & 3.80 & 5592 & 160620 & 11.91 & 11.55\\
\hline
\end{tabular}\end{table}

All compilations compared here implement the \emph{same} product formula, so
they require the same number of Trotter steps for a given accuracy and the step
count cancels from every ratio in this paper. For the absolute figures in the
surface-code estimate below we fix $10^3$ steps for every circuit alike; the
leading-order, interference-aware error estimate of
Refs.~\cite{layden2022,tran2020} is the appropriate tool for converting that
into a simulated time, and a validation of it for this model at
$\mathrm{SU}(2)$, $d=2$, $Q=2$ is included in the code package.

\subsection{Surface-code spacetime}
\label{sec:surface}
T counts are a proxy; what a machine actually spends is physical qubits
multiplied by wall-clock time. We therefore price whole Trotter steps with the
Azure Quantum Resource Estimator (the \texttt{qdk} Python package,
logical-count interface), whose model is documented in
Ref.~\cite{beverland2022}: rotated surface code with $2d^2$ physical qubits per
logical qubit and logical cycle $(4t_{\rm gate}+2t_{\rm meas})\,d$, a fixed
layout of the algorithmic logical qubits, $15$-to-$1$ distillation factories,
and a logical depth that schedules every T gate, measurement, rotation and
CCZ/CCiX sequentially. The same model, error budget ($0.01$ in total) and number
of Trotter steps ($10^3$, potential $+$ kinetic) are applied to every circuit;
only the logical counts differ:
\begin{description}
\item[Ours.] Logical qubits: the built width. Every T gate of the oracle belongs
to a logical AND, so it enters as $T/4$ CCiX gates with $T/4$
measurement-based uncomputations.
\item[Hamming-weight phasing.] System $+L+$ its ancillas (unbatched, or
width-matched); CCiX $=$ the adder-tree ANDs plus the $b-1$ ANDs of each
$b\ge4$ dyadic rotation; T $=$ the $b=3$ rotations; $b\le2$ rotations are Clifford
(Eq.~\eqref{eq:Td}).
\item[Published (secondary).] System $+L$ logical qubits and $d(d-1)N^4Q^4$
arbitrary-angle rotations per step~\cite{halimeh2026}, synthesised by the
estimator itself at the budget's precision ($20$--$25$ T each here, against
$T_{\rm typ}=10$ in the T-count tables).
\end{description}
The optimised kinetic layer is split into ANDs and T gates in the same way and
added to all three; every circuit's width is at least system $+L+$ that layer's
ancillas. For the narrow published incumbent the ancilla-free direct layer of
Eq.~\eqref{eq:kinT} is used instead whenever it gives less spacetime, which it
always does. The splits are asserted to reproduce every T-count used elsewhere.
Three of the estimator's predefined profiles are used:
\texttt{qubit\_gate\_ns\_e3} and \texttt{ns\_e4} (gate $50$\,ns, measurement
$100$\,ns, error rates $10^{-3}$ and $10^{-4}$) and \texttt{qubit\_gate\_us\_e3}
($100\,\mu$s, $10^{-3}$, T-gate error $10^{-6}$), as documented by the
estimator.

Table~\ref{tab:surface} gives the physical qubits, runtime and code distance of
our oracle and of the width-matched incumbent for \texttt{ns\_e4}, and the
spacetime ratios over all three profiles. Two things change relative to a
width$\times$T proxy. First, magic-state factories are a large part of every
footprint, so physical-qubit counts are close between circuits (often within a
few per cent) and the comparison is carried by the runtime, i.e.\ by the logical
depth. Second, the unbatched incumbent's width now costs it code distance as
well as qubits ($d=13$--$15$ against our $11$--$13$ at the large points). The
outcome: our oracle uses less spacetime than the width-matched incumbent at
every point except $\mathrm{SU}(2)$ adjoint, $Q=2$ ($0.69$--$0.73\times$):
$1.28\times$ at the $\mathrm{SU}(2)$ matrix entries, $Q=2$, $1.64\times$ at
$\mathrm{SU}(2)$, $Q=3$, and up to $25.5\times$ ($17.1$--$25.5\times$ at
$\mathrm{SU}(4)$ and $20.3$--$25.4\times$ at $\mathrm{SU}(5)$, $Q=4$). In
absolute terms, $10^3$ steps at $\mathrm{SU}(2)$, matrix entries, $Q=4$ need
$1.35\times10^5$ physical qubits for $39$\,s at $d=11$ (\texttt{ns\_e4});
$\mathrm{SU}(5)$, $Q=4$ needs $5.7\times10^5$ qubits for $864$\,s, against
$68$ and $1120$\,s with the direct kinetic layer. These are estimates within one
estimator's model; the caveats are in Discussion and Conclusions.

\begin{table}[htbp]\centering\footnotesize
\setlength{\tabcolsep}{3.5pt}
\caption{Surface-code estimate (Sec.~\ref{sec:surface}), $1000$ Trotter steps (potential $+$ kinetic), total error budget $0.01$, Azure Quantum Resource Estimator (model of Ref.~\cite{beverland2022}). Physical qubits, runtime (s) and code distance $d$ of our oracle, and of the width-matched Hamming-weight-phasing incumbent (HWP$_W$), for the \texttt{qubit\_gate\_ns\_e4} profile. Last three columns: spacetime (physical qubits $\times$ runtime) of each incumbent divided by ours, range over the three profiles \texttt{ns\_e3}, \texttt{ns\_e4}, \texttt{us\_e3}. The unbatched incumbent, whose ancilla demand is far larger than ours, is omitted here; its ratios are correspondingly larger and are in the code package. ``Published'': $d(d-1)N^4Q^4$ arbitrary-angle rotations synthesised by the estimator itself, with whichever kinetic layer --- the optimised one, or the direct ancilla-free one of Eq.~\eqref{eq:kinT} --- gives it less spacetime.}
\label{tab:surface}
\begin{tabular}{lrrrrrrrr}
\hline
model & $Q$ & phys. & time & $d$ & phys. & time & ST HWP$_W$/ & ST publ./\\ & & ours & ours & & HWP$_W$ & HWP$_W$ & ours & ours\\
\hline
SU(2) adj. & 2 & 53k & 6 & 9 & 46k & 5 & 0.69--0.73 & 5.7--9.7\\
SU(2) adj. & 3 & 85k & 11 & 9 & 85k & 18 & 1.64--1.64 & 16.3--27.1\\
SU(2) adj. & 4 & 95k & 19 & 9 & 109k & 65 & 3.47--4.02 & 43.1--58.9\\
SU(2) adj. & 5 & 124k & 35 & 11 & 123k & 147 & 4.22--6.26 & 55.2--75.9\\
SU(2) adj. & 6 & 138k & 49 & 11 & 137k & 284 & 5.78--9.02 & 86.9--108.0\\
SU(2) matrix & 2 & 89k & 11 & 9 & 89k & 14 & 1.28--1.28 & 2.6--4.7\\
SU(2) matrix & 3 & 117k & 24 & 11 & 116k & 59 & 2.40--2.65 & 6.0--12.1\\
SU(2) matrix & 4 & 135k & 39 & 11 & 133k & 165 & 4.16--6.04 & 18.1--25.1\\
SU(3) matrix & 2 & 143k & 53 & 11 & 141k & 99 & 1.85--2.41 & 2.4--4.4\\
SU(3) matrix & 3 & 176k & 93 & 11 & 207k & 465 & 5.91--6.62 & 10.5--11.6\\
SU(3) matrix & 4 & 207k & 146 & 11 & 325k & 1646 & 11.22--17.70 & 18.0--28.2\\
SU(4) matrix & 2 & 191k & 127 & 11 & 191k & 366 & 2.88--2.88 & 3.4--4.9\\
SU(4) matrix & 3 & 248k & 240 & 11 & 381k & 2085 & 8.72--13.35 & 10.2--15.6\\
SU(4) matrix & 4 & 298k & 371 & 11 & 450k & 6267 & 17.13--25.53 & 26.8--31.6\\
SU(5) matrix & 2 & 254k & 276 & 11 & 369k & 1127 & 4.04--5.92 & 4.1--6.4\\
SU(5) matrix & 3 & 358k & 471 & 11 & 492k & 5380 & 11.69--15.69 & 13.7--17.0\\
SU(5) matrix & 4 & 570k & 864 & 13 & 592k & 16846 & 20.26--25.41 & 25.0--30.0\\
\hline
\end{tabular}\end{table}

\section{Correctness certificate: complete Trotter evolutions}
\label{sec:aer}
The verification so far is layer by layer: the potential circuits as diagonal
phase maps on basis states (Secs.~\ref{sec:competitor} and~\ref{sec:nisq}), the
kinetic layer on dense unitaries (Sec.~\ref{sec:kinopt}). Here the pieces are
run together as a real-time evolution: complete Trotter circuits, kinetic layer
and every ancilla included, executed on the Qiskit Aer statevector simulator and
compared with exact diagonalisation (ED) of the same truncated Hamiltonian. The
question this section answers is whether any compilation changes the physics,
not whether it is cheaper, which at the one point reachable by statevector
simulation it is not. Such benchmarks are standard, and nothing here is claimed
as new; the set-up and the measured tables are in Appendix~\ref{app:aer}.

At $\mathrm{SU}(2)$, $d=2$, $Q=2$ (six bosons, $12$ system qubits) the potential
term is compiled three ways, the structure-aware parity network of
Sec.~\ref{sec:competitor} and our circuit in both emissions, each combined
with the optimised kinetic layer, and each run as a complete quench evolution
from a product of single-boson ground states. The reference Hamiltonian is built
by an independent code path, from the Gram form~\eqref{eq:gram} rather than from
the bilinears the circuits compile, and the two paths meet in one identity
checked on all $4096$ basis states.

\begin{figure}[tbp]\centering
\includegraphics[width=\textwidth]{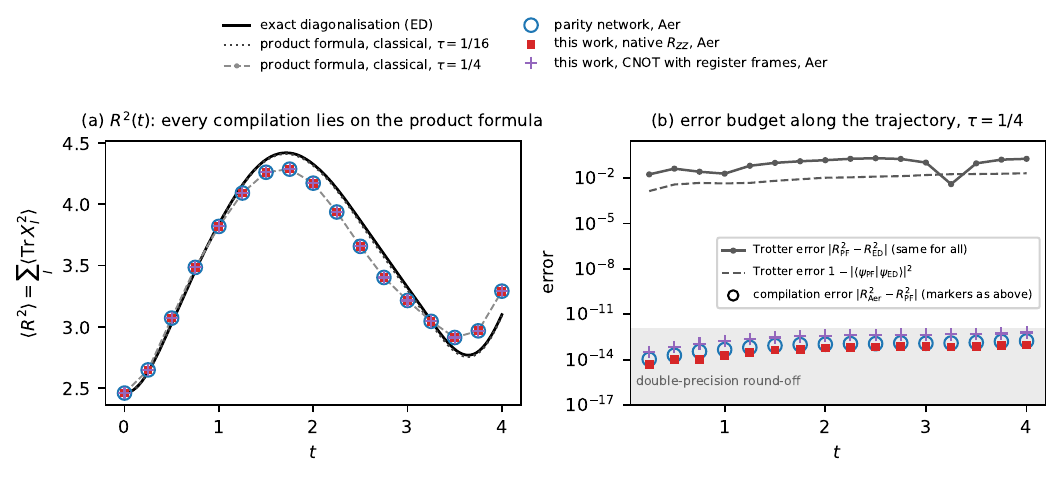}
\caption{Full-Trotter dynamics of the $\mathrm{SU}(2)$, $d=2$ matrix model at
$Q=2$ on the Qiskit Aer statevector simulator ($g=1$,
$\delta=\sqrt{2\pi/\Lambda}$; quench from the ground state of
$\sum_a\frac12(p_a^2+4x_a^2)$). Each compiled circuit, the structure-aware
parity network and our circuit in both emissions, each with the optimised
kinetic layer and all ancillas, is run as a complete time evolution and
compared with exact diagonalisation (ED) of the same truncated Hamiltonian,
built by an independent code path. (a)~$\langle
R^2\rangle=\sum_I\langle\mathrm{Tr}\,X_I^2\rangle$: the three circuits coincide
with the classical product formula at $\tau=1/4$, and the visible gap to ED is
Trotter error, which closes at $\tau=1/16$ (dotted). (b)~The error budget: the
compilation error sits about twelve orders of magnitude below the Trotter error,
which is identical for every compilation.}
\label{fig:aer}
\end{figure}

The outcome (Fig.~\ref{fig:aer}, Tables~\ref{tab:aer} and~\ref{tab:aerconv}) is
that the compilation error is at double-precision round-off and the only
departure from exact diagonalisation is Trotter error. For all three circuits
$\langle R^2\rangle$ equals the product formula to $6.3\times10^{-13}$ along the
trajectory and to $1.9\times10^{-12}$ in the convergence runs, $\langle
H\rangle$ to $6.3\times10^{-13}$, and the state infidelity with the
product-formula state is at most $6.0\times10^{-13}$; after every step every
ancilla is back in $|0\rangle$ to within $2.4\times10^{-14}$, and the three
trajectories agree with one another to $5.3\times10^{-13}$. The Trotter error at
$\tau=1/4$ is at most $0.21$ in $\langle R^2\rangle$ and $2.1\times10^{-2}$ in
infidelity, identical for every compilation and about twelve orders of magnitude
above the compilation error, and it converges with the expected orders $2$ (in
$\langle R^2\rangle$ and the state distance) and $4$ (in infidelity). The
compiled circuits implement the product formula exactly; the parity network and
both emissions of our circuit yield the same states and therefore the same
physics.

\paragraph{What this does not show.} It does not show the gate advantage. At
$Q=2$ our circuit costs $390$ two-qubit gates per merged step with native
$R_{ZZ}$ and $570$ with CNOT only, against $108$ and $126$ for the parity
network: it is $3.6$--$4.5\times$ \emph{more} expensive here, as
Table~\ref{tab:crossover} records. The advantage begins at $Q=6$ (native
$R_{ZZ}$) and $Q=7$ (CNOT only), that is on $8Q+2=50$ and $58$ qubits, beyond
statevector simulation. There the claim rests on two facts established above.
First, the potential step acts diagonally on the system and returns its ancillas
to $|0\rangle$, so it is fixed by its phases on system basis states; these equal
$-\theta V_{\rm int}$ up to a constant for our circuit on random basis states up
to $Q=8$ (Sec.~\ref{sec:nisqcorrect}) and for the network up to $Q=10$
(Sec.~\ref{sec:competitor}), exhaustively at $Q=2$. Second, the kinetic layer is
the same circuit in every compilation. Random sampling is evidence, not proof,
of equality on all basis states; exhaustive checks are possible only at $Q=2$.

\section{Supersymmetric mini-BMN with fermions}
\label{sec:bmn}
Both constructions were built for a quartic commutator potential. The obvious
objection is that they work because that potential happens to be a pure
commutator square. The supersymmetric mini-BMN model is the direct test: it adds
a cubic Myers term, a mass term and adjoint fermions, and the compilation unit
survives all three.

The model~\cite{asplund2016}, in the conventions of
Ref.~\cite{hanhartnoll2020}, has three bosonic matrices and one two-component
adjoint fermion,
\begin{equation}
  H=\mathrm{tr}\Big(\tfrac12\Pi^i\Pi^i-\tfrac14[X^i,X^j]^2+\tfrac12\nu^2X^iX^i
  +i\nu\epsilon_{ijk}X^iX^jX^k\Big)
  +\mathrm{tr}\Big(\lambda^\dagger\sigma^k[X^k,\lambda]+\tfrac32\nu\lambda^\dagger\lambda\Big)
  -\tfrac32\nu(N^2-1),
  \label{eq:bmn}
\end{equation}
which we take at $N=2$ with $T_A=\sigma_A/\sqrt2$: nine bosons $X^i_A$ and six
fermion modes, mapped by Jordan--Wigner. (This is not the two-matrix ``minimal
BMN'' model of Ref.~\cite{rinaldi2022}; it has the cubic Myers term and four
supercharges.) Han and Hartnoll note that its potential is a total
square~\cite{hanhartnoll2020}: $V=\tfrac12\sum_{i,A}(E^i_A)^2$ with
$E^i_A=\nu X^i_A-\tfrac12\epsilon_{ijk}f_{ABC}X^j_BX^k_C$. On the
grid~\eqref{eq:grid} this becomes
\begin{equation}
  V=\frac{\delta^4}{16}\sum_{i,A}\big(K^A_{jk}-\kappa\,m^i_A\big)^2,\qquad
  K^A_{jk}=m^j_Bm^k_C-m^j_Cm^k_B,\qquad \kappa=\frac{\sqrt2\,\nu}{\delta},
  \label{eq:bmnsos}
\end{equation}
with $(i,j,k)$ and $(A,B,C)$ cyclic. \textbf{The cubic Myers term is the cross
term and the mass term is the square of the linear part: both live inside the
same nine squares as the quartic term, so the compilation unit is unchanged and
only the bilinear becomes affine.} On the NISQ side an affine term is one extra
fold-wire parity per gadget, visited after the product blocks at one CNOT per
move, with S1--S5 unchanged; fault-tolerantly it is one signed shifted addition
of $|m_p|$ into the entry register before squaring. A dyadic $\kappa=K_\nu/2^p$
keeps $V_{\rm int}$ an exact integer, which a choice of grid spacing secures;
where it cannot be chosen, a split variant --- our gadget on $\sum_eK_e^2$ and
the parity network on the $\kappa$-dependent remainder --- is valid for any real
$\kappa$ and is the cheaper of the two at larger $Q$ (Table~\ref{tab:kappa}).

The fermion layer is identical in every compilation, so it cannot change a
ranking; it is built, verified against the exact exponential for all nine
Yukawa groups ($\le4.4\times10^{-16}$) and costed in both eras in
Appendix~\ref{app:bmn}, together with the potential-term counts. All identities
above are verified: the matrix, component and grid forms of $V$ agree at $200$
random configurations to $1.8\times10^{-12}$; $V=0$ at the fuzzy sphere
$X^i=\nu J^i$; and the one-body fermion spectrum there is
$\nu\{\tfrac52\ (\times4),-\tfrac12\ (\times2)\}$, the $N=2$ case of the
spectrum of Ref.~\cite{hanhartnoll2020}, which fixes the relative sign of the
Yukawa and mass terms.

The result is in Table~\ref{tab:bmnstep} and Fig.~\ref{fig:scaling}(b), at
$\kappa=1$. \textbf{The NISQ crossovers are exactly those of the bosonic
$\mathrm{SU}(2)$, $d=2$ model}, $Q=6$ with native $R_{ZZ}$ and $Q=7$ with
CNOT only, and the full step reaches $1.36\times$ and $1.15\times$ at $Q=7$.
Fault-tolerantly the full step reaches $2.60\times$ against unbatched and
$4.78\times$ against width-matched Hamming-weight phasing at $Q=6$, and
$5.90\times$ in surface-code spacetime; here the crossover against the
width-matched incumbent is again $Q=3$, as in the bosonic model, while against
the unbatched one it moves by one step, to $Q=4$
(Table~\ref{tab:bmnft}). The supersymmetric deformation therefore costs at most one step
in $Q$, and nothing at all on the NISQ side, which is the generality claim this
section is for. $\mathrm{SU}(2)$ is also the least favourable $N$
(Sec.~\ref{sec:scaling}), so these are the weakest margins the method produces.

\begin{table}[htbp]\centering\footnotesize
\setlength{\tabcolsep}{3.5pt}
\caption{Supersymmetric mini-BMN, $\mathrm{SU}(2)$, $\kappa=1$ (Sec.~\ref{sec:bmn}): the full Trotter step in both eras. NISQ (two-qubit gates): $m_{\min}$ is the floor of every parity network on the full bosonic potential (quartic, Myers and mass terms); ``ours'' is the better of the affine gadget and the split circuit in each gate set; ``pot.'' is $m_{\min}$/ours for the potential term alone --- the quantity Table~\ref{tab:crossover} reports --- and ``step'' adds the identical kinetic and fermion layers to both sides. FT (T gates): our affine sum-of-squares oracle against precision-matched Hamming-weight phasing, unbatched (HWP) and width-matched (HWP$_W$), again with the identical common layers; ST is the surface-code spacetime ratio against HWP$_W$ ($10^3$ steps, profile \texttt{ns\_e4}). Every crossover coincides with that of the bosonic $\mathrm{SU}(2)$ model. Potential-term counts, the fermion layer and the $\kappa$ dependence are in Tables~\ref{tab:bmnnisq}--\ref{tab:kappa}.}
\label{tab:bmnstep}
\begin{tabular}{rrrrrrrrrrrrr}
\hline
$Q$ & $m_{\min}$ & ours & ours & pot. & pot. & step & step & FT & FT & FT step & FT step & ST\\ & (NISQ) & RZZ & CNOT & RZZ & CNOT & RZZ & CNOT & ours & HWP$_W$ & HWP & HWP$_W$ & HWP$_W$\\
\hline
2 & 219 & 1190 & 1676 & 0.18 & 0.13 & 0.33 & 0.29 & 4716 & 2388 & 0.61 & 0.65 & 0.65\\
3 & 1080 & 2931 & 3822 & 0.37 & 0.28 & 0.44 & 0.37 & 8640 & 11288 & 0.97 & 1.21 & 1.21\\
4 & 3390 & 5710 & 7222 & 0.59 & 0.47 & 0.62 & 0.52 & 13680 & 33628 & 1.40 & 2.02 & 2.02\\
5 & 8265 & 9943 & 12112 & 0.83 & 0.68 & 0.84 & 0.71 & 19908 & 81184 & 1.95 & 3.15 & 3.78\\
6 & 17145 & 15630 & 18744 & 1.10 & 0.91 & 1.09 & 0.92 & 27252 & 175100 & 2.60 & 4.78 & 5.90\\
7 & 31794 & 23166 & 27448 & 1.37 & 1.16 & 1.36 & 1.15 & -- & -- & -- & -- & --\\
\hline
\end{tabular}\end{table}

One limitation is specific to this model and is not hidden. For $N\ge3$ the
fermion layer is not built. With Jordan--Wigner the adjoint fermions give strings
of length $O(N^2)$ and $O(N^3)$ Yukawa terms, hence $O(N^5Q)$ CNOTs, as noted in
Ref.~\cite{halimeh2025ferm}, which would outgrow our $O(N^3Q^3)$ potential at
large $N$. That is a statement about the fermion encoding, not about the
potential compilation this paper is concerned with, and fermionic swap
networks~\cite{kivlichan2018} are the obvious remedy; they are not evaluated
here.

\section{Discussion}
\label{sec:limits}

The binding constraint on quantum simulation of matrix models is circuit cost:
coordinate-basis compilation costs $O(d^2N^4Q^4)$ rotations per Trotter
step~\cite{halimeh2026,hanada2025}, and already at $\gtrsim10^3$ two-qubit gates
the first hardware simulation of a single-matrix $\mathrm{SU}(2)$ model was
noise-dominated~\cite{hartnett2026}. We asked how far that can be pushed down,
in each era, without changing the Hamiltonian. The answer has three parts: the
standard compilation is almost optimal in its own gate class; an exact
factorisation of the potential gives a cheap construction outside it, once per
era; and it is the factorisation, not the arithmetic, that pays.

The bound came first, because it decides what an improvement can mean here and
because a resource claim is only as strong as the incumbent it is measured
against. For a generic Trotter step the Walsh support $m_{\min}$ bounds the
rotation count of every $\{\mathrm{CNOT},X,R_z\}$ circuit and the two-qubit
count of every
$\{\mathrm{CNOT},\mathrm{CZ},\mathrm{CPhase},R_{ZZ},\mathrm{SWAP},X,R_z\}$
circuit, ancillas included (Lemmas~\ref{lem:anc} and~\ref{lem:twoq}); and
because an even-degree potential on the centred grid produces only even-weight
characters, the weight-one escape route is closed and the floor is exactly
$m_{\min}$ (Proposition~\ref{prop:even}). We compute that support exactly, give it in
closed form for every $N$ (Eqs.~\eqref{eq:mminsu2} and~\eqref{eq:mminN}), and
bound it below by $Q^4$ per distinct quartic vertex (Theorem~\ref{thm:lb}); and
since a floor is not a circuit, we built a structure-aware parity network
attaining $1.009$--$1.096\,m_{\min}$ for $Q\ge4$.
Two things follow. No ordering heuristic, ZX rewrite, synthesis strategy or
ancilla allocation can improve the published compilation, so improvement
requires a gate outside the class (Corollary~\ref{cor:wall}). And the comparison
is fixed on terms that cannot be gamed: our NISQ circuit is measured against a
quantity no parity network can go below, in the gate set in which it is emitted,
and since the floor covers both gate sets alike, the like-for-like rule is
enforced rather than promised.

What leaves the gate class cheaply is an algebraic statement before it is a
circuit. Hermiticity alone gives $V=\tfrac{g^2}{2}\sum_{I<J}\|[X_I,X_J]\|_F^2$
(Eq.~\eqref{eq:sos}); read as a compilation instruction, it says that the object
to evaluate is not the $\Theta(N^4)$ quartic monomials of the expanded potential
but the $\Theta(N^2)$ commutator entries, each a bilinear with $\Theta(N)$ terms
squared exactly once. The object count therefore grows one power of
$N$ more slowly than the number of Walsh characters a phase polynomial has to
phase, and Eq.~\eqref{eq:growth} makes this exact: the ratio of quartic
monomials to product terms is $(12N^2-28N+17)/[4(2N-1)]\sim3N/2$. It also
flattens the dependence on the grid resolution, but by a different mechanism in
each era, which is why the two constructions are different circuits and not one
circuit priced with two cost models. On
NISQ hardware the squared entry is added into a Fourier register and squared
there as a quadratic form, so the $\sim Q^4$ characters of a squared bilinear
become $\sim Q^2$ register rotations with no arithmetic at all; fault-tolerantly
each product is a single $Q\times Q$ multiplication and each entry one squaring
and one injection into a catalytic phase-gradient register. Every margin
reported here reduces to the scalings that follow, $\Theta(d^2N^3Q^3)$ two-qubit
gates and $\Theta(d^2N^3Q^2)$ T gates against a floor of $\Theta(d^2N^4Q^4)$
(Eq.~\eqref{eq:leading}).

Which primitive to leave the class with is decided by hardware, and the two eras
decide it oppositely. A logical AND costs four T gates fault-tolerantly but
about six CNOTs on NISQ hardware, whereas a dyadic rotation costs one native
$R_{ZZ}$ on NISQ hardware but up to $4(L-1)=68$--$140$ T gates at the
phase-gradient register widths $L$ used here: arithmetic is cheap exactly where
rotations are expensive, and expensive where they are cheap. This is not a plausibility
argument: the Toffoli-arithmetic route was built and charged with NISQ costs,
where it needs $15$--$6.4\times$ the two-qubit gates of the Fourier route
(Appendix~\ref{app:negative}). The two exits are not interchangeable, which is
why this paper carries two constructions.

The measured outcome is in Table~\ref{tab:crossover}, with four qualifications.
The NISQ advantage begins at $4.7\times10^3$ two-qubit gates per step with a
native $R_{ZZ}$ and $8.3\times10^3$ with CNOT only, beyond what present hardware
can run; below the crossover the parity network is cheaper. The counts are exact
on all-to-all hardware, where the depth crosses a step earlier than the count,
but on heavy-hex and square lattices every crossover moves one step later in $Q$
(Sec.~\ref{sec:routing}). The fault-tolerant incumbent is built and verified,
but its batching is one simple choice and the dyadic ladder may admit others, so
we do not claim it optimal. And the oracle carries $52$ to $475$ ancillas
against $12$ to $200$ system qubits, which is why every width claim is made
against a width-matched incumbent and checked in surface-code
spacetime~\cite{beverland2022}.

Three controls keep the central claim honest. The same arithmetic applied to the
quartic monomials instead of the squares costs $1.87$--$20.7\times$ more and
loses to the incumbent at $15$ of $17$ points, while the kinetic term, already
optimal as a phase polynomial, is $5.4$--$9.5\times$ \emph{more} expensive
compiled arithmetically; the mini-BMN potential is again a single sum of squares
with the Myers and mass terms inside it~\cite{asplund2016,hanhartnoll2020}, so
supersymmetry costs at most one step in $Q$; and complete Trotter evolutions on
a statevector simulator put the compilation error twelve orders of magnitude
below the Trotter error. Since
every margin grows with $N$ and with $Q$, and the holographic regime demands
large $N$ and with it a large grid per boson~\cite{hanada2025,hartnett2026}, the
compilation is cheapest exactly where the physics is, without removing the
$O(N^2)$ register growth, which remains the binding constraint.

\section{Conclusions and Outlook}
\label{sec:conclusions}

We have shown that coordinate-basis compilation of the $\mathrm{SU}(N)$
matrix-model potential is already essentially optimal within the gate class it
uses, and that leaving that class buys a large and growing saving in both
hardware eras. The optimality statement is a proved two-qubit lower bound, the
exactly computed Walsh support that sets it in closed form at every $N$, and a
parity network built to within $1$--$10\%$ of that floor. Both constructions that leave the
class rest on the same exact factorisation of the potential into a sum of
squared commutator entries, which replaces $\Theta(N^4)$ quartic monomials by
$\Theta(N^2)$ objects: a Fourier-register squaring giving $\Theta(d^2N^3Q^3)$
two-qubit gates, and reversible arithmetic with one phase-gradient injection per
entry~\cite{gidney2018} giving $\Theta(d^2N^3Q^2)$ T gates with no
arbitrary-angle rotation and hence no synthesis error; both are built, verified,
and equal to closed-form resource formulas at every point. Against the strongest
incumbents we could construct they use fewer two-qubit gates than \emph{every}
parity network from six qubits per boson at $\mathrm{SU}(2)$ and two at
$\mathrm{SU}(5)$, by up to $3.8\times$, and up to $11.9\times$ fewer T gates
and, at equal width, $25.5\times$ less surface-code
spacetime~\cite{beverland2022}, with both margins growing linearly in $N$ and,
in the grid resolution, like $Q$ and $Q^2$. Two controls identify the
factorisation rather than the arithmetic as the source of the gain, the
construction survives the supersymmetric deformation to mini-BMN with fermions,
and complete Trotter evolutions reproduce the exact product formula to
$2\times10^{-12}$.

Several things are deliberately left open. The circuits built here cover the
bosonic potential with two matrices and the supersymmetric model only at
$\mathrm{SU}(2)$: the fermion layer at larger $N$ and the $d=9$,
sixteen-supercharge BFSS and BMN models remain to be constructed, and a
Jordan--Wigner adjoint Yukawa layer would eventually outgrow the
potential~\cite{halimeh2025ferm}. The routed and surface-code figures rest on
one heuristic router and one estimator model~\cite{beverland2022}, and the
end-to-end evolutions are noiseless and at the smallest truncation. None of this
bears on the ranking, which is carried by the potential term and by layers
identical in each compilation.

The two constructions are not specific to this potential: what they use is that
the interaction is a sum of squares of low-degree forms in the grid coordinates,
which holds for any commutator-square potential and survives supersymmetric
deformation. Orbifold-lattice Yang--Mills, which uses the same encoding and a
potential again built from commutators~\cite{halimeh2026,bergner2026}, is the
obvious next target, as are the $d=9$ BFSS and BMN models once the fermion layer
is put on a swap network~\cite{kivlichan2018}. The two routes point at different
machines: the Fourier construction wants a native $ZZ$ interaction and
all-to-all connectivity, as trapped ions provide, while the arithmetic route
wants cheap Toffolis and a catalytic register, as early fault tolerance will.
Between them they lower the cost of a matrix-model Trotter step by a
factor that grows with the size of the matrix and with the resolution of the
grid: the two parameters that must both be increased to reach the black-hole
regime, and the two in which every crossover reported here is a floor rather
than a ceiling.

\section*{Acknowledgement}
The authors would like to thank Prof. Indrakshi Raychowdhury (BITS Pilani, Goa) for endorsing the submission.

\section*{Author Contributions}
A.D. conceptualized the study, designed the quantum circuits, the NISQ Fourier-register and fault-tolerant sum-of-squares constructions, developed the mini-BMN extension, and performed the complete Trotter-evolution and circuit-level validation. N.S. constructed the phase-polynomial lower bound, derived the exact Walsh supports, carried out the ablation studies and resource computations, and provided mathematical support.

\section*{Data and code availability}
The data that support the findings of this study are openly available in the
software package accompanying this work~\cite{codepackage}, archived at
\url{https://doi.org/10.5281/zenodo.22998235}. The package contains every circuit
construction reported here, the verification suite, the stored gate counts and logs,
and the scripts that regenerate each table and figure in this paper directly from
the built circuits. It is released under the MIT licence, so the community is free
to use, modify and build upon it with attribution. The repository is maintained on GitHub at \url{https://github.com/ayanendudutta/matrix-model-sos-compilation} for ongoing maintenance and contributions.


\appendix

\section{Lower-bound proofs and exact supports}
\label{app:proofs}

\begin{proof}[Proof of Lemma~\ref{lem:anc}]
CNOT and $X$ generate the affine group, so each wire holds
$\ell_q(x,0)\oplus\varepsilon_q$ with $\ell_q$ linear; ancilla inputs are $0$,
so $\ell_q$ depends on the data bits only and an $R_z$ contributes one character
of $x$. The realised phase $f$ satisfies $f\equiv\theta P+c\pmod{2\pi}$, so
$\hat f(y)-\theta\hat P(y)\in2\pi\,2^{-n}\mathbb{Z}$ for $y\ne0$; for generic
$\theta$ this forces $\hat f(y)\ne0$ on the support. Attainability: one
rotation per character in any parity network.
\end{proof}

\begin{proof}[Proof of Lemma~\ref{lem:twoq}]
Every gate of $\mathcal G$ maps basis states to basis states up to phase, so
each wire holds an affine form as above. Let $\mathcal P$ be the set of linear
forms ever held by some wire. Initially $\mathcal P=\{e_1,\dots,e_n,0\}$. A CNOT
changes the form on one wire, adding at most one element to $\mathcal P$; $X$
and SWAP add none; diagonal gates add none. An $R_z$ realises a character in
$\mathcal P$; a CZ, CPhase or $R_{ZZ}$ on wires $a,b$ realises characters
$\ell_a$, $\ell_b$ (both in $\mathcal P$) and $\ell_a\oplus\ell_b$ (at most one
new character). Hence the realised non-constant characters lie in a set of size
at most $m_1+\#\mathrm{CNOT}+\#\{\mathrm{CZ},\mathrm{CPhase},R_{ZZ}\}$, and by
the genericity argument of Lemma~\ref{lem:anc} this set contains the support.
\end{proof}

\begin{proof}[Proof of Proposition~\ref{prop:even}]
By~\eqref{eq:mZ}, $m_a^{\,n}=(-\sum_j2^{\,j}Z_{a,j})^n$ expands into products of
$n$ single-qubit $Z$'s, and $Z^2=I$, so every character it produces has weight
$\equiv n\pmod 2$. A monomial $\prod_am_a^{\,n_a}$ therefore produces only
characters of weight $\equiv\sum_an_a\pmod 2$, i.e.\ the parity of its total
degree; cancellations between monomials only remove characters. $V$ is quartic
and the kinetic term quadratic, so all their characters have even weight.
\end{proof}

\begin{proof}[Proof of Theorem~\ref{thm:lb}]
By~\eqref{eq:mZ} the monomial expands as
$(-1)^\kappa c\sum2^{\,j_1+\dots+j_\kappa}Z_{a_1j_1}\cdots Z_{a_\kappa j_\kappa}$;
distinct bosons give distinct qubits, no $Z^2=I$ reduction, and a character
names the bosons it touches, so different multisets cannot cancel.
\end{proof}

\begin{table}[htbp]\centering\small
\caption{Exact Walsh support $m_{\min}$ of $V$ (adjoint coordinates) against the lower bound of Theorem~\ref{thm:lb}. Every entry is computed, not estimated. No support contains a weight-one character, so by Lemma~\ref{lem:twoq} $m_{\min}$ is also a lower bound on the two-qubit gate count of every $\{\mathrm{CNOT},\mathrm{CZ},R_{ZZ},X,R_z\}$ circuit.}
\label{tab:support}
\begin{tabular}{llrrrr}
\hline
model & $Q$ & quadruples & bound $Q^4\,\#$ & $m_{\min}$ & ratio\\
\hline
SU(2) $d=2$ & 2 & 3 & 48 & 60 & 1.25\\
SU(2) $d=2$ & 4 & 3 & 768 & 1020 & 1.33\\
SU(2) $d=2$ & 6 & 3 & 3888 & 5328 & 1.37\\
SU(2) $d=3$ & 2 & 9 & 144 & 171 & 1.19\\
SU(2) $d=3$ & 4 & 9 & 2304 & 3006 & 1.30\\
SU(2) $d=3$ & 6 & 9 & 11664 & 15849 & 1.36\\
SU(3) $d=2$ & 2 & 85 & 1360 & 1458 & 1.07\\
SU(3) $d=2$ & 4 & 85 & 21760 & 24424 & 1.12\\
SU(3) $d=2$ & 6 & 85 & 110160 & 126012 & 1.14\\
SU(3) $d=3$ & 2 & 255 & 4080 & 4350 & 1.07\\
SU(3) $d=3$ & 4 & 255 & 65280 & 73128 & 1.12\\
SU(3) $d=3$ & 6 & 255 & 330480 & 377628 & 1.14\\
\hline
\end{tabular}\end{table}

\begin{figure}[htbp]
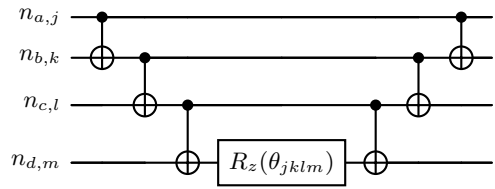
\centering
\figbaselinewfour
\caption{The incumbent's elementary unit as drawn in Ref.~\cite{halimeh2026}:
one weight-four character $Z_{a,j}Z_{b,k}Z_{c,l}Z_{d,m}$ realised as a CNOT
ladder, one $R_z(\theta_{jklm})$, and the unfolding ladder. Theorem~\ref{thm:lb}
says a quartic vertex in four distinct bosons needs $Q^4$ characters;
Fig.~\ref{fig:competitor} shows how a good parity network visits them at about
one CNOT each.}
\label{fig:baselinew4}
\end{figure}

\begin{figure}[htbp]
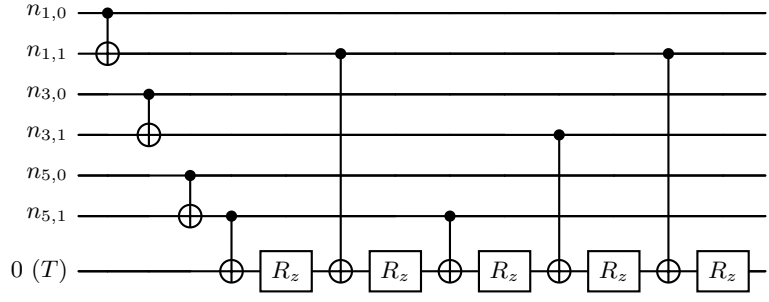
\centering
\figcompetitor
\caption{The first steps of the structure-aware parity network
($\mathrm{SU}(2)$, $d=2$, $Q=2$), transcribed from the built circuit: the
prefix-difference basis change on the bosons touched, then the tour on the
target $T$ --- one CNOT and one $R_z$ per character. Wire labels $n_{a,j}$:
boson $a$, bit $j$.}
\label{fig:competitor}
\end{figure}

\section{Circuit constructions, bookkeeping and primitive costs}
\label{app:circuits}

\subsection{The remaining NISQ simplifications}
S1 and S5, which move the crossover, are in Sec.~\ref{sec:simpl}; the other
three are exact rewrites of the same gadget and are recorded here.
\begin{description}
\item[S2, cancelling system phases.] The factors $e^{i\alpha_lP_t/2}$ of step 2
are functions of the system bits only. Every gate of the gadget either leaves
the system bits unchanged or is diagonal, so these factors commute with the
whole gadget and cancel between $\mathcal A_+$ and $\mathcal A_-$; they are
dropped.
\item[S3, fold tour.] System registers are in the prefix-difference basis of
Sec.~\ref{sec:competitor} (once per potential step); terms are visited in a
boustrophedon over $(j,k)$ in each $(p,q)$ block, so each move costs one CNOT
into $F$. The subtraction $\mathcal A_-$ retraces the tour, so $F$ is never
cleared between the passes.
\item[S4, halving.] Every $m_pm_q$ is odd, so $v\equiv\sum c_{pq}\pmod2$. When
$\sum c_{pq}$ is even the register stores $v/2$ ($r=1$): one qubit fewer and
smaller QFTs and square. For $\mathrm{SU}(2)$, $K^c$ is a difference of two odd
products and is always even, which is why $w=2Q+1$ in
Eq.~\eqref{eq:nisqcount}.
\end{description}

\subsection{The remaining fault-tolerant bookkeeping}
F1 is in Sec.~\ref{sec:ftconstr}.
\begin{description}
\item[F2, signed accumulation by carry-in.] $C\mathrel{+}=(-1)^sP$ is done as
$C\mathrel{+}=(P\oplus s\cdots s)+s$: CNOTs from the sign bit, then one Cuccaro
addition whose borrowed carry qubit \emph{is} the sign bit. The AND count stays
$n-1$, and a per-product conditional negation disappears.
\item[F3, merged products and per-entry widths.] Products sharing a boson pair
are merged (so $|c_{pq}|=2$ occurs, on the $\mathrm{Im}_{ii}$ entries, and is
added once into $C[1{:}]$), and each entry register has its own width
$w_e=\lceil\log_2(\sum|c_{pq}|m_{\max}^2+1)\rceil+1$.
\end{description}

\subsection{The kinetic-layer rewrites}
\label{app:kinopt}
In application order the direct layer is, per boson,
$\Phi\,S\,\mathcal L^\dagger\,\Phi\,P\,\Phi^\dagger\,\mathcal L\,S\,\Phi^\dagger$,
with $\Phi$ the centring ramp, $S$ the bit-reversal swaps, $\mathcal L$ the
Hadamard/CPhase ladder and $P=e^{-2\pi iA_{\rm kin}\tilde m^2/2^L}$. None of the
four rewrites changes the simulated Hamiltonian.
\begin{description}
\item[K1.] $\Phi P\Phi^\dagger=P$: all three are diagonal in the momentum basis.
\item[K2.] The outer ramps telescope over the Trotter product: $\Phi^\dagger$ of
one step and $\Phi$ of the next meet across $e^{-i\tau V}$, which is diagonal in
the same (position) basis, so $\Phi^\dagger e^{-i\tau V}\Phi=e^{-i\tau V}$. One
$\Phi$ remains at the start of the evolution and one $\Phi^\dagger$ at the end,
a phase for a position-basis measurement; the per-step cost is zero. For
mini-BMN the fermion layer is block-diagonal in the boson position basis (the
bosons only control phases), so the ramps telescope through it too.
\item[K3.] $SXS$ equals $X$ with the wires relabelled $j\mapsto Q-1-j$, so the
swaps are removed.
\item[K4.] (Fault tolerance.) Every rotation of the layer has the same angle on
every boson. The phase polynomial of all bosons, whose characters $Z_jZ_k$ carry
the weight $2^{\,j+k}$, is one weighted Hamming-weight tree over the parity bits
followed by \emph{one} subtraction of the integer
$y'=\sum_a\sum_{j<k}2^{\,j+k-1}(b_{a,j}\oplus b_{a,k})$ from the gradient
register ($\tilde m^2=(\Lambda-1)^2-8y'$). Each QFT CPhase $(j,k)$ of separation
$m$ is a class of $3n_b$ equal angles, via
$e^{i\theta ab}=e^{i\frac\theta2(a+b+\overline{a\oplus b}-1)}$, which is phased
by one tree and one addition into the top $m+2$ gradient bits whenever that is
cheaper than the gates (for $m\ge2$). Bosons are batched so that the layer's
ancillas stay within our oracle's, which are free while it runs.
\end{description}

\subsection{Primitive costs}
\paragraph{Windowed multiplier.} For $\mathrm{out}\mathrel{+}=a\cdot b$ with
$\mathrm{out}$ initially zero the window $\mathrm{out}[j:j{+}w_a{+}1]$ holds a
value $<2^{\,w_a}$ before row $j$, so no carry escapes and
\begin{equation}
  T_{\rm mul}(w_a,w_b,w_o)=\!\!\sum_{j=0}^{\min(w_b,w_o)-1}\!\!
  \big[4k_j+T_{\rm add}(v_j)\big],
  \label{eq:Tmul}
\end{equation}
with $h_j=\min(j{+}w_a{+}1,w_o)$, $v_j=h_j-j$, $k_j=\min(w_a,v_j)$.

\paragraph{Squarer.} Using $a_i^2=a_i$,
\begin{equation}
  T_{\rm sq}(w,w_o)=\sum_{j=1}^{w-1}\big[4k_j+T_{\rm add}(v_j)\big],
  \label{eq:Tsq}
\end{equation}
with $s_j=j{+}1$, $h_j=\min(s_j{+}j{+}2,w_o)$, $v_j=h_j-s_j$,
$k_j=\min(j,v_j)$.

\subsection{Transcribed circuits}
Every circuit below is transcribed gate by gate from the built op list, not
redrawn by hand.

\begin{figure}[htbp]\centering
\figabsm
\caption{The free absolute value, $Q=3$. The top three wires hold $n$; the
fourth is a scratch qubit set to $\lnot\mathrm{MSB}(n)$ by the leading CNOT--$X$
pair and restored by the trailing one; the bottom three become $|m|$. The $X$ on
$|m|_0$ writes the always-one low bit of the odd integer $|m|$; each remaining
bit is $n_j$ copied up one position and conditionally complemented by the
scratch qubit. Six CNOTs, three $X$, zero AND gates.}
\label{fig:absm}
\end{figure}

\begin{figure}[htbp]
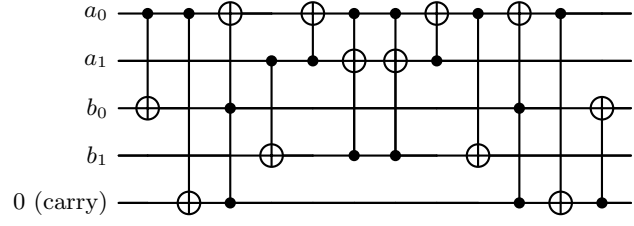
\centering
\figadder
\caption{The $n=2$ Cuccaro ripple-carry adder $b\mathrel{+}=a$. The forward MAJ
sweep threads the carry through the high bits of $a$; the backward UMA sweep
writes the sum into $b$ and restores $a$ and the borrowed carry ancilla. Four
Toffolis ($2n$); the Gidney model charges this block $T_{\rm add}(n)=4(n-1)$.}
\label{fig:adder}
\end{figure}

\begin{figure}[htbp]
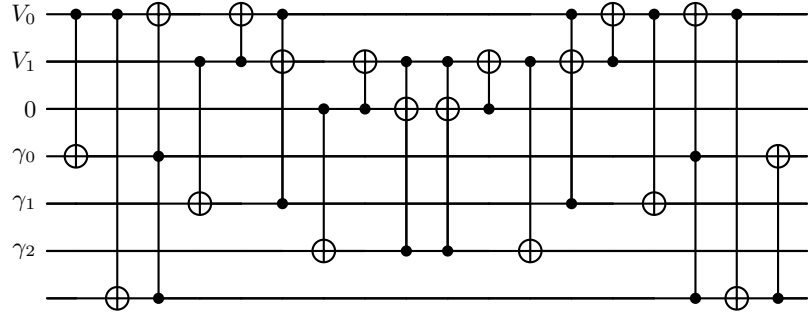
\centering
\figinject
\caption{Phase-gradient injection: one addition of a value register
(zero-extended) into the catalytic register $|\gamma_L\rangle$ imprints
$e^{-2\pi iV/2^L}$ and returns $|\gamma_L\rangle$ unchanged. Here the value has
two bits and $L=3$; in the real circuits $L=18$--$36$, and the value is
$\mu_eC_e^2$.}
\label{fig:inject}
\end{figure}

\begin{figure}[htbp]
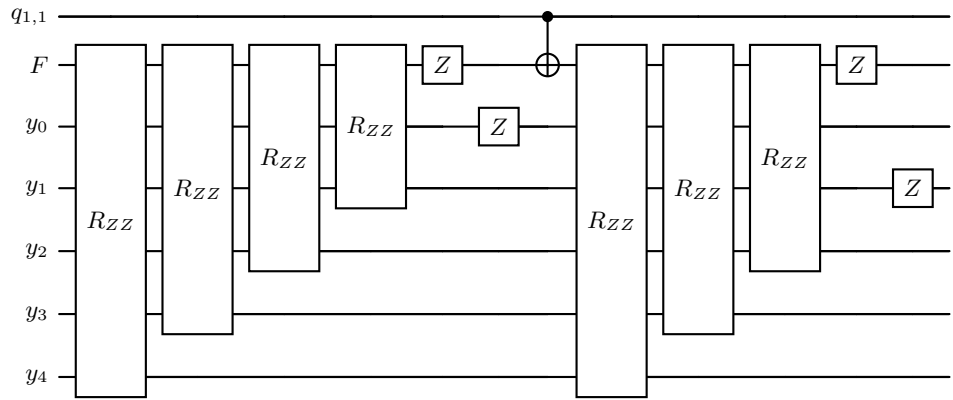
\centering
\figweightthree
\caption{Two consecutive terms of the Fourier addition, from the built
$\mathrm{SU}(2)$ gadget ($Q=2$, component $c=1$, register holds $K/2$, $w=5$):
the terms $(j,k)=(0,1)$ and $(1,1)$ of one block, with the one-CNOT fold move
between them (control $q_{1,1}$, boson $1$ in the prefix-difference basis; boson
index $1$ counted from $0$). Each $R_{ZZ}$ box acts on $F$ and on the register
wire at its lower edge; wires strictly inside a box pass through it. Draper
triangularity (S1): four and then three $R_{ZZ}$ levels, the next level a pair of
single-qubit $Z$ gates, the rest the identity. Without S1 every term would carry
$w$ levels plus a weight-two rotation.}
\label{fig:weight3}
\end{figure}

\begin{figure}[htbp]
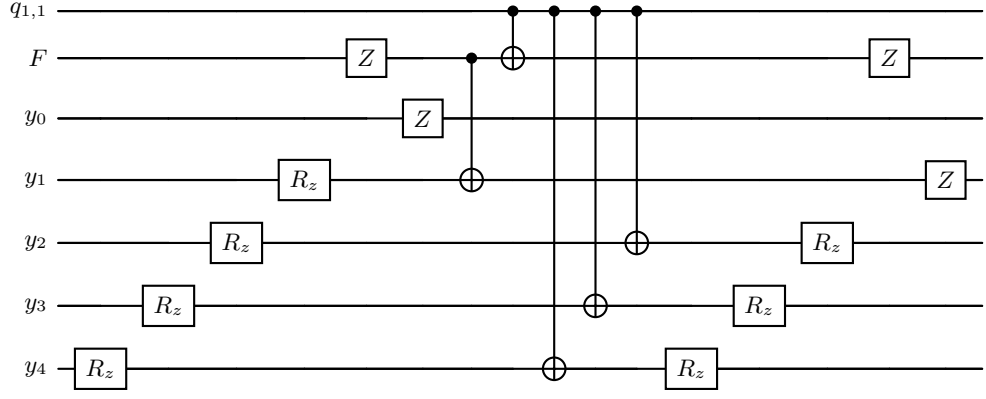
\centering
\figframes
\caption{The same two terms in the CNOT-only emission (S5). The active register
qubits hold $y_l\oplus P_t$, so each level is a single-qubit $R_z$. Between the
terms: one CNOT from $F$ unframes the level that leaves ($y_1$), the fold move is
one CNOT $q_{1,1}\to F$, and one CNOT from $q_{1,1}$ retargets each of the three
staying levels; then three $R_z$ and the $\pi$-level pair of $Z$ gates. Five
CNOTs here, against one CNOT and seven $R_{ZZ}$ (fifteen CNOTs naively) in
Fig.~\ref{fig:weight3}.}
\label{fig:frames}
\end{figure}

\begin{figure}[htbp]
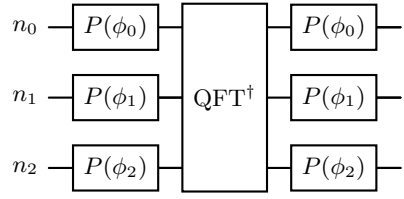
\centering
\figcentredf
\caption{The centred-grid transform $\Phi\cdot\mathrm{QFT}^\dagger\cdot\Phi$ for
$Q=3$, $\phi_j=2\pi c\,2^{\,j}/\Lambda$. The phase layers are single-qubit gates,
yet omitting them corrupts $\hat p^2$ at $O(1)$.}
\label{fig:centredf}
\end{figure}

\section{Secondary baselines, routed counts and supporting tables}
\label{app:secondary}
The tables collected here are not used for any headline claim. The NISQ
baselines of Table~\ref{tab:nisqrobust} all lie above the floor $m_{\min}$ that
Table~\ref{tab:nisqsu2} compares against; the two extra fault-tolerant baselines
of Table~\ref{tab:ftrobust} are weaker than the precision-matched
Hamming-weight-phasing incumbent of Sec.~\ref{sec:incumbents}, the first because
it gives every character its own rotation and the second because it charges every
rotation the full register width instead of its own angle.
Table~\ref{tab:routing} gives the routed counts summarised in
Sec.~\ref{sec:routing}, and Tables~\ref{tab:scalepoly} and~\ref{tab:scaleratio}
the counting polynomials and the measured-against-leading-order comparison of
Sec.~\ref{sec:scaling}.

\begin{table}[htbp]\centering\small
\setlength{\tabcolsep}{4pt}
\caption{NISQ baselines weaker than the floor, $\mathrm{SU}(2)$, $d=2$, potential term per Trotter step (two-qubit gates); the headline comparisons of Table~\ref{tab:nisqsu2} do not use these. ``Published'': Ref.~\cite{halimeh2026} compiles with CNOT ladders and $R_z$ and counts CNOTs, $(2+4/Q)\,m_{\min}$; the same grouping with a native $R_{ZZ}$ in place of the last CNOT pair gives $(1+4/Q)\,m_{\min}$ (our evaluation). GraySynth: Qiskit's \texttt{synth\_cnot\_phase\_aam}. ``Unsimpl.'': our gadget without S1--S5, i.e.\ all $w$ level rotations per term. ``CNOT naive'': every $R_{ZZ}$/CPhase replaced by two CNOTs, against the register-frame emission (S5) actually used.}
\label{tab:nisqrobust}
\begin{tabular}{rrrrrrrr}
\hline
$Q$ & $m_{\min}$ & published & published & GraySynth & unsimpl. & ours & ours\\ & (floor) & CNOT & RZZ & & RZZ & CNOT naive & CNOT frames\\
\hline
2 & 60 & 240 & 180 & 114 & 567 & 654 & 534\\
3 & 315 & 1050 & 735 & 635 & 1440 & 1638 & 1206\\
4 & 1020 & 3060 & 2040 & 2320 & 2901 & 3222 & 2250\\
5 & 2535 & 7098 & 4563 & -- & 5094 & 5610 & 3738\\
6 & 5328 & 14208 & 8880 & -- & 8163 & 8862 & 5742\\
7 & 9975 & 25650 & 15675 & -- & 12252 & 13230 & 8334\\
8 & 17160 & 42900 & 25740 & -- & 17505 & 18726 & 11586\\
9 & 27675 & 67650 & 39975 & -- & 24066 & 25650 & 15570\\
10 & 42420 & 101808 & 59388 & -- & 32079 & 33966 & 20358\\
\hline
\end{tabular}\end{table}

\begin{table}[htbp]\centering\footnotesize
\setlength{\tabcolsep}{4pt}
\caption{Fault-tolerant robustness, width and depth (the headline comparisons of Tables~\ref{tab:ftsu2} and~\ref{tab:ftmatrix} do not use the first two columns). $B_{\rm W}^{\rm match}$: one rotation per Walsh character, each priced by its own angle. $B_{\rm HWP}^{\rm flat}$: Hamming-weight phasing with every rotation charged the full $4(L-1)$ instead of $T_d$; both are weaker than the precision-matched incumbent used in the main tables, and both ratios are incumbent/ours. ``Classes'': number of equal-$|$angle$|$ classes in the Walsh expansion (largest class size next). Ancillas: unbatched Hamming-weight phasing (parity copies $+$ adder tree), the batch size that fits our ancilla budget, and our oracle's ancillas. Toffoli depth: ASAP depth of the as-built oracle counting every CCX of the Cuccaro adders, an upper bound on its reaction depth.}
\label{tab:ftrobust}
\begin{tabular}{lrrrrrrrrr}
\hline
model & $Q$ & $B_{\rm W}^{\rm match}$ & $B_{\rm HWP}^{\rm flat}$ & classes & max & HWP & HWP & our & Tof.\ depth\\ & & /ours & /ours & & class & anc. & batch & anc. & ours\\
\hline
SU(2) adj. & 2 & 2.06 & 1.08 & 6 & 18 & 34 & 18 & 52 & 636\\
SU(2) adj. & 3 & 6.83 & 2.03 & 12 & 69 & 135 & 35 & 68 & 1116\\
SU(2) adj. & 4 & 15.70 & 3.01 & 18 & 168 & 333 & 44 & 85 & 1800\\
SU(2) adj. & 5 & 29.86 & 4.02 & 24 & 339 & 673 & 53 & 103 & 2592\\
SU(2) adj. & 6 & 50.43 & 5.19 & 30 & 594 & 1184 & 63 & 121 & 3528\\
SU(2) matrix & 2 & 4.45 & 1.54 & 11 & 53 & 102 & 41 & 80 & 1168\\
SU(2) matrix & 3 & 13.41 & 2.50 & 19 & 202 & 400 & 51 & 100 & 1972\\
SU(2) matrix & 4 & 31.00 & 3.66 & 27 & 506 & 1005 & 63 & 123 & 3120\\
SU(3) matrix & 2 & 12.69 & 1.43 & 16 & 615 & 1224 & 71 & 141 & 3982\\
SU(3) matrix & 3 & 44.64 & 3.19 & 28 & 2178 & 4353 & 93 & 182 & 6870\\
SU(3) matrix & 4 & 101.74 & 5.69 & 40 & 5241 & 10475 & 113 & 223 & 10915\\
SU(4) matrix & 2 & 26.75 & 1.79 & 17 & 2910 & 5812 & 105 & 207 & 9022\\
SU(4) matrix & 3 & 80.80 & 4.26 & 31 & 10092 & 20176 & 143 & 283 & 17100\\
SU(4) matrix & 4 & 184.70 & 8.42 & 45 & 24114 & 48220 & 175 & 347 & 26742\\
SU(5) matrix & 2 & 36.70 & 2.09 & 17 & 8930 & 17854 & 150 & 296 & 19853\\
SU(5) matrix & 3 & 122.50 & 5.95 & 29 & 30700 & 61389 & 194 & 385 & 32342\\
SU(5) matrix & 4 & 300.61 & 12.07 & 41 & 73150 & 146289 & 239 & 475 & 51404\\
\hline
\end{tabular}\end{table}

\begin{table}[htbp]\centering\footnotesize
\setlength{\tabcolsep}{3pt}
\caption{Routed two-qubit counts, potential term per Trotter step, with the same router for every circuit (Qiskit SabreLayout$+$SabreSwap, three seeds, best kept; a SWAP is three two-qubit gates). Heavy-hex: distance~7 (115 qubits); square: $10\times12$ lattice. ``Ours'': native-$R_{ZZ}$ emission (RZZ) and CNOT-only emission with register frames (CNOT); ``net'': the built structure-aware parity network (a CNOT circuit). Ratios net/ours $>1$: ours cheaper after routing. On an all-to-all device (trapped ions) no routing is needed and the counts are the logical ones of Tables~\ref{tab:nisqsu2} and~\ref{tab:nisqsun}.}
\label{tab:routing}
\begin{tabular}{lrrrrrrrrrrr}
\hline
& & \multicolumn{5}{c}{heavy-hex (115 q)} & \multicolumn{5}{c}{square $10\times12$}\\ model & $Q$ & RZZ & CNOT & net & net/ & net/ & RZZ & CNOT & net & net/ & net/\\ & & & & & RZZ & CNOT & & & & RZZ & CNOT\\
\hline
SU(2) & 2 & 1598 & 2091 & 330 & 0.21 & 0.16 & 1066 & 1436 & 199 & 0.19 & 0.14\\
SU(2) & 3 & 4485 & 5750 & 1462 & 0.33 & 0.25 & 2986 & 3658 & 921 & 0.31 & 0.25\\
SU(2) & 4 & 9173 & 11762 & 4744 & 0.52 & 0.40 & 6110 & 7704 & 3158 & 0.52 & 0.41\\
SU(2) & 5 & 17215 & 20891 & 12037 & 0.70 & 0.58 & 11204 & 13625 & 7767 & 0.69 & 0.57\\
SU(2) & 6 & 28409 & 33597 & 27159 & 0.96 & 0.81 & 18172 & 21768 & 17503 & 0.96 & 0.80\\
SU(2) & 7 & 43702 & 50789 & 54172 & 1.24 & 1.07 & 28440 & 33280 & 33545 & 1.18 & 1.01\\
SU(2) & 8 & 62364 & 72583 & 95742 & 1.54 & 1.32 & 40380 & 47430 & 60421 & 1.50 & 1.27\\
SU(2) & 9 & 88710 & 99411 & 162395 & 1.83 & 1.63 & 57504 & 65853 & 101806 & 1.77 & 1.55\\
SU(2) & 10 & 119673 & 133069 & 255211 & 2.13 & 1.92 & 77017 & 87271 & 160852 & 2.09 & 1.84\\
SU(3) m. & 2 & 25388 & 32053 & 12257 & 0.48 & 0.38 & 16282 & 20765 & 8213 & 0.50 & 0.40\\
SU(3) m. & 3 & 65379 & 77569 & 55102 & 0.84 & 0.71 & 40822 & 49170 & 36450 & 0.89 & 0.74\\
SU(3) m. & 4 & 123730 & 146967 & 175815 & 1.42 & 1.20 & 78741 & 94511 & 114255 & 1.45 & 1.21\\
SU(4) m. & 2 & 78009 & 94883 & 69336 & 0.89 & 0.73 & 49626 & 60971 & 46153 & 0.93 & 0.76\\
SU(4) m. & 3 & 222576 & 260150 & 276745 & 1.24 & 1.06 & 138475 & 163267 & 182782 & 1.32 & 1.12\\
SU(5) m. & 2 & 208177 & 247827 & 247476 & 1.19 & 1.00 & 128745 & 157485 & 161231 & 1.25 & 1.02\\
\hline
\end{tabular}\end{table}

\begin{table}[htbp]\centering\small
\caption{Exact counts per matrix pair, matrix-entry coordinates (Sec.~\ref{sec:scaling}): products summed into the squared entries, $\sum_en_e=2N(N-1)(2N-1)$, and distinct quartic monomials by multiplicity type, $n_{1111}=N(N-1)(12N^2-28N+17)/2$, $n_{22}=N(N-1)(4N-3)$, $n_{112}=2N(N-1)$. Every entry is counted from the built bilinears and equals the polynomial.}
\label{tab:scalepoly}
\begin{tabular}{rrrrrr}
\hline
$N$ & $\sum_en_e$ & $n_{1111}$ & $n_{22}$ & $n_{112}$ & $n_{1111}/\sum_en_e$\\
\hline
2 & 12 & 9 & 10 & 4 & 0.75\\
3 & 60 & 123 & 54 & 12 & 2.05\\
4 & 168 & 582 & 156 & 24 & 3.46\\
5 & 360 & 1770 & 340 & 40 & 4.92\\
6 & 660 & 4215 & 630 & 60 & 6.39\\
7 & 1092 & 8589 & 1050 & 84 & 7.87\\
8 & 1680 & 15708 & 1624 & 112 & 9.35\\
9 & 2448 & 26532 & 2376 & 144 & 10.84\\
10 & 3420 & 42165 & 3330 & 180 & 12.33\\
\hline
\end{tabular}\end{table}

\begin{table}[htbp]\centering\small
\caption{Measured ratios against the leading-order predictions of Eq.~\eqref{eq:leading} (matrix-entry route, $d=2$): NISQ $m_{\min}/$ours (native $R_{ZZ}$), FT HWP$/$ours, and ablation direct$/$ours. The predictions keep only the leading term in $Q$ of each count and every factor of $N$ exactly.}
\label{tab:scaleratio}
\begin{tabular}{rrrrrrrr}
\hline
$N$ & $Q$ & NISQ & NISQ & FT & FT & ablation & ablation\\ & & meas. & lead. & meas. & lead. & meas. & lead.\\
\hline
2 & 2 & 0.23 & 0.30 & 1.13 & 0.43 & 2.88 & 4.50\\
2 & 3 & 0.51 & 0.56 & 1.83 & 1.17 & 3.37 & 4.50\\
2 & 4 & 0.86 & 0.86 & 2.82 & 2.29 & 3.72 & 4.50\\
3 & 2 & 0.50 & 0.68 & 1.17 & 1.09 & 6.58 & 12.30\\
3 & 3 & 1.07 & 1.32 & 2.79 & 3.08 & 7.88 & 12.30\\
3 & 4 & 1.82 & 2.05 & 5.20 & 6.10 & 8.71 & 12.30\\
4 & 2 & 0.84 & 1.15 & 1.63 & 1.85 & 11.72 & 20.79\\
4 & 3 & 1.59 & 1.95 & 4.05 & 5.01 & 12.69 & 20.79\\
4 & 4 & 2.71 & 3.08 & 8.16 & 10.08 & 14.18 & 20.79\\
5 & 2 & 1.05 & 1.40 & 2.00 & 2.46 & 15.45 & 29.50\\
5 & 3 & 2.25 & 2.77 & 5.82 & 7.11 & 18.58 & 29.50\\
5 & 4 & 3.82 & 4.37 & 11.91 & 14.30 & 20.72 & 29.50\\
\hline
\end{tabular}\end{table}

\section{The mini-BMN fermion layer and mass coupling}
\label{app:bmn}
This appendix carries the parts of Sec.~\ref{sec:bmn} that cannot change a
ranking: the fermion layer, which is identical in every compilation, the
potential-term counts behind Table~\ref{tab:bmnstep}, and the dependence on the
mass coupling $\kappa$.

\paragraph{Absorbing the Myers and mass terms.} With $\kappa=K_\nu/2^p$ each
square of Eq.~\eqref{eq:bmnsos} is the affine bilinear $v_e=2^pK_e-K_\nu m_e$
and $V=(\delta^4/16)\,4^{-p}V_{\rm int}$ with $V_{\rm int}$ an exact integer.
\emph{NISQ:} a linear term $l\,m_p=-l\sum_j2^{\,j}Z_{p,j}$ contributes
weight-one Walsh terms; in the Fourier addition each is one more fold-wire
parity (a prefix mask of boson $p$), visited after the product blocks with one
CNOT per move. S1--S5 are unchanged; halving applies when
$\sum c_{pq}+\sum l_p$ is even. The \emph{split} variant, valid for any real
$\kappa$, applies the gadget to $\sum_eK_e^2$ (no $\kappa$ in it) and the
structure-aware network to the remainder
$-2\kappa\sum_eK_em_e+\kappa^2\sum_em_e^2$, whose support ($6Q^3$ weight-three
and $9\binom Q2$ weight-two characters) does not depend on $\kappa$.
\emph{Fault tolerance:} each linear term is one signed, shifted addition of
$|m_p|$ per set bit of $|l_p|$ into the entry register before squaring;
Eq.~\eqref{eq:Tsos} gains $\sum_p\sum_{s\in\mathrm{bits}(|l_p|)}
T_{\rm add}(w_e-s)$ per entry and pass. The Walsh support of the full bosonic
potential now has weight-two, -three and -four characters ($m_{\min}=219$,
$1080$, $3390$ at $Q=2,3,4$, $\kappa=1$; the cubic part is exactly the $6Q^3$
characters of the six distinct Myers triples, Theorem~\ref{thm:lb}); odd weights
come only from the Myers term, whose monomials have three distinct bosons, so
there is still no weight-one character and $m_{\min}$ remains the floor of
Lemma~\ref{lem:twoq}. Verification: both NISQ emissions on $150$ random basis
states at $Q=2,3$ (leakage $\le5\times10^{-15}$, phase error
$\le6\times10^{-13}$); the fault-tolerant payload exact on $40$ random basis
states at $Q=2,3$, with the built T-count equal to the closed form at every
point; the split network by its exact integer phase identity.

\paragraph{The fermion layer.} The fermionic part is
\begin{equation}
  H_F=\sum_{k,B}X^k_BO_{kB}+\tfrac32\nu\sum_fn_f+{\rm const},\qquad
  O_{kB}=i\sqrt2\sum_{A,C}\epsilon_{ABC}\,\sigma^k_{\alpha\beta}\,
  \lambda^\dagger_{A\alpha}\lambda_{C\beta},
  \label{eq:HF}
\end{equation}
with $\lambda_f=Z_0\cdots Z_{f-1}(X_f+iY_f)/2$, $f=2A+\alpha$, and $|1\rangle$
occupied. Each $O_{kB}$ is two Hermitian hoppings on disjoint mode pairs, i.e.\
four \emph{commuting} Jordan--Wigner strings $P_s$ with coefficients
$c_s=\pm1/\sqrt2$ (checked), so
\begin{equation}
  e^{-i\theta_Fm\,O_{kB}}=\prod_s\prod_{j}e^{\,i\theta_Fc_s2^jZ_{x,j}P_s}
  \qquad(\theta_F=\tau\delta/2)
\end{equation}
exactly; the only splitting is between different $(k,B)$, which is the product
formula itself and is common to every compilation. Per hopping pair, one CNOT
and one single-qubit Clifford map both strings to $Z$-strings (derived with the
stabiliser formalism and asserted); the gadget's fold wire $F$ (returned to $0$)
carries their parity, and each character is one $R_{ZZ}(F,x_j)$ (native) or
CX--$R_z$--CX; the mass term is six $R_z$. Every group circuit is verified
against $e^{-i\theta m O_{kB}}$ exactly, for all nine groups at $Q=2,3$
($\le4.4\times10^{-16}$). Cost: $36Q+168$ two-qubit gates (native) or $72Q+168$
(CNOT). Fault-tolerantly the Yukawa and mass angles are fixed by the physics and
cannot all be put on the potential's ladder, so each is represented to $L$ bits;
within a group the $4Q$ characters carry angles $\pm\varphi2^j$, so weighted
Hamming-weight phasing~\cite{gidney2018,kivlichan2020} applies: a
column-compression tree ($3Q$ ANDs, verified exhaustively) leaves $Q+2$
rotations per group, and the six mass rotations form one class ($4$ ANDs, $3$
rotations). This gives $1876$, $3232$, $4684$, $6652$ and $8644$ T at
$Q=2,\dots,6$, against $3600$--$24\,480$ with one rotation per character. The
fermion layer is the largest common part of the fault-tolerant step, which is
where further work on this model would pay.

\paragraph{The mass coupling.} Every fractional bit of $\kappa$ widens each
entry register: at $Q=3$ the fault-tolerant affine oracle grows from $8640$ T
($\kappa=1$) to $23\,472$ ($p=7$), while the incumbent barely moves. The split
variant has no such cost: its NISQ count is the same for every $\kappa$, and
fault-tolerantly it costs $11\,580$ T at $Q=3$ (a loss against $8300$) and
$32\,304$ at $Q=5$ (a win against $47\,140$--$50\,976$); at $Q=5$ even the $p=7$
affine oracle wins ($39\,924$). A dyadic $\kappa$ is therefore worth choosing
where the grid spacing permits, and where it does not, the split variant keeps
the win at the larger $Q$ (Table~\ref{tab:kappa}).

\begin{table}[htbp]\centering\footnotesize
\setlength{\tabcolsep}{2.6pt}
\caption{NISQ, mini-BMN, $\mathrm{SU}(2)$, $\kappa=1$ (Sec.~\ref{sec:bmn}): two-qubit gates per Trotter step. $m_{\min}$: Walsh support of the full bosonic potential (quartic, Myers and mass terms), the floor of every parity network; ``network'': the built structure-aware network on it. Ours: the affine gadget (``aff.'') and the split circuit (``split'': gadget on $\sum K^2$ plus the network on the cubic and quadratic remainder), in native $R_{ZZ}$ (RZZ) and CNOT-only emissions. Fermion and kinetic layers (RZZ/CNOT) are identical in every compilation. Ratios: floor divided by the better of our two circuits, potential only and full step.}
\label{tab:bmnnisq}
\begin{tabular}{rrrrrrrrrrrrr}
\hline
$Q$ & $m_{\min}$ & net & aff. & split & aff. & split & ferm. & kin. & pot. & pot. & step & step\\ & & & RZZ & RZZ & CNOT & CNOT & & & RZZ & CNOT & RZZ & CNOT\\
\hline
2 & 219 & 294 & 1485 & 1190 & 2196 & 1676 & 240/312 & 27/54 & 0.18 & 0.13 & 0.33 & 0.29\\
3 & 1080 & 1242 & 3366 & 2931 & 4572 & 3822 & 276/384 & 81/162 & 0.37 & 0.28 & 0.44 & 0.37\\
4 & 3390 & 3632 & 6129 & 5710 & 8046 & 7222 & 312/456 & 162/324 & 0.59 & 0.47 & 0.62 & 0.52\\
5 & 8265 & 8584 & 10350 & 9943 & 13014 & 12112 & 348/528 & 270/540 & 0.83 & 0.68 & 0.84 & 0.71\\
6 & 17145 & 17654 & 15741 & 15630 & 19440 & 18744 & 384/600 & 405/810 & 1.10 & 0.91 & 1.09 & 0.92\\
7 & 31794 & -- & 23166 & 23425 & 27864 & 27448 & 420/672 & 567/1134 & 1.37 & 1.16 & 1.36 & 1.15\\
\hline
\end{tabular}\end{table}

\begin{table}[htbp]\centering\footnotesize
\setlength{\tabcolsep}{2.4pt}
\caption{Fault tolerance, mini-BMN, $\mathrm{SU}(2)$, $\kappa=1$: T gates per Trotter step. Ours: the affine sum-of-squares oracle (built, verified, equal to its closed form), width $W$. Incumbents on the full bosonic potential: one precision-matched rotation per character ($B_{\rm W}$), Hamming-weight phasing unbatched (HWP) and width-matched (HWP$_W$). Kinetic and fermion layers (the latter by weighted Hamming-weight phasing) are identical in every compilation. Ratios incumbent/ours for the potential and the full step; last column: surface-code spacetime of the full step, HWP$_W$/ours, $1000$ steps, profile \texttt{ns\_e4}.}
\label{tab:bmnft}
\begin{tabular}{rrrrrrrrrrrrrr}
\hline
$Q$ & $L$ & $W$ & ours & $B_{\rm W}$ & HWP & HWP$_W$ & kin. & ferm. & pot. & pot. & step & step & ST\\ & & & & & & & & & HWP & HWP$_W$ & HWP & HWP$_W$ & HWP$_W$\\
\hline
2 & 14 & 92 & 4716 & 7944 & 2128 & 2388 & 122 & 1876 & 0.45 & 0.51 & 0.61 & 0.65 & 0.65\\
3 & 19 & 128 & 8640 & 52056 & 8300 & 11288 & 452 & 3232 & 0.96 & 1.31 & 0.97 & 1.21 & 1.21\\
4 & 23 & 163 & 13680 & 190200 & 21388 & 33628 & 1110 & 4684 & 1.56 & 2.46 & 1.40 & 2.02 & 2.02\\
5 & 28 & 199 & 19908 & 562260 & 47140 & 81184 & 1984 & 6652 & 2.37 & 4.08 & 1.95 & 3.15 & 3.78\\
6 & 32 & 234 & 27252 & 1302480 & 89980 & 175100 & 3214 & 8644 & 3.30 & 6.43 & 2.60 & 4.78 & 5.90\\
\hline
\end{tabular}\end{table}

\begin{table}[htbp]\centering\footnotesize
\setlength{\tabcolsep}{3.5pt}
\caption{The mass coupling $\kappa=\sqrt2\nu/\delta=K_\nu/2^p$ (Sec.~\ref{sec:bmn}). The affine circuits carry $p$ fractional bits in every entry register; the split circuits do not depend on $\kappa$ (valid for any real $\kappa$; their fault-tolerant remainder is priced at the full $L$ bits). HWP: the incumbent at the same dyadic $\kappa$ (with valuation savings, which a generic real $\kappa$ would not give it).}
\label{tab:kappa}
\begin{tabular}{rrrrrrrrrrr}
\hline
$Q$ & $\kappa$ & $m_{\min}$ & aff. & split & aff. & split & $L$ & FT & FT & FT\\ & & & RZZ & RZZ & CNOT & CNOT & & aff. & split & HWP\\
\hline
3 & $1/2^{0}$ & 1080 & 3366 & 2931 & 4572 & 3822 & 19 & 8640 & 11580 & 8300\\
3 & $3/2^{0}$ & 1080 & 3366 & 2931 & 4572 & 3822 & 19 & 9072 & 11580 & 8300\\
3 & $3/2^{1}$ & 1080 & 3636 & 2931 & 5076 & 3822 & 21 & 10440 & 11580 & 8532\\
3 & $5/2^{2}$ & 1080 & 3933 & 2931 & 5634 & 3822 & 23 & 11880 & 11580 & 8764\\
3 & $23/2^{4}$ & 1080 & 4608 & 2931 & 6912 & 3822 & 27 & 16560 & 11580 & 9228\\
3 & $181/2^{7}$ & 1080 & 5823 & 2931 & 9234 & 3822 & 33 & 23472 & 11580 & 9924\\
5 & $1/2^{0}$ & 8265 & 10350 & 9943 & 13014 & 12112 & 28 & 19908 & 32304 & 47140\\
5 & $3/2^{0}$ & 8265 & 10350 & 9943 & 13014 & 12112 & 28 & 20628 & 32304 & 47140\\
5 & $3/2^{1}$ & 8265 & 10764 & 9943 & 13770 & 12112 & 30 & 22572 & 32304 & 47688\\
5 & $5/2^{2}$ & 8265 & 11205 & 9943 & 14580 & 12112 & 32 & 24588 & 32304 & 48236\\
5 & $23/2^{4}$ & 8265 & 12168 & 9943 & 16362 & 12112 & 36 & 30996 & 32304 & 49332\\
5 & $181/2^{7}$ & 8265 & 13815 & 9943 & 19440 & 12112 & 42 & 39924 & 32304 & 50976\\
\hline
\end{tabular}\end{table}

\section{The statevector validation: set-up and measurements}
\label{app:aer}
This appendix gives the set-up and the measured tables behind
Sec.~\ref{sec:aer}.

\emph{Model and state.} $\mathrm{SU}(2)$, $d=2$, $Q=2$: six bosons on $12$
system qubits, $g=1$, and $\delta=\sqrt{2\pi/\Lambda}$, for which the momentum
spacing $2\pi/(\Lambda\delta)$ equals $\delta$ and the position and momentum
grids coincide. The initial state is the product over bosons of the ground state
of $\tfrac12(p^2+\omega^2x^2)$, $\omega=2$, on the same grid, evolved under $H$
of Eq.~\eqref{eq:H} (a quench). Its energy is $\langle H\rangle=3.110$, against
the ground-state energy $E_0=2.816$ and first excited level $E_1=3.610$ of the
truncated $H$. The observables are
$\langle R^2\rangle=\sum_I\langle\mathrm{Tr}\,X_I^2\rangle$, diagonal in the
position basis, and $\langle H\rangle$. Under ED, $\langle R^2\rangle$ rises
from $2.46$ to $4.42$ at $t=1.71$ and falls back to $2.77$ at $t=3.64$.

\emph{Reference, by an independent code path.} $H$ is built from the grid
coordinates with $V$ in the Gram form~\eqref{eq:gram}, not from the bilinears
that the circuits compile. It is real symmetric
($|\mathrm{Im}\,H|\le3\times10^{-17}$) and is diagonalised densely
($4096\times4096$); the resulting evolution agrees with an independent Krylov
exponential to $3.4\times10^{-14}$ at $t=2$. The two code paths meet in one
identity, checked on all $4096$ basis states to $2.8\times10^{-14}$:
$V_{\rm Gram}(x)=g^2(\delta/2)^4V_{\rm int}(m)$, the normalisation that fixes
the circuit angle $\theta=\tau g^2(\delta/2)^4$.

\emph{Circuits.} The potential term $e^{-i\theta V_{\rm int}}$ is compiled by
(i)~the structure-aware parity network of Sec.~\ref{sec:competitor} ($13$
qubits), (ii)~our circuit in the native-$R_{ZZ}$ emission and (iii)~in the
CNOT-only emission with register frames ($18$ qubits each: fold wire and a
$5$-qubit register), each generated by the same code as in
Secs.~\ref{sec:competitor} and~\ref{sec:nisq}. The kinetic layer is the
optimised core of Sec.~\ref{sec:kinopt} (K1, K3) with its $R_{ZZ}$ angles set to
$\kappa=\tfrac\tau2\,(\pi/(\Lambda\delta))^2$ instead of the dyadic value,
conjugated in every step by the centring ramp $\Phi$. K2 is not applied, so that
every saved statevector is the physical state; K2 removes single-qubit gates
only, so the two-qubit content per step is exactly that of
Sec.~\ref{sec:kinopt}. In the CNOT-only gate set, CPhase and $R_{ZZ}$ are
written as CNOT and $R_z$ gates. The built layer equals $e^{-i\tau p^2/2}$ to
$5.7\times10^{-16}$ for one boson (dense unitary) and the classical per-boson
evolution to $2.2\times10^{-15}$ for all six bosons on a random state. Only the
angles depend on $\tau$; every executed gate comes from the built emitters, and
the circuits are transpiled at optimisation level~$0$, so the simulator executes
the emitted gate sequence.

\emph{Protocols.} (a)~\emph{Trajectory:} $\tau=1/4$, $16$ Strang steps
$V(\tau/2)\,K(\tau)\,V(\tau/2)$, with the statevector saved after every step.
The half steps of consecutive steps are applied as separate circuits, so each
saved state is exactly the product-formula state at $t=k\tau$; merging them
gives the same unitary. (b)~\emph{Convergence:} at $T=2$, the deployable circuit
$V(\tau/2)\,[K(\tau)V(\tau)]^{r-1}K(\tau)\,V(\tau/2)$ with merged half steps,
for $r=8,16,32,64$; final state only. The classical product formula (PF) is
evaluated in both protocols and, in (b), up to $r=512$.

\emph{Convergence and energy.} At $T=2$ the error in $\langle R^2\rangle$ has
local orders $1.85$, $2.06$, $2.05$ between $r=8$ and $64$ and $2.01$, $2.00$,
$2.00$ up to $r=512$, and the state distance reaches $2.00$. The infidelity is
not yet asymptotic at $r\le64$ (local orders $2.26$, $4.71$, $4.96$) and reaches
$4.06$, $4.01$, $4.00$; a single fit over $r=8$--$64$ would suggest $4.05$ and
overstate the evidence. The Aer circuits lie on the product formula at every $r$
at which they are run. A Trotter step does not conserve $\langle H\rangle$: at
$\tau=1/4$ the deviation from $3.110$ reaches $0.25$ over $t\in[0,4]$ and $1.30$
over $t\in[0,16]$; at $\tau=1/8$ and $1/16$ it is $0.17$ and $0.017$ over
$t\in[0,4]$. The step is far from $\tau\|H\|\ll1$: the truncated $H$ has
$\|H\|=108.1$, so $\tau\|H\|=27$ at $\tau=1/4$. What converges here is the
evolution of a low-energy state, not the operator, and the deviation is the
same, to round-off, for all three circuits and the classical product formula.

\begin{table}[htbp]\centering\footnotesize
\setlength{\tabcolsep}{2.2pt}
\caption{Full-Trotter dynamics on the Qiskit Aer statevector simulator, $\mathrm{SU}(2)$, $d=2$, $Q=2$ (Sec.~\ref{sec:aer}). Two-qubit gates per merged Trotter step (potential $+$ optimised kinetic layer) in the native-$R_{ZZ}$ and CNOT-only gate sets, read off the executed op lists (the floor $m_{\min}$ of the potential term is 60); two-qubit gates executed in the trajectory run (16 unmerged Strang steps, $\tau=1/4$). Largest ancilla weight after any step of any run; compilation error along the trajectory, $|R^2_{\rm Aer}-R^2_{\rm PF}|$ and $1-|\langle\psi_{\rm PF}|\psi_{\rm Aer}\rangle|^2$ against the classical product formula (PF); and the same two quantities for the deployable merged circuits of the convergence study ($r=8,16,32,64$, $T=2$). Every entry of the compilation error is at double-precision round-off, while the Trotter error of the same runs is $\sim10^{-1}$ (Table~\ref{tab:aerconv}).}
\label{tab:aer}
\begin{tabular}{lrrrrrrrrr}
\hline
circuit & qubits & 2q/step & 2q/step & 2q gates & ancilla & \multicolumn{2}{c}{trajectory} & \multicolumn{2}{c}{convergence}\\ & & RZZ set & CNOT set & executed & weight & $|\Delta R^2|$ & $1-F$ & $|\Delta R^2|$ & $1-F$\\
\hline
network & 13 & 108 & 126 & 3456 & $0$ & $1.7{\times}10^{-13}$ & $2.6{\times}10^{-14}$ & $1.2{\times}10^{-13}$ & $1.2{\times}10^{-13}$\\
ours, $R_{ZZ}$ & 18 & 390 & -- & 12192 & $2.4{\times}10^{-14}$ & $9.3{\times}10^{-14}$ & $7.8{\times}10^{-15}$ & $6.2{\times}10^{-13}$ & $1.0{\times}10^{-14}$\\
ours, CNOT & 18 & -- & 570 & 17664 & $2.0{\times}10^{-14}$ & $6.3{\times}10^{-13}$ & $1.6{\times}10^{-13}$ & $1.9{\times}10^{-12}$ & $6.0{\times}10^{-13}$\\
\hline
\end{tabular}\end{table}

\begin{table}[htbp]\centering\footnotesize
\setlength{\tabcolsep}{2.6pt}
\caption{Convergence to exact diagonalisation (ED) at $T=2$ (Sec.~\ref{sec:aer}). Left: the classical product formula (PF): $|R^2(T)-R^2_{\rm ED}(T)|$, the infidelity $1-F$ and the state distance $\|\psi-\psi_{\rm ED}\|$ (the PF carries a definite global phase), each with its local order, $\log_2$ of the ratio to the previous row. Right: the deployable merged Aer circuits (all three compilations), run for $r\le64$: the largest deviation of their $|R^2-R^2_{\rm ED}|$ from the PF value, their largest infidelity with the PF state, and the two-qubit gates executed by the native-$R_{ZZ}$ emission of our circuit and by the parity network (CNOT set).}
\label{tab:aerconv}
\begin{tabular}{rrrrrrrrrrrr}
\hline
$r$ & $\tau$ & \multicolumn{6}{c}{product formula (classical) against ED} & \multicolumn{4}{c}{Aer circuits}\\ & & $|\Delta R^2|$ & order & $1-F$ & order & $\|\Delta\psi\|$ & order & $|\Delta R^2|{-}$PF & $1-F$(PF) & 2q ours & 2q net.\\
\hline
8 & $1/4$ & $1.5{\times}10^{-1}$ & -- & $1.1{\times}10^{-2}$ & -- & $1.4{\times}10^{-1}$ & -- & $2.9{\times}10^{-13}$ & $5.4{\times}10^{-14}$ & 3492 & 1098\\
16 & $1/8$ & $4.3{\times}10^{-2}$ & 1.85 & $2.2{\times}10^{-3}$ & 2.26 & $5.2{\times}10^{-2}$ & 1.48 & $5.4{\times}10^{-13}$ & $1.2{\times}10^{-13}$ & 6612 & 2106\\
32 & $1/16$ & $1.0{\times}10^{-2}$ & 2.06 & $8.5{\times}10^{-5}$ & 4.71 & $1.1{\times}10^{-2}$ & 2.25 & $9.2{\times}10^{-13}$ & $2.2{\times}10^{-13}$ & 12852 & 4122\\
64 & $1/32$ & $2.5{\times}10^{-3}$ & 2.05 & $2.7{\times}10^{-6}$ & 4.96 & $2.2{\times}10^{-3}$ & 2.32 & $1.9{\times}10^{-12}$ & $6.0{\times}10^{-13}$ & 25332 & 8154\\
128 & $1/64$ & $6.2{\times}10^{-4}$ & 2.01 & $1.6{\times}10^{-7}$ & 4.06 & $5.4{\times}10^{-4}$ & 2.02 & -- & -- & -- & --\\
256 & $1/128$ & $1.5{\times}10^{-4}$ & 2.00 & $1.0{\times}10^{-8}$ & 4.01 & $1.3{\times}10^{-4}$ & 2.00 & -- & -- & -- & --\\
512 & $1/256$ & $3.9{\times}10^{-5}$ & 2.00 & $6.3{\times}10^{-10}$ & 4.00 & $3.3{\times}10^{-5}$ & 2.00 & -- & -- & -- & --\\
\hline
\end{tabular}\end{table}

\section{Negative results}
\label{app:negative}
Five things that were tried and did not work are recorded here, because each
closes off an approach that a reader would otherwise reasonably suggest.

\paragraph{The force-gradient corrector does not survive discretisation.}
For this model $\sum_a(\partial_aV)^2=g^4\sum_I\|\sum_J[X_J,[X_I,X_J]]\|_F^2$
(verified), which is what makes the force-gradient integrator attractive here;
but on the grid $[V,[T,V]]_{mn}=-T_{mn}(V_m-V_n)^2$ has identically zero diagonal
(verified to $10^{-11}$): it is not a multiplication operator, so it cannot be
absorbed into the potential layer. Scanning the coefficient over
$c\in\{-\tfrac1{24},-\tfrac1{48},0,\tfrac1{48},\tfrac1{24},\tfrac1{12}\}$ in the
modified-potential Strang scheme and $c\in\{0,\tfrac1{48},\tfrac1{24},\tfrac1{12}\}$
in Chin's 4A scheme yields no fourth-order point (observed orders
$2.01$--$2.21$); $c=0$ is optimal.

\paragraph{Approximate truncation of the Fourier addition does not pay.}
Dropping small-angle levels, as opposed to the exact S1, which drops only
levels that are exactly the identity, induces a phase error
$\sim2\theta K\,\delta K$ with no headroom, because the register is
subsequently squared. At $Q=4$ a $12\%$ saving costs $0.21$ in ancilla leakage;
keeping leakage below $10^{-4}$ saves under $0.4\%$. The exactness of S1 is
therefore not a stylistic choice.

\paragraph{Accumulate-then-inject does not pay.} Summing all squares (or, in the
ablation, all monomials) into one register and injecting once forces the compute
stage to be mirrored, and the ablation's ``single'' mode costs about twice its
per-monomial counterpart at every point (Table~\ref{tab:ablation}). This is why
the oracle of Sec.~\ref{sec:ft} injects per entry (F1), at the price of one
$T_{\rm add}(L)$ per entry.

\paragraph{Arithmetic without the factorisation loses fault-tolerantly, and
	arithmetic where there is nothing to factor loses too.} Both controls are in
Sec.~\ref{sec:ablation}: the direct-monomial evaluation loses to Hamming-weight
phasing at $15$ of $17$ points, and the arithmetic kinetic oracle costs
$5.4$--$9.5\times$ the phase-polynomial layer.

\paragraph{The arithmetic route loses on NISQ metrics.} The Toffoli-arithmetic circuit
with register phase polynomials (built and verified), charged with NISQ costs
(relative-phase Toffolis for AND pairs, six CNOTs otherwise), needs
$5706$--$30\,022$ two-qubit gates for $Q=2,\dots,6$, its few $R_{ZZ}$ counted
as native, which only favours it. That is $63$--$5.5\times$ the built parity
network and $15$--$6.4\times$ our Fourier circuit with native $R_{ZZ}$. The
two routes are not interchangeable, and the primitive-cost argument of
Sec.~\ref{sec:exchange} is not a heuristic.

\section{Verification protocol and code package}
\label{sec:verif}
Every resource number quoted in this paper is read off a circuit that was built
and verified; the closed forms are checked against the built circuits, not the
other way round. The complete package that produces every number, table and
figure is archived as Ref.~\cite{codepackage}; the scripts named below are the
ones in it. The
dependencies are \texttt{numpy}, \texttt{scipy}, \texttt{qiskit},
\texttt{qiskit-aer}, \texttt{matplotlib}, and \texttt{qdk} (the resource
estimator) for the surface-code scripts; the statevector runs of
Sec.~\ref{sec:aer} used \texttt{qiskit} 2.5.2 and \texttt{qiskit-aer} 0.17.2.
The driver \texttt{make\_data.py} regenerates \texttt{data/results.json}, and
\texttt{make\_tables.py}, \texttt{make\_figs.py} and \texttt{make\_circuits.py}
emit every table, figure and transcribed circuit of this paper; no number in any
table is transcribed by hand.

\begin{enumerate}
	\item \textbf{Primitives}: exhaustive over all inputs
	(\texttt{arith.py}, \texttt{test\_arith.py}).
	\item \textbf{Fault-tolerant oracle}: gradient payload $=A\,V_{\rm int}\bmod2^L$,
	system restored, every ancilla $|0\rangle$; exhaustive for $\mathrm{SU}(2)$,
	$d=2$, $Q=2$ ($4096$ states), $12$--$400$ random basis states at eight further
	points, $5$--$8$ at $N=4,5$, $Q=3,4$ (\texttt{verify\_sos\_large.py}), and
	$A\in\{1,5,11,183\}$; built T-count equal to Eq.~\eqref{eq:Tsos} at all $17$
	points (\texttt{sos\_oracle.py}).
	\item \textbf{Ablation oracles} (\texttt{ablation.py}): same checks, $5$ points
	$\times$ $2$ injection modes, $0$ failures; built counts equal to their closed
	forms.
	\item \textbf{Counting polynomials}: Eqs.~\eqref{eq:sumne} and~\eqref{eq:nmono}
	against the built bilinears and monomials at $N=2,\dots,10$, and
	Eq.~\eqref{eq:mminN} against the exact Walsh expansion at twelve points
	(\texttt{scaling.py}).
	\item \textbf{Mini-BMN} (\texttt{minibmn\_model.py}, \texttt{minibmn.py}):
	the three forms of the potential to $1.8\times10^{-12}$; the fuzzy-sphere
	vacuum and fermion spectrum; $O_{kB}$ Hermitian, number-conserving and equal to
	its one-body matrix; affine NISQ gadget and fault-tolerant oracle and split
	network, as in Sec.~\ref{sec:bmn} (\texttt{minibmn\_circuits.py}); all nine
	fermion group circuits against the exact exponential ($\le4.4\times10^{-16}$,
	\texttt{minibmn\_fermion.py}); the weighted Hamming-weight tree exhaustively for
	$Q\le4$.
	\item \textbf{Optimised kinetic layer} (\texttt{kinetic\_opt.py}): K1--K3 on dense
	unitaries for $Q=2,\dots,6$, including three-step telescoping; K4 by reversible
	simulation ($120$ basis states per batch and $80$ per CPhase class and sign, $0$
	failures), with the built T-count equal to the formula.
	\item \textbf{Surface-code inputs}: every AND/T split reproduces the T-count used
	in the tables (asserted in \texttt{surface\_code.py},
	\texttt{surface\_code\_minibmn.py}).
	\item \textbf{Physics}: the integer oracle reproduces
	$-\tfrac{g^2}{4}\sum\mathrm{Tr}[X_I,X_J]^2$ from Hermitian matrices to
	$5.6\times10^{-16}$ (\texttt{verify\_identities.py},
	\texttt{verify\_forcegrad.py}, \texttt{endtoend.py}).
	\item \textbf{Statevector spot checks}: phase-gradient primitive and a toy
	compute--inject--uncompute oracle to $6\times10^{-16}$
	(\texttt{verify\_phase.py}); the full NISQ potential at $\mathrm{SU}(2)$, $Q=2$
	(18 qubits) to $2.1\times10^{-15}$ (native) and $4.0\times10^{-15}$ (register
	frames); the parity network at $Q=2$ (13 qubits) to $3.3\times10^{-16}$.
	\item \textbf{NISQ gadget}, both emissions: exact chained simulation with frame
	tracking; leakage below $10^{-14}$ and phase error at most
	$2.1\times10^{-12}$ at the points listed in Sec.~\ref{sec:nisqcorrect}
	(\texttt{fourier2.py}, \texttt{verify\_fourier2\_large.py}).
	\item \textbf{Parity network} (\texttt{parity\_network.py},
	\texttt{run\_struct\_sun.py}): exact integer phase identity, system and target
	restored; exhaustive at $Q=2$, $2000$ ($\mathrm{SU}(2)$) or $300$ (matrix route)
	random states otherwise.
	\item \textbf{Kinetic layer} (\texttt{kinetic\_layer.py}): transform against the
	analytic kernel ($2\times10^{-15}$), whole layer against
	$e^{-2\pi iA\hat p^2_{\rm int}/2^L}$ ($1.5\times10^{-15}$), and both closed forms
	against the built gate list, for $Q=2,\dots,5$ and $L=18,26,34$.
	\item \textbf{Hamming-weight-phasing incumbent} (\texttt{hwp\_circuit.py},
	\texttt{ft\_incumbent.py}): exact streaming simulation (system restored,
	ancillas clean, exact integer phase identity) at ten configurations up to
	$\mathrm{SU}(4)$, and built T-count and width equal to the formula at all $17$
	points, unbatched and width-matched.
	\item \textbf{Routing} (\texttt{routing.py}): transpiling without a coupling map
	adds no two-qubit gate (all-to-all); heavy-hex and square-lattice counts read off
	the transpiled circuits.
	\item \textbf{Weight parity} (\texttt{models.py}): all Walsh weights even, hence
	$m_1=0$ (Proposition~\ref{prop:even}), at every model point.
	\item \textbf{Cost formulas} (\texttt{costmodel.py}, \texttt{sos\_oracle.py}):
	exact agreement with the built circuits at $17$ fault-tolerant and $21$ NISQ
	points.
	\item \textbf{Trotter constant} (\texttt{trotter.py}, \texttt{c2\_sweep.py}): the
	leading-order, interference-aware error estimate of
	Refs.~\cite{layden2022,tran2020} evaluated matrix-free against adaptive-$r$
	direct evolution reproduces the measured constant to $0.042\%$ over nine points
	($\mathrm{SU}(2)$, $d=2$, $Q=2$); the fourth-order Suzuki formula gives observed
	order $4.03$ at $r=64$. No gate-count ratio depends on this.
	\item \textbf{Full-Trotter dynamics} (\texttt{aer\_dynamics.py}). Classical
	checks first: the normalisation linking the reference potential to the compiled
	one, $V_{\rm Gram}=g^2(\delta/2)^4V_{\rm int}$, on all $4096$ basis states; dense
	exact diagonalisation against an independent Krylov exponential; the kinetic
	layer against $e^{-i\tau p^2/2}$ (one boson, dense) and against the classical
	per-boson evolution (six bosons, random state); per-step two-qubit counts read
	off the executed op lists and asserted equal to the optimised-layer formula.
	Then, on Aer, for all three compilations, after every step of the trajectory and
	at the end of every convergence run: the assertions that the ancilla weight,
	$|R^2-R^2_{\rm PF}|$ and the infidelity with the product-formula state are below
	$10^{-10}$ (measured: $\le2.4\times10^{-14}$, $\le1.9\times10^{-12}$,
	$\le6.0\times10^{-13}$), and that the asymptotic orders of the product formula
	equal $2$ and $4$ within $0.05$. The output is logged verbatim in
	\texttt{data/logs/}.
\end{enumerate}

\end{document}